\documentclass[12pt]{article}
 \usepackage{amsmath, amsthm, amssymb, bbm, setspace,bigints}
 \usepackage[margin=1 in]{geometry}
\usepackage{caption, enumitem}
\usepackage{subcaption}
\usepackage[toc,page]{appendix}
\usepackage{graphicx, amsfonts, tikz, multirow}

\usepackage{natbib}
\setcitestyle{authoryear, open={((},close={))}}
\usetikzlibrary{bayesnet}
\usepackage{pdfpages}
\usepackage{epstopdf}
\usepackage{booktabs}
\usepackage[ruled]{algorithm2e}
\usepackage[colorlinks,citecolor=blue,urlcolor=blue]{hyperref}
\usepackage[utf8]{inputenc}
\usepackage{authblk}

\def\Gauss{{ \mathrm{N} }}
\def\T{\mathrm{\scriptscriptstyle{T}}}

\def\ind{\mathbbm{1}}
\newcommand{\norm}[1]{\left\lVert#1\right\rVert}

\newtheorem{theorem}{Theorem}[section]
\newtheorem{lemma}[theorem]{Lemma}

\newtheorem{proposition}[theorem]{Proposition}

\date{}
\title{Bayesian semiparametric long memory models for discretized event data}

\author[1]{Antik Chakraborty\thanks{antik.chakraborty@duke.edu}}
\author[2]{Otso Ovaskainen\thanks{otso.ovaskainen@helsinki.fi}}
\author[1]{David B. Dunson \thanks{dunson@duke.edu}}
\affil[1]{Department of Statistical Sciences, Duke University}
\affil[2]{Research Center for Ecological Change, University of Helsinki}
\begin{document}
\maketitle
\begin{abstract}
\noindent We introduce a new class of semiparametric latent variable models for long memory discretized event data.  The proposed methodology is motivated by a study of bird vocalizations in the Amazon rain forest; the timings of vocalizations exhibit self-similarity and long range dependence ruling out models based on Poisson processes.  The proposed class of FRActional Probit (FRAP) models is based on thresholding of a latent process consisting of an additive expansion of a smooth Gaussian process with a fractional Brownian motion. We develop a Bayesian approach to inference using Markov chain Monte Carlo, and show good performance in simulation studies.  Applying the methods to the Amazon bird vocalization data, we find substantial evidence for self-similarity and non-Markovian/Poisson dynamics. To accommodate the bird vocalization data, in which there are many different species of birds exhibiting their own vocalization dynamics, a hierarchical expansion of FRAP is provided in Supplementary Materials.
\end{abstract}

\section{Introduction}

\label{sec:intro}
Event data are often obtained in a discretized form in environmental and ecological applications. Instead of recording exact times of event occurrence, one records whether or not at least one event occurred within each interval.  Such data can potentially be treated as a discrete time series \citep{tiao1976some, stern1984model} ignoring the underlying continuous time process that generated the events. While this simplification may be more amenable to standard time series analysis, it is often desirable to provide a self-explanatory stochastic model that is capable of capturing the temporal dynamics of the underlying event generating process \citep{davison1996some}. 

In \cite{davison1996some, ramesh2013multi} the authors use a Markov modulated Poisson process (MMPP) \citep{fischer1993markov} for the discretized events. Event intensities of an MMPP are directed by the states of an independently evolving continuous time Markov process whose different states correspond to different rates of events. \cite{davison1996some} derived expressions for the likelihood of the observed binary series for an MMPP using Chapman-Kolmogorov equations of a continuous time Markov chain. They proposed a maximum likelihood approach for inference on the model parameters, which include the instantaneous transition rate matrix of the continuous time Markov chain and the Poisson rates corresponding to each state of the chain. They also show that the autocorrelation function of the binary time series generated by an MMPP exhibits a geometric decay. \cite{fearnhead2006exact} proposed a Gibbs sampling algorithm for Bayesian inference.

%The geometric decay of autocorrelations limits the applicability of MMPPs to time series with long range dependence \citep{pipiras2017long}. 
The geometric rate of decay in autocorrelations of an MMPP makes it inapplicable to model time series with slower decay in autocorrelations. This is true for time series where the dependence structure is non-Markovian; a special class of time series that has non-Markovian dependence and is a focus in this article is known as long range dependent series.
Roughly speaking, a time series is long range dependent if its autocovariance function decays like a power function. Long range dependence has been encountered in time series data from a large variety of fields including hydrology \citep{hurst1951long}, finance \citep{lo1989long}, network traffic \citep{willinger2003long}, and climatology \citep{franzke2020structure} among others.
A natural extension of the MMPP to accommodate long range dependence is the fractional Poisson process \citep{laskin2003fractional}. However, likelihood computation of discretized data obtained from a fractional Poisson process is not straightforward.

In seminal work, \cite{mandelbrot1968fractional} introduced fractional Brownian motion, a generalization of standard Brownian motion, and showed that the increments of this process are stationary and exhibit long range dependence. The general definition of fractional Brownian motion is a stochastic integral with respect to a standard Brownian motion where the order of integration is defined by a parameter $H \in (0,1)$.  \cite{mandelbrot1968fractional} referred to $H$ as the Hurst parameter after the hydrologist Harold Hurst who discovered long range dependence in time series while studying storage capacities of dams on the Nile river.
\cite{mandelbrot1968fractional} also established that the fractional Brownian motion is a self-similar stochastic process with no characteristic time-scale \citep{graves2014brief}. Intuitively, self-similar processes retain statistical properties over different time scales and when their increments are stationary, they exhibit long range dependence. 

For discretized events, the intensity of the latent counting process determines the correlation structure of the binary time series. If the binary series is long range dependent, then an inhomogeneous Poisson process with fixed intensity $\lambda(t)$ is insufficient to explain the observed data, as it implies that increments in disjoint time intervals are independent. Furthermore, \cite[Chapter 2]{beran2016long} showed that a doubly stochastic Poisson process with random intensity $\lambda(t)$ is long range dependent if and only if $\lambda(t)$ is long range dependent. Refer to \cite{samorodnitsky2007long, pipiras2017long} for reviews on long range dependence and self-similarity.

In this article, we propose a latent semiparametric framework to model long range dependent discretized event data via a 
FRActional Probit (FRAP) model. The FRAP model assumes a latent stochastic process responsible for generating the events of interest. Positive values of the process within a time interval imply one or more event occurrences within that interval. By setting the latent process as the fractional Brownian motion parameterized by the Hurst coefficient, we show the FRAP model is able to capture long range dependence of the discretized events. By varying the Hurst coefficient within $(0,1)$, the spectrum of the model encompasses anti-persistence when $H \in (0, 1/2)$, independence for $H = 1/2$ and long range dependence when $H \in (1/2, 1)$. Moreover, we also include a nonparametric trend component in our model to account for  non-stationarity of event occurrences. The proposed framework accommodates testing of long range dependence in the data by comparing
$H_0: H = 0.5$ versus $H_1: H > 0.5$. 
We define a Bayesian approach to inference using a Gaussian process prior for the nonparametric trend.  A Markov chain Monte Carlo (MCMC) sampling algorithm is proposed relying on sampling the latent process.  
%We also extend the model and MCMC algorithm to account for multiple different event types via a grade-of-membership (GoM) model.

The rest of the article is organized as follows.
In Section \ref{sec:primary_analysis} we introduce the motivating Amazon bird vocalization data, including exploratory analyses revealing possible long range dependence. Section \ref{sec:FRAP} is dedicated to the development and analysis of the FRAP model. 
%We apply the FRAP model to the bird vocalization data and compare the fit to the MMPP model  \citep{davison1996some}. 
Section \ref{sec:frac_prob_simulation} contains simulation experiments evaluating the proposed approach, and Section \ref{sec:amazon_frap} analyzes the Amazon data.
In %Section \ref{sec:shared_trait} 
the Supplementary Materials, 
we extend the FRAP model to allow multiple types of events through a grade-of-membership model and provide details on prior specification and posterior computation.

\section{Amazon bird vocalization data}\label{sec:primary_analysis}
Bird songs play a major role in mate selection and thus have a pronounced impact on their population dynamics \citep{slabbekoorn2002bird}. Identifying birds based on their vocalizations is a widely used method for estimating bird population sizes and following population trends over time, and automated acoustic monitoring is increasingly used in both ecological studies and in conservation \citep{laiolo2010emerging}. Bird songs are well known to follow a circadian pattern in that they sing most intensely early in the morning and late in the day \citep{krebs1983dawn}. 

We are motivated by an Amazon bird vocalization data set containing observations from the years 2010 to 2014. Audio monitoring devices were placed at different locations throughout the Amazon rain forest. Using the methods of \cite{ovaskainen2018animal}, these recordings were converted to discretized binary time series \citep{de2019spatio} containing 0-1 indicators of which species vocalized at least once in one minute time intervals for a 180 minute period staring at sunrise. A visual depiction of the binary sequence of vocalizations for the bird species Automolus ochrolaemus is provided in Figure \ref{fig:audio_recording_data}.
Based on the audio recordings, it is not possible to reliably distinguish different individual birds of the same species or to infer the number of birds vocalizing.  We focus on three locations which are similar in habitat and close in latitude and longitude. Our data consist of recordings for 15 relatively common bird species. For each species we have about 5 to 10 days of recordings during the months of June to September with recordings starting typically around 5:15 AM. On average, a given species vocalized in 25-30 out of the 180 intervals. 

\begin{figure}
   \centering
    \includegraphics[height = 6cm, width = 8cm]{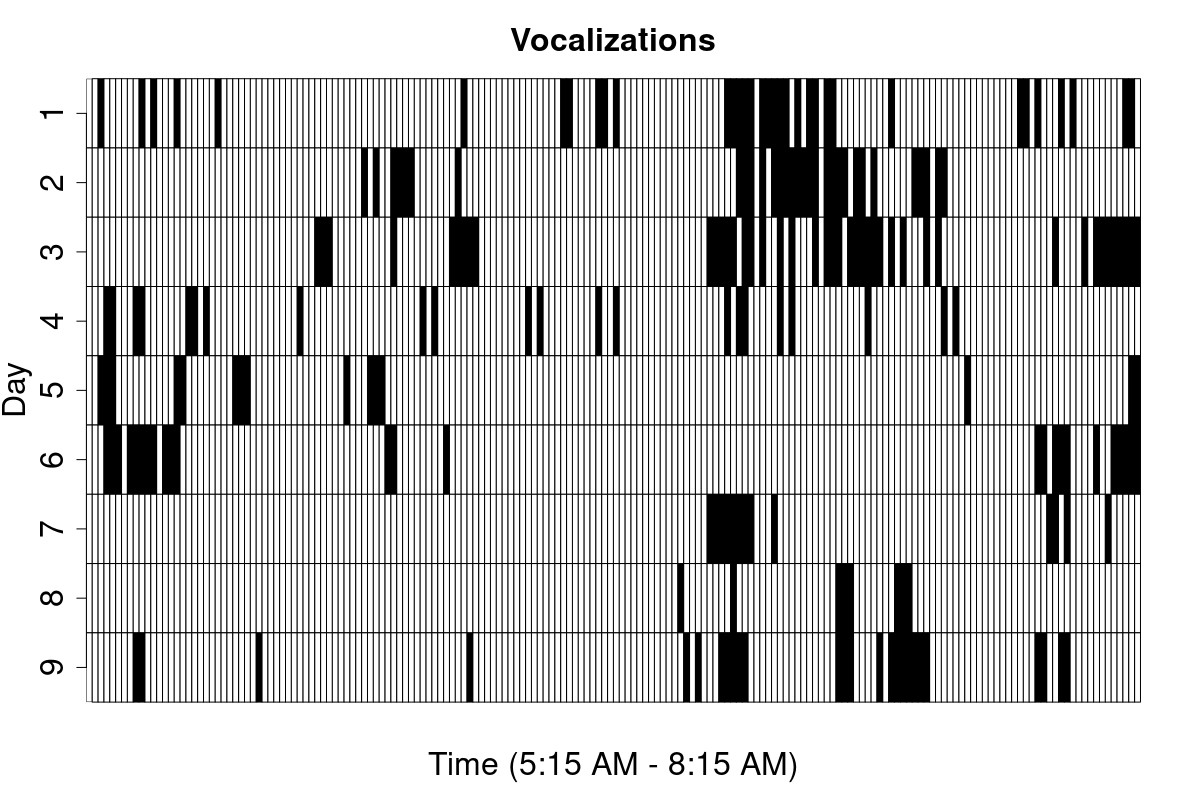}
    \caption{Binary sequence of all vocalizations of birds from the Automolus ochrolaemus species, during 9 days of recording. White and black grids represent absence or presence of vocalizations, respectively.}
    \label{fig:audio_recording_data}
\end{figure}

%\begin{figure}
% \begin{subfigure}{0.5\textwidth}
%\includegraphics[width=0.95\linewidth, height=6cm]{voc_plot.png} 
%\caption{Marginal probabilities}
%\label{fig:subim1}
%\end{subfigure}
%\begin{subfigure}{0.48\textwidth}
%\includegraphics[width=0.85\linewidth, height=5.5cm]{acf_plot.png}
%\caption{Conditional probabilities}
%\label{fig:subim2}
%\end{subfigure}
% \caption{Marginal and conditional probabilities of bird vocalizations for 15 different species at different time scales $\Delta t = \{1,2,4, 9, 15, 30, 60, 90\}$.}
%\label{fig:audio_recording_data}
%\end{figure}

Our analysis focuses on two characteristics of the bird vocalization dynamics.
First, we are interested in the distribution of duration of bird song activity and inactivity - in particular, our results indicate that the duration cannot be adequately modeled by the exponential distribution. In the context of event data, exponential inter-event times are routinely assumed for mathematical and computational simplicity. However, many naturally occurring events, such as earthquakes \citep{ogata1991some}, landscape evolution \citep{weymer2018statistical}, and human brain activity \citep{tagliazucchi2013breakdown}, have been shown not to follow such patterns. We are also interested in identifying time periods when birds are more likely to sing and recovering groups of bird species that have similar singing patterns.

\begin{figure}
    \centering
    \includegraphics[height = 6cm, width = 0.9\textwidth]{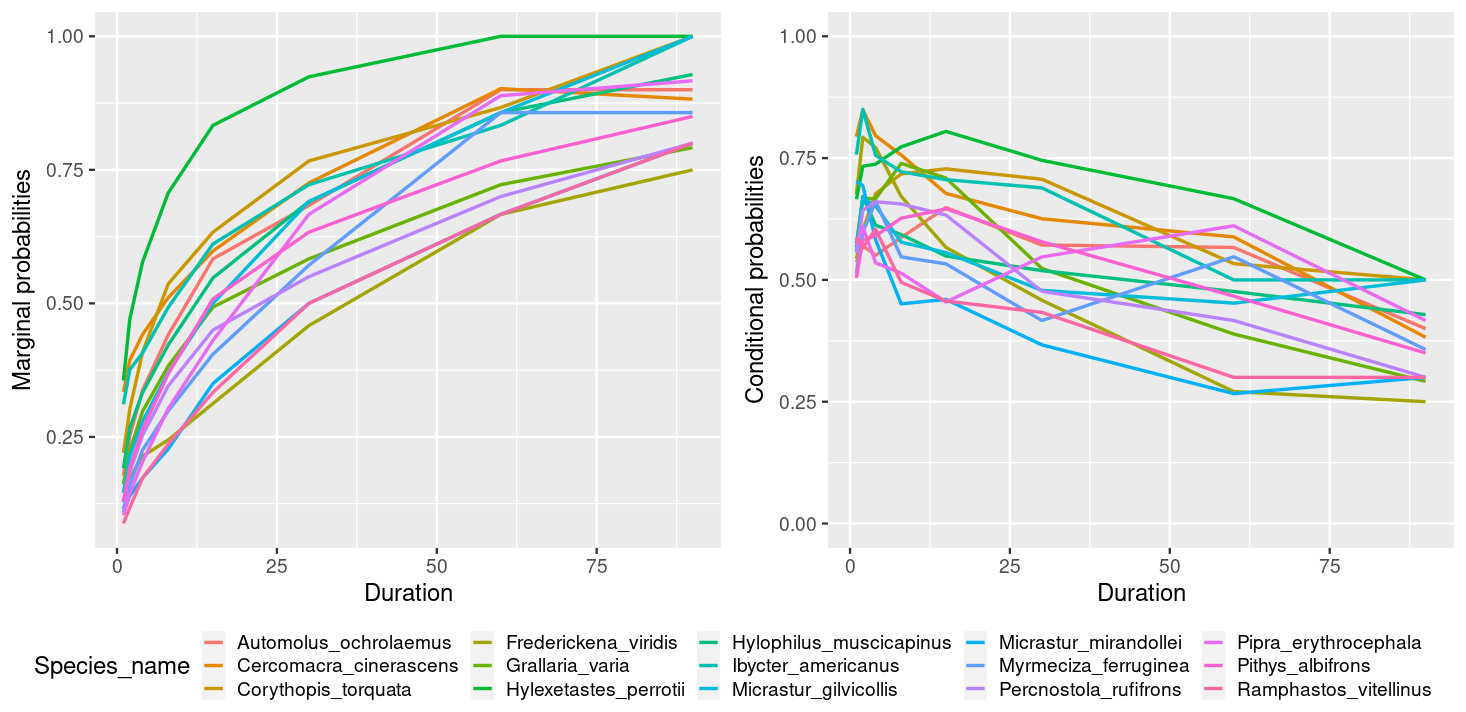}
    \caption{Marginal (left panel) and conditional (right panel) probabilities of bird vocalizations for 15 different species at different time scales $\Delta t = \{1,2,4, 9, 15, 30, 60, 90\}$.}
    \label{fig:marginal_and_conditional_probabilites}
\end{figure}

Define the marginal probability of vocalization for a given time interval of length $\Delta t$ to be the probability of observing at least one vocalization when a time interval of this length is selected at random. In the left panel of Figure \ref{fig:marginal_and_conditional_probabilites} we show the marginal probabilities of a vocalization during minute intervals of length $\Delta t = \{1,2,4, 9, 15, 30, 60, 90\}$ for 15 different bird species. On the right panel of Figure \ref{fig:marginal_and_conditional_probabilites}, we show the probabilities of vocalizations conditioned on the event that the bird vocalized in the previous interval of the same length. Quite naturally the marginal probabilities show an increasing pattern with the length of intervals. However, there is very little variation in the conditional probabilities with changes in $\Delta t$. Such scaling of summary statistics is commonly encountered in self-similar stochastic processes \citep{pipiras2017long}. Additionally, the distance autocorrelations \citep{zhou2012measuring} and the periodogram of the binary series for one day of recording for the species Corythopis torquata is displayed in Figure \ref{fig:dacf_periodogram}. The distance autocorrelation is a popular alternative to the standard autocorrelation function for investigating non-linear dependence structures and thus is more suitable for the binary time series data presented here. The slow decay in the distance autocorrelation and the spikes in the spectrum for small frequencies indicate potential long range dependence in the data.

\begin{figure}
    \centering
    \includegraphics[height = 6cm, width = 1\textwidth]{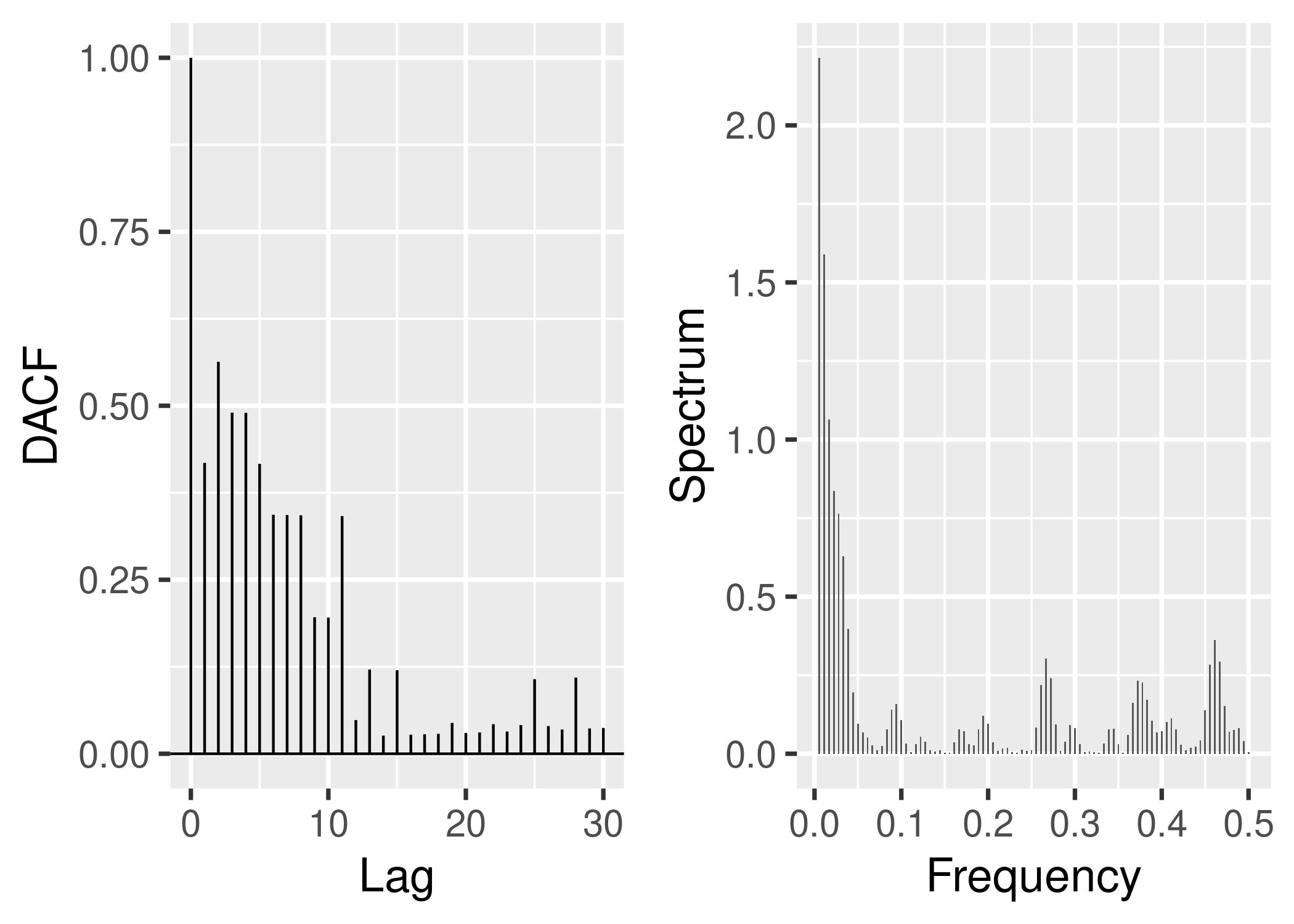}
    \caption{Distance autocorrelation (left panel) at different lags for the binary indicators of vocalizations for the species Corythopis torquata. On the right panel, the periodogram for the same time series is shown.}
    \label{fig:dacf_periodogram}
\end{figure}

We will use the notation $X(t)$ for the stochastic process $\{X_t\}_{t \in \mathbb{R}}$. A stochastic process $X(t)$ is said to be self-similar if for any $c>0$ we have $X(ct) \overset{d}{=}c^H X(t)$, so that the random variables $X(t)$ and $X(ct)$ are equivalent in distribution up to scaling factors governed by the parameter $H$. This parameter $H \in (0,1)$ is commonly known as the Hurst exponent. A self-similar process with stationary increments has non-summable autocovariances \citep{pipiras2017long} and is known as a long range dependent (LRD) time series.  In such series, the degree of long range dependence is controlled by $H$. For continuous time series data, many methods have been proposed to estimate $H$: the ReScaled range (RS) analysis \citep{hurst1951long, mandelbrot1969some}, detrended fluctuation analysis \citep{peng1994mosaic}, log periodogram regression \citep{geweke1983estimation}, local Whittle approximation \citep{robinson1995gaussian} etc. Although these methods typically apply to continuous data, we use these estimators in our exploratory analyses.
%Given a time series $\{X_t, t \in \mathcal{T}\}$ many methods have been proposed to estimate $H$ mostly focusing on continuous data. To our knowledge, currently there is no literature on estimating the Hurst exponent for long memory binary time series. Nonetheless, we use ReScaled range (RS) analysis \citep{hurst1951long, mandelbrot1969some} and detrended fluctuation analysis (DTA) \citep{peng1994mosaic} in our exploratory analysis.

The rescaled range statistic of a time series $\{X_t, t\in \mathcal{T}\}$ is the ratio of the range of cumulative deviations from the mean to the standard deviation.  Estimates of the RS statistic are obtained by dividing the time series into sub-series of different lengths and computing the RS statistic for each scale \citep{bassingthwaighte1994evaluating}. For a self-similar stationary time series with Hurst coefficient $H$, \cite{mandelbrot1969some} showed that the RS statistic varies roughly as $\delta^H$ where $\delta$ refers to the time scale. 
%The Hurst coefficient measures the intrinsic dimension of a $D$- dimensional object via the relation $H = E + 1 - D$, where $E$ stands for the Euclidean dimension of the object. In our current context $E = 1$ and a low value of $H$ represents a time series that fills space like a two-dimensional surface, whereas a high value of $H$ means the time series only intrudes into the two-dimensional space slightly \citep{bassingthwaighte1994evaluating}.
Non-stationarity in the time series may lead to false detection of long-range dependence \citep{kantelhardt2001detecting}, motivating DFA, which removes the trend in a first stage before estimating $H$. \cite{geweke1983estimation} proposed an alternative semiparametric approach to estimate the Hurst exponent $H$ based on the log periodogram of the data. \cite{robinson1995gaussian} instead estimate $H$ using the local Whittle approximation to the likelihood.

We use \texttt{R} packages \texttt{pracma} and \texttt{fractal} to estimate the Hurst exponent using RS analysis and DFA, respectively; \texttt{fractal} allows for polynomial trends. To estimate $H$ according to \cite{geweke1983estimation} and \cite{robinson1995gaussian} we use the $\texttt{LongMemoryTS}$ package in $\texttt{R}$.
Table \ref{tab:hurst_data} shows the estimates of the Hurst exponent from the RS ($\hat{H}_{\mathrm{RS}}$) analysis and DFA ($\hat{H}_{\mathrm{DFA}}$) with a linear trend for the 15 bird species from Figure \ref{fig:marginal_and_conditional_probabilites} along with the estimates of $H$ according to \cite{geweke1983estimation} ($\hat{H}_{\mathrm{GPH}}$) and \cite{robinson1995gaussian} ($\hat{H}_{\mathrm{W}}$). The posterior mean estimate $\hat{H}_{\mathrm{FRAP}}$ of $H$ obtained from the proposed model here is also included in Table \ref{tab:hurst_data} together with the lower ($\hat{H}_{\mathrm{LR}}$) and upper ($\hat{H}_{\mathrm{UR}}$) 95\% credible intervals.
For DFA, estimates of the Hurst exponent were not sensitive to the choice of the degree of the polynomial. There is a considerable discrepancy in the estimates of $H$ from the RS and DFA analyses, with the DFA analysis surprisingly suggesting anti-persistence in the data.  However, it is not clear how to obtain confidence intervals for these estimates, and whether they are entirely appropriate given that RS and DFA methods were developed for continuous and not binary time series.
The estimates of $H$ as seen from $\hat{H}_{\mathrm{GPH}}$ and $\hat{H}_{\mathrm{W}}$ in Table \ref{tab:hurst_data} are closer to the estimates obtained from the proposed model although they often do not satisfy the constraint $0<H<1$.
\begin{table}[!ht]
\centering
\scalebox{0.8}{
\begin{tabular}{cccccccccc}
  \hline
 & Species name & $\hat{H}_{\mathrm{RS}}$ & $\hat{H}_{\mathrm{DFA}}$ & $\hat{H}_{\mathrm{GPH}}$ & $\hat{H}_{\mathrm{W}}$& $\hat{H}_{\mathrm{QMLE}}$ & $\hat{H}_{\mathrm{LR}}$ & $\hat{H}_{\mathrm{FRAP}}$ & $\hat{H}_{\mathrm{UR}}$ \\ 
  \hline
1 & Automolus ochrolaemus & 0.67 & 0.19 & 0.83 & 0.78 & 0.67 &0.85 & 0.89 & 0.95 \\ 
  2 & Cercomacra cinerascens & 0.77 & 0.17 & 1.14 & 1.07 & 0.83 & 0.90 & 0.92 & 0.94 \\ 
  3 & Corythopis torquata & 0.70 & 0.18 & 0.98 & 0.90 & 0.72 & 0.80 & 0.84 & 0.88 \\ 
  4 & Frederickena viridis & 0.75 & 0.13 &1.12 & 1.05 & 0.63 & 0.79 & 0.86 & 0.93 \\ 
  5 & Grallaria varia & 0.74 & 0.14 & 1.01 & 0.94 & 0.66 & 0.84 & 0.89 & 0.93 \\ 
  6 & Hylexetastes perrotii & 0.70 & 0.21 &1.01 & 0.93 & 0.85 & 0.89 & 0.93 & 0.95 \\ 
  7 & Hylophilus muscicapinus & 0.71 & 0.15 & 1.08 & 0.94 & 0.69 & 0.80 & 0.87 & 0.93 \\ 
  8 & Ibycter americanus & 0.74 & 0.25 & 1.19 & 1.11 & 0.81 & 0.90 & 0.94 & 0.96 \\ 
  9 & Micrastur gilvicollis & 0.70 & 0.17 &0.98 & 0.91 & 0.64 & 0.80 & 0.84 & 0.89 \\ 
  10 & Micrastur mirandollei & 0.72 & 0.34 &1.00 & 0.94 & 0.61 & 0.83 & 0.88 & 0.93 \\ 
  11 & Myrmeciza ferruginea & 0.70 & 0.14 & 1.00 & 0.88 & 0.61 & 0.81 & 0.85 & 0.88 \\ 
  12 & Percnostola rufifrons & 0.70 & 0.12 & 1.06 & 0.95 & 0.63 & 0.87 & 0.92 & 0.96 \\ 
  13 & Pipra erythrocephala & 0.69 & 0.20 & 0.90 & 0.85 & 0.60 & 0.79 & 0.84 & 0.88 \\ 
  14 & Pithys albifrons & 0.70 & 0.17 & 0.96 & 0.84 & 0.63 & 0.83 & 0.87 & 0.90 \\ 
  15 & Ramphastos vitellinus & 0.70 & 0.16 & 1.08 & 0.97 & 0.59 & 0.77 & 0.83 & 0.89 \\ 
   \hline
\end{tabular}
}
\caption{Estimated Hurst exponents for the 15 bird species using the RS analysis ($\hat{H}_{\mbox{RS}}$), DFA with a linear trend ($\hat{H}_{\mathrm{DFA}}$), \cite{geweke1983estimation} ($\hat{H}_{\mathrm{GPH}}$), \cite{robinson1995gaussian} ($\hat{H}_{\mathrm{W}}$), \cite{livsey2018multivariate} ($\hat{H}_{\mathrm{QMLE}}$) and the FRAP model. For the FRAP model we include the posterior mean ($\hat{H}_{\mathrm{FRAP}}$) along with the 95\% credible intervals ($\hat{H}_{\mathrm{LR}}$, $\hat{H}_{\mathrm{UR}}$).}
\label{tab:hurst_data}
\end{table}

Time series models for discrete valued data with LRD structure are relatively sparse. Classical approaches for count/discrete valued times series, such as the integer autoregressive moving-average \citep{ mckenzie1985some, mckenzie1986autoregressive, mckenzie1988some} and discrete autoregressive moving-average \citep{jacobs1978discrete, jacobs1978asyptotic}, cannot account for LRD \citep[Chapter 21]{davis2016handbook}. \cite{cui2009new} developed a model for stationary Bernoulli sequences with LRD based on renewal sequences. \cite{livsey2018multivariate} provide a recipe for multivariate count time series with Poisson marginals and a flexible autocovariance structure that can adequately handle LRD; see also \cite{jia2018latent}. Estimates of the Hurst exponent obtained from the quasi-maximum likelihood method from \cite{livsey2018multivariate} is also included in Table \ref{tab:hurst_data} under the column $\hat{H}_{\mathrm{QMLE}}$. A major difference of the proposed method from the aforementioned works is in its ability to include non-stationarity in the data. 

Our goal is not simply to estimate the Hurst coefficient; we would like to define a realistic generative probability model for these data that takes into account the data collection process and can be used as a useful baseline for future ecological analyses that include spatial dependence, environmental covariates and other complications.  The estimated Hurst coefficients for our proposed fractional probit model, see Section \ref{sec:frac_prob} below, are provided in Table \ref{tab:hurst_data}.  Interestingly, the Hurst coefficients are significantly above 0.5 for all fifteen bird species.  This suggests long range dependence, a new finding of ecological interest, which should be considered in future analyses of animal occurrence time series and conflicts with usual Poisson process-based models.

\section{Discretized event data}\label{sec:FRAP}
We begin this section by defining some notation. Suppose event recordings are discretized at time points $\{t_0, t_1, \ldots, t_n\}$ where the time points belong to some index set $\mathcal{T}$. In this article we assume that $t_{i+1} - t_i = \Delta$ for all $i=0,1,\ldots,n-1$. Corresponding to each time interval, we have the following binary event indicators
\begin{equation}
\label{eq:interval_definition}
   Z(t_{i-1}, t_{i}) = \begin{cases} 1 \mbox{\quad if at least one event occurred in } (t_{i-1}, t_{i}]\\
   0 \mbox{\quad otherwise}
   \end{cases}
\end{equation}
We consider $R$ replications of this binary time series $\mathbf{Z}=\{Z^{(1)}, Z^{(2)}, \ldots, Z^{(R)}\}$. In our particular setting, the replications correspond to different days of recording at a fixed location and for a fixed bird species.

\subsection{Fractional probit model}\label{sec:frac_prob}
Consider for now a single replication of the binary series $Z$. We assume a latent continuous time process $y(t), t \in \mathcal{T}$, is responsible for instigating events of interest. Let $\rho_0(y(s), y(t))$ denote the covariance function of $y(\cdot)$ for $s, t \in \mathcal{T}$. We want to derive 
 a discrete time series from $y(t)$ so that it reflects the autocovariance structure of the observed binary data. Of particular interest are time series that exhibit long range dependence motivated by the bird vocalization data. A time series $\{X_t, t \in \mathbb{Z}\}$ is said to have long range dependence if its autocovariance function $\rho_X(k)$ at lag $k \in \mathbb{Z}$ decays polynomially as $k \to \infty$, 
\begin{equation}\label{eq:LRD_def}
    \rho_X(k) = L(k)k^{2d-1}, \quad \mbox{for } d\in (0,1/2),
\end{equation}
where $L(\cdot)$ is a slowly varying function at infinity, meaning it is positive on $[c,\infty)$ with $c\geq 0$ and for any $a>0$, $\lim_{u\to \infty}L(au)/L(u) = 1$. The parameter $d$ is called the long-range dependence parameter and the series is said to have {\it{long memory}}. A popular alternative characterization of long range dependent series relies on properties in the frequency domain. If $s_X(\lambda)$ is the spectral density of the times series $\{X_t, t \in \mathbb{Z}\}$, then the series is long range dependent if 
\begin{equation}\label{eq:LRD_def_freq}
    s_X(\lambda) = L^*(\lambda) \lambda^{-2d}, \quad \mbox{for } d \in (0,1/2) \,\, \mbox{and } 0<\lambda\leq \pi,
\end{equation}
for some slowly varying function $L^*(\cdot)$.
This definition implies that spectral densities of long range dependent series have an infinite spike in a neighborhood around 0.   

The concept of long memory is intricately related to self-similarity of processes. Broadly speaking, self-similar processes are obtained as normalized limits of partial sum processes of a long memory series \citep{pipiras2017long}. While there are several well studied self-similar processes, one of the most fundamental and perhaps the most popular is the fractional Brownian motion (fBM). A standard Brownian motion $B(t)$ is a stationary Gaussian process with covariance function $K_B(s,t) = \min(s,t)$. The fBM generalizes this covariance structure to the form
\begin{equation}\label{eq:fBM_covariance}
    K_H(s,t) = \frac{\sigma^2}{2}(|t|^{2H} + |s|^{2H} - |t-s|^{2H}), \quad H \in (0,1).
\end{equation}
The parameter $H$ is known as the Hurst exponent of the fBM.  Henceforth, we shall write $B_H(t)$ to denote an fBM with Hurst exponent $H$. In \eqref{eq:fBM_covariance}, $\sigma^2 = \mathbb{E}\{B_H(1)\}^2$ which we set to 1. For $H = 0.5$ the standard Brownian motion is recovered. The self-similarity of the process stems from the fact that $B_H(ct)\overset{d}{=}c^H B_H(t)$. Setting $\epsilon_i^H = B_H(i) - B_H(i-1), i \in \mathbbm{Z}$, we obtain a stationary discrete time series known as fractional Gaussian noise (fGN), elements of which marginally follow a standard Gaussian distribution. The autocovariance function $\rho_\epsilon(k), k = 0,1,2, \ldots$ of $\{\epsilon_n^H\}$ is
\begin{equation}\label{eq:fGN_covariance}
    \rho_\epsilon(k) = \frac{1}{2}(|k+1|^{2H} - 2|k|^{2H} + |k-1|^{2H}) \sim H(2H-1)k^{2H -2} \mbox{ as } k \to \infty,
\end{equation}
where for two sequences $a_n$ and $b_n$, $a_n \sim b_n$ implies that $a_n/b_n = 1$ as $n \to \infty$.
Hence, for $H \in (1/2, 1)$ the series is LRD in the sense of equation \eqref{eq:LRD_def} with LRD parameter $d = H - 1/2$. Our proposed model relies heavily on the simple observation that if we define a series $Z_i^* = \mathbbm{I}(\epsilon_i^H>0)$, where $\epsilon_i^H$ is a fGN with Hurst exponent $H$, then the autocovariance function of this binary series $Z_i^*$ is
\begin{equation}
    \rho_{Z^*}(k) = \frac{1}{2\pi}\arcsin{\rho_\epsilon(k)}.
    \label{eq:binary_acf}
\end{equation}
When the series $\{\epsilon_i^H\}$ is long range dependent, that is $H \in (1/2,1)$, it follows that for large lags $k$, $\rho_{Z^*}(k) \approx \rho_\epsilon(k)$ since $\sin{x} \approx x$ for small $x$, i.e. the series $Z_i^*$ is also long range dependent with Hurst coefficient $H$. More generally, the inheritance of the LRD property in the binary series is true for any underlying LRD Gaussian series since \eqref{eq:binary_acf} holds for any binary series obtained by discretizing a stationary Gaussian series. We refer the reader to \cite[Lemma 4.1]{livsey2018multivariate} for the general proof.  
In the context of discretized event data as described in \eqref{eq:interval_definition}, we then have the following latent formulation,
\begin{equation}\label{eq:frac_prob_indicators}
    Z(t_{i-1}, t_{i}) = \begin{cases} 1 \mbox{ if }  \epsilon_{i}^H = B_H(t_{i}) - B_H(t_{i-1}) >0\\
   0 \mbox{\quad otherwise,}
   \end{cases}
\end{equation}
for $i = 0, 1, \ldots$; see also \cite[Equation (4.4)]{livsey2018multivariate} for an equivalent formulation for any latent Gaussian series. The above formulation accounts for long memory in the observed binary series, with the autocorrelation decay mimicking that of an fGN. Moreover, as a consequence of the scaling property of an fBM, a scale free property of conditional probabilities consistent with Figure \ref{fig:marginal_and_conditional_probabilites} is 
established in the following Lemma.
\begin{lemma}\label{lemma:conditional_scaling}
Let $B_H(t)$ be an fBM with Hurst coefficient $H$ with $\sigma^2 = 1$. Suppose we observe $B_H(t)$ at $i \in \mathbb{N}$ and let $X_{i} \equiv B_H(i), \, i \geq 1$, $X_0 \equiv B_H(0)$. Define the binary series of indicators at time scale $m$, $Z^{(m)}_i = \mathbbm{1} \left\lbrace X_{{i2^m}} - X_{{(i-1)2^m}} > 0\right\rbrace, \, i\geq 1$ so that for $m=0$, the series $Z^{(0)}_1, Z^{(0)}_2, \ldots$ is as in \eqref{eq:frac_prob_indicators}. Then for any $m =0,1,\ldots$, the conditional probability $\mathbb{P}(Z_{i+1}^{(m)} = 1 | Z_{i}^{(m)} = 1)$ is independent of the time scale $m$. In particular,
\begin{equation}\label{eq:scale_free_cond_probs}
P(Z_{i+1}^{(m)} = 1 | Z_i^{(m)} = 1) = \frac{1}{2} + \frac{1}{\pi} \arcsin{(2^{2H - 1})}
\end{equation} 
%Let $B_H(t)$ be an fBM with Hurst coefficient $H$ defined on the positive real line $[0,\infty)$ where we assume $B_H(0) = 0$ a.s. Suppose we observe $B_H(t)$ at $\{t_0, t_1, \ldots\}$ where $t_0 = 0$ and $t_i - t_{i-1} = 1, \, i \geq 1$. Define the series of indicators $Z_n^{(l)} = \mathbbm{I}\left[B_H(2^l\, t_i) - B_H\{2^l(t_{i-1}-1) \} > 0\right]$ for $l = 0, 1,2, \ldots$, so that for $l = 0$ we recover the binary series defined in \eqref{eq:frac_prob_indicators}.
%{\bf DD - you're now using t as an index instead of actually time - this needs to be fixed and also make clearer the definition of Z below}

\end{lemma}
\begin{proof}
See Appendix \ref{sec:proof_of_lemma}.
\end{proof}
Two remarks are in order. First, for the special case $H = 0.5$, the conditional probability in equation \eqref{eq:scale_free_cond_probs} becomes $1/2$, so that when the series of indicators are generated from an underlying white noise series, the conditional probability of $Z_{i+1} = 1| Z_i = 1$ and the marginal probability of $Z_i = 1$ are equal. Second, since the function $\arcsin{(\cdot)}$ is increasing, the conditional probability of $Z_{i+1} = 1 | Z_i = 1$ increases with $H$,  covering the cases of anti-persistence $H<0.5$, independence $H = 0.5$ and LRD for $H> 0.5$. Figure \ref{fig:cond_probs_versus_H} depicts the relationship between the Hurst coefficient $H$ and the conditional probabilities.

\begin{figure}[!ht]
    \centering
    \includegraphics[width = 0.7\textwidth]{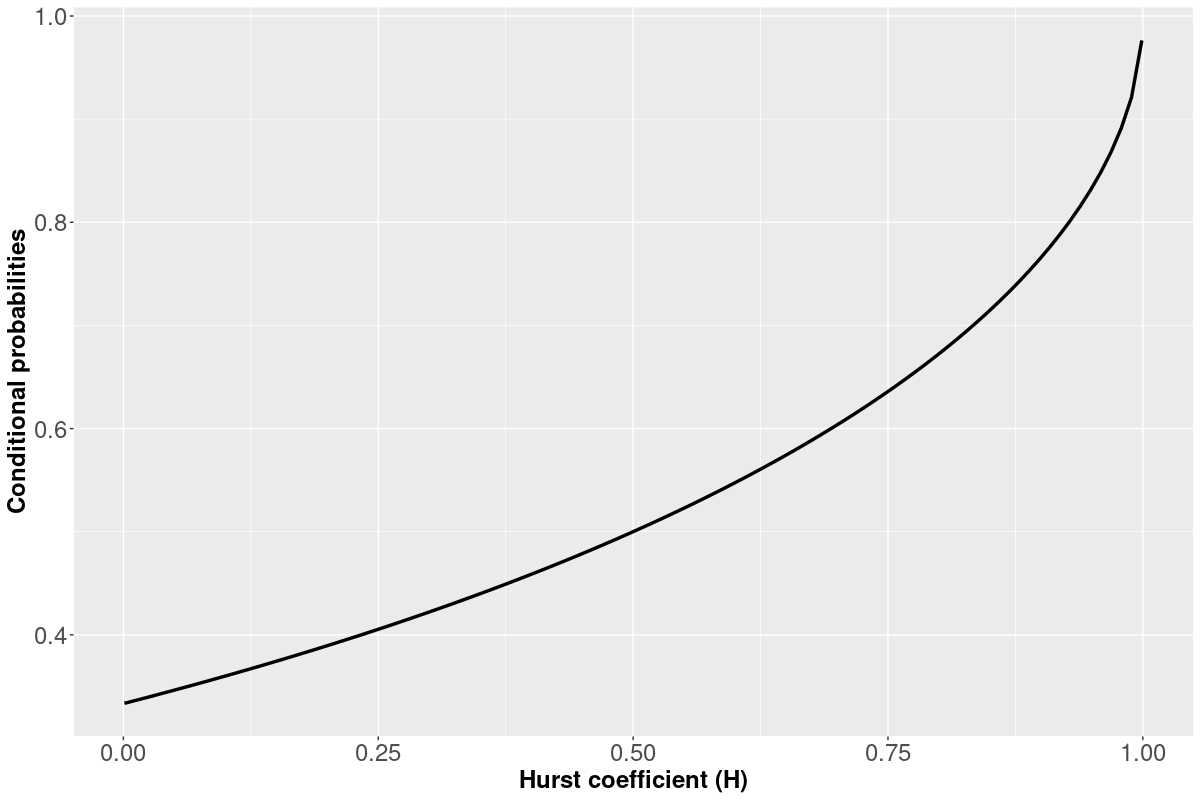}
    \caption{Relation between the Hurst coefficient $H$ and the conditional probabilities obtained from equation \eqref{eq:scale_free_cond_probs}.}
    %\label{fig:my_label}
    \label{fig:cond_probs_versus_H}
\end{figure}
Additionally, the spectral density of the series $Z_n$ can be shown to have a pole at zero frequency when $H>1/2$, a distinctive feature of LRD series. Let $s_{Z}(\lambda)$ and $s_\epsilon(\lambda)$ denote the spectral density of the series $Z_n$ and $\epsilon_n$, respectively, for $-\pi\leq \lambda \leq \pi$. Then we have for $H>1/2$,
\begin{align*}\label{eq:frap_spectral_density}
    s_Z(\lambda) = \sum_{k=-\infty}^\infty \rho_z(k) \exp(ik\lambda) &= \sum_{k=-\infty}^\infty \frac{1}{2\pi} \arcsin \rho_\epsilon (k)\exp(ik\lambda) \\
    & \geq \sum_{k=-\infty}^\infty \frac{1}{2\pi} \rho_\epsilon (k)\exp(ik\lambda) = \frac{1}{2\pi}s_\epsilon(\lambda),
\end{align*}
where we have used the Jordan inequality $\arcsin x - x \geq 0$ for $0<x<1$ \citep{mitrinovic1970analytic}. Combining this with the fact that $s_\epsilon(\lambda) \sim \lambda^{1-2H}$ in a neighborhood of $0$, we see $s_Z(\lambda)$ also has a pole at $\lambda = 0$ for $H > 1/2$ and hence is LRD according to definition \eqref{eq:LRD_def_freq}. 

When considering the Amazon bird vocalization data and other real data applications, a clear limitation of model (\ref{eq:frac_prob_indicators}) is the restriction of the marginal probabilities being fixed at 0.5. To be realistic, we need to allow the marginal probabilities to be arbitrary and varying smoothly according to the time of the day. Moreover, \cite{mikosch2004nonstationarities, chen2010localized} among many others noted that purely from a modeling perspective, long memory behavior in the sample autocovariances can be sufficiently explained by non-stationarity.
%It is well known, particularly in the finance literature, that   
%apparent long memory type behaviour can often be sufficiently explained by introducing suitable non-stationarity \citep{mikosch2004nonstationarities, chen2010localized}. 

With this motivation, we introduce a non-stationary component in the FRAP model by assuming that the latent process driving the events, say $y(t)$, admits an additive decomposition of the form $y(t) = f(t) + B_H(t)$ while letting 
\begin{equation}
\label{eq:trend_and_noise}
    Z(t_{i-1}, t_i) = \begin{cases} 1 \mbox{ if } y(t_{i}) - y(t_{i-1}) = f(t_i) - f(t_{i-1})+ \epsilon_i^H>0\\
   0 \mbox{\quad otherwise,}
   \end{cases}
\end{equation}
where we assume $f(\cdot)$ is continuously differentiable.
The marginal probability of observing an event in interval $(t_{i-1}, t_i]$ is then $P[Z(t_{i-1}, t_i) = 1] = P[f(t_i) - f(t_{i-1}) + \epsilon_i^H > 0] = \Phi\{f(t_i) - f(t_{i-1})\}$, where $\Phi(\cdot)$ is the cumulative density function of a standard Gaussian random variable. Hence, the variation in $f(\cdot)$ during $(t_{i-1}, t_i]$ determines the probability of observing an event during this time; a positive change increases the marginal probability, whereas a negative change decreases it. If $f(t_i) - f(t_{i-1}) = 0$, then the marginal probability is $P(\epsilon_i^H > 0) = 1/2$. 
To simplify notation, we write $Z_i = Z(t_{i-1}, t_{i})$. The vector $\mathbf{\epsilon}^H = (\epsilon_1^H,\ldots, \epsilon_{n}^H)$ follows an $n$-dimensional Gaussian distribution with mean $0$ and covariance matrix $\Sigma_H$ whose $(i,j)$-th element is $\Sigma_H(i,j) = \rho_\epsilon(|i-j|)$ defined in equation \eqref{eq:fGN_covariance}. We will also include a precision parameter $\tau^2$ so that $\mathbf{\epsilon}^H \sim \mathrm{N}(0, \tau^2\Sigma_H)$. The marginal probability of an event occurrence in the interval $(t_{i-1}, t_i]$ then becomes $P[Z(t_{i-1}, t_i) = 1] = \Phi[\{f(t_i) - f(t_{i-1})\}/\tau] $. In Figure \ref{fig:example_marginal_probability} we show the variations in marginal probabilities when the non-stationary component $f(t)$ in model \eqref{eq:trend_and_noise} is set to $f(t) = \sin (4\pi t)/90$ with $\tau = 1$.
\begin{figure}
    \centering
    \includegraphics[width = 0.8\textwidth, height = 6cm]{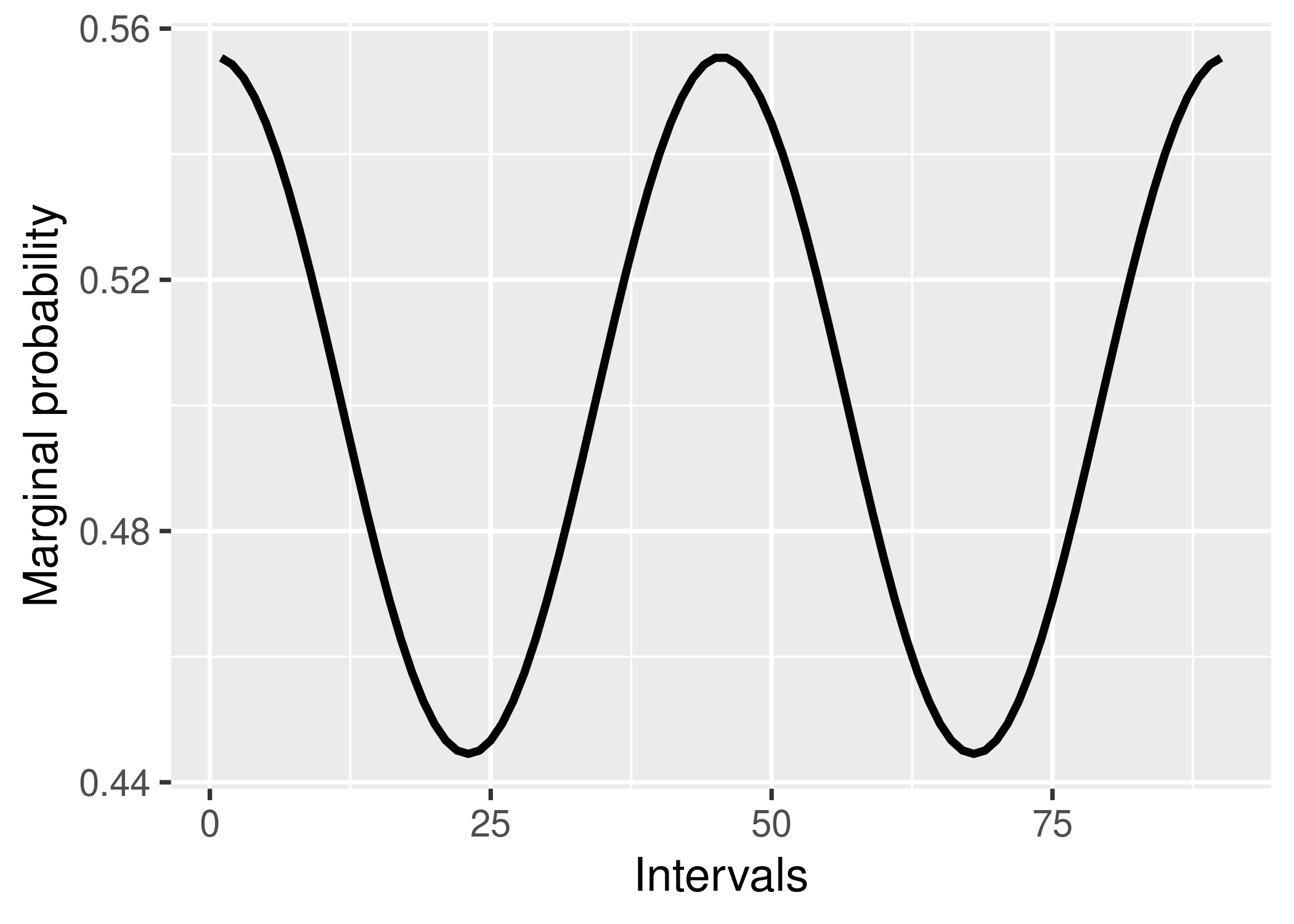}
    \caption{Variation in marginal probabilities of observing a vocalization or an event when $f(t) = \sin (4\pi t)/90$ for time intervals $(0,1], (1,2], \ldots, (89, 90]$. Here $\tau = 1$ and the marginal probabilities are calculated as $\Phi\{f(i+1) - f(i)\}$ for $i = 0, \ldots, 89$.}
    \label{fig:example_marginal_probability}
\end{figure}

Akin to probit models for longitudinal binary data with covariate information \citep{chib1998analysis} we are interested in modeling the likelihood of the observed events $Z=(Z_1, \ldots, Z_{n})\in \{0,1\}^n$.  However, in our context we have time series data with smooth trend $f(t)$ and temporal dependence captured through  $\mathbf{\epsilon}^H$. Letting $\mathbf{f} = \{f(t_0), \ldots, f(t_n)\}$ and putting the pieces together, we get the following probit-type model, \begin{equation}
\label{eq:fractional_probit}
 P(Z \in E\mid \mathbf{f}, H) = P(W \in E_W \mid \mathbf{f}, H), \quad W \sim \mathrm{N}(A\mathbf{f}, \tau^2\Sigma_H), \quad E \subset \{0,1\}^n,
\end{equation}
where $E_W$ is the intersection of half-planes $E_W = \cap_{i:Z_i = 1}(W_i>0) \cap_{i:Z_i=0}(W_i\leq 0)$ and the matrix $A \in \Re^{n\times n}$ is such that $A_{ii} = 1, A_{i,i-1}= -1$ and $A_{ij} = 0$ for $j \neq i, i-1$. For identifiability, we impose the restriction that $f(0) = 0$. Then under model \eqref{eq:fractional_probit} $f(\cdot)/\tau$ is identifiable. To accommodate this restriction, we let $A_{11} =1, A_{1,j} = 0, j = 2,\ldots, n$; the other rows of $A$ remain unchanged. 

Model \eqref{eq:fractional_probit} is quite flexible in incorporating a smooth trend $f(t)$ and auto-correlated errors.  In the special case in which $H=0.5$, the error term becomes uncorrelated so that $f(t)$ is assumed to characterize the pattern over time in the data.  When $H>0.5$ in contrast, we obtain long range dependence.  The model provides a useful basis for testing of long range dependence via comparing $H_0: H=0.5$ to $H_1: H>0.5$, in the presence of potential non-stationarity.   
%Our proposed test accounts for an unknown smooth trend over time and hence avoids the tendency to conclude LRD when in fact the data are simply non-stationary.

%It is well known that for a fBM on $\mathcal{T}$ with Hurst coefficient $H$ that $B_H(t) \in \mathcal{M}_H$ almost surely where $\mathcal{M}_H$ denote the set H\"older continuous function of order at most $2H - 1$ on $\mathcal{T}$. Set $\mathcal{M}_H^0 \in \mathcal{M}_H$ to be the subset of functions that are monotonically increasing. Following \cite{lin2014bayesian} such functions can be generated via a projection onto $\mathcal{M}_H^0$ maintaining continuity. Additionally, restrict $f(t)$ to be monotonically increasing. Since class of monotonic function is closed under addition, the process $y(t)$  
%Indeed, by setting the noise process $B_H(t)$ to zero we recover a fixed rate Poisson process whose rate $\Lambda(t) = \exp\{f(t)\}$ with the restriction that $\Lambda(t)$ is monotonically increasing. By further restricting $\Lambda(t)$  to be in the class of functions which are piecewise constant a MMPP is recovered exactly. On the other hand, the general class of doubly stochastic Poisson process can be recovered from model \eqref{eq:trend_and_noise} when $H = 0.5$. In the following subsections we shall see that the proposed formulation renders simpler prior and posterior formulation compared to the samplers available for MMPP namely \cite{fearnhead2006exact}.

\subsection{Priors and posterior computation}\label{sec:priors}
Without loss of generality, we assume that the time points $\{t_0, \ldots, t_n\}\in \mathcal{T} = [0,T]$. Let $\Theta = \{(f,\beta, \tau): f \in \mathcal{F}, \, \beta \in \Re, \,\tau \in \Re^+\}$ be the parameter space in model \eqref{eq:fractional_probit}, where we let $\mathcal{F}$ be the space of continuously differentiable functions on $\mathcal{T}$ and $\beta = \log\{H/1-H\}$. Let $\Pi_\beta$ denote the prior on $\beta$ and $\Pi_\tau$ denote the prior on $\tau$. We choose $\Pi_\beta \equiv \mathrm{N}(0,1)$ and $\Pi_\tau \equiv$ Inverse-Gamma$(a_\tau, b_\tau)$ for positive constants $a_\tau, b_\tau$. For the nonparametric component, we let $f \sim \Pi_f$, where $\Pi_f$ is an appropriate prior for an unknown smooth function.  In particular, we choose a zero mean 
Gaussian process (GP) with a squared exponential covariance kernel \citep{Rasmussen:2005:GPM:1162254} scaled by the precision parameter $\tau^2$ of the latent process defined as 
\begin{equation}\label{eq:squared_exponential_kernel}
    %C(s,t) = \frac{2^{1-\nu}}{\Gamma(\nu)}\left(\frac{\sqrt{2\nu}r}{\phi}\right)^\nu K^M_\nu\left(\frac{\sqrt{2\nu}r}{\phi}\right),
    C(s, t) = \tau^2\sigma^2\exp\left\lbrace -\dfrac{(s-t)^2}{2\phi^2}\right\rbrace, \quad \sigma, \phi > 0,
\end{equation}
for $s, t \in \mathcal{T}$. For numerical stability we follow the standard practice of adding a small positive quantity $\nu$ to the diagonal elements of the GP covariance matrix so that $C(s, t) = \tau^2\sigma^2\exp -\{(s-t)^2/ 2\phi^2\} + \nu \mathbbm{1}(s=t)$. 
%The prior specification on $f$ is completed by scaling $C(s, t)$ with the precision parameter $\tau$ of the latent variables. That is, we set $f \sim GP(0, \tau^2C)$.  
Consequently, the induced prior on $\mathbf{g} = A\mathbf{f}$ is again a multivariate Gaussian distribution with covariance matrix $C_\mathbf{g} = \tau^2 ACA^\prime$, where $C$ is an $n \times n$ matrix with $C_{ij} = C(t_i, t_j)$.
To learn the hyperparameters $(\sigma, \phi)$ from the data, we transform them to the logarithmic scale and augment the parameter space $\Theta$ to $\Theta_* = \Theta \times \eta$ where $\eta = \{(\log \sigma, \log \phi): \sigma, \phi>0\}$. We place independent standard Gaussian priors on each component of $\eta$. Thus $\Pi_\eta \equiv \mathrm{N}(0,1)\times \mathrm{N}(0,1)$. The prior specification is completed by setting $\Pi =  \Pi_f \times \Pi_\beta \times \Pi_\tau \times \Pi_\eta$.

For the Amazon bird vocalization data, we have replications $\{Z^{(1)}, \ldots, Z^{(R)}\}$ of $Z$ over different days which have minimal empirical correlations. We assume these replicates are conditionally independent involving the same $f(t)$ but with different realizations of the latent residual term leading to different realizations $W^{(r)}$, for $r=1,\ldots,R$, of $W$ in equation \eqref{eq:fractional_probit}.  Including also the priors, this leads to the following hierarchy: 
\begin{align}\label{eq:frac_prob_with_priors}
    P(Z^{(r)} \in E_r\mid & \, \mathbf{f}, \beta, \tau)  = P(W^{(r)} \in E_{W^{(r)}} \mid \mathbf{f}, \rho), \, \beta = \log\{H/(1-H)\},\\ \nonumber
    & W^{(r)}\mid \rho, \mathbf{f}, \eta \sim \mathrm{N}(A\mathbf{f}, \tau^2\Sigma_H),\\\nonumber
    & \mathbf{f}\mid \tau^2, \eta \, \sim \Pi_f,\\\nonumber
    & \beta \sim \Pi_\beta, \, \, \tau \sim \Pi_\tau, \, \, \eta \sim \Pi_\eta.
\end{align}
for any $E_r \subset \{0,1\}^n$ and $E_{W^{(r)}}$ as defined after equation \eqref{eq:fractional_probit}.

Posterior computation under the hierarchical FRAP model \eqref{eq:frac_prob_with_priors} is potentially challenging.  We initially considered an integrated nested Laplace approximation (INLA), which was developed for approximate Bayesian inference in latent Gaussian models by \cite{rue2009approximate}. However, the non-Markovian structure of the FRAP model renders the INLA paradigm non-applicable \citep{rue2005gaussian}. In a recent article, \cite{sorbye2018fractional} applied the INLA framework to a fGN model where the authors approximate the fGN by a mixture of first-order autoregressive processes. This approximation technique works quite well when the observed time series is quite long $n \sim 500$ and the number of replications available is also very high $R \sim 1000$. For the Amazon bird vocalization data both the length and the replications are quite small compare to these numbers. 
%For a crude approximation of the covariance structure of fGN using Markov random fields, see \cite{sorbye2018fractional}. 
%In a recent article \cite{durante2019conjugate} showed regression coefficients in a probit model with Gaussian priors follow a skew-Normal distribution \citep{arellano2006unification}. The result is particularly useful in regression settings where the the number of variables $p \gg n$, the sample size wherein one only needs to simulate a $n$-dimensional truncated Gaussian variable. 

\begin{algorithm}[!ht]

 \KwData{$\mathbf{Z}=\{Z_1, \ldots, Z_R\}$, $L$ = number of MCMC samples}
 
 \KwResult{$L$ posterior samples from $\Pi(\Theta_* \mid \mathbf{Z})$ : $ \left\lbrace\Hat{\Theta}_*^{(l)}\right\rbrace_{l=1}^L$}
 Initialize $\beta = 0$, $\mathbf{g}= A\mathbf{f} = 0$, $\log \sigma = 0$ and $\log \phi = 0$;
 
 \For{$l = 1:L$}{
  \begin{itemize}
      \item \textbf{Update} $W_r\mid - \overset{ind}{\sim}\Gauss(\mathbf{g}, \tau^2\Sigma_H)\ind_{B_{W_r}},\,\, r= 1,\ldots R$
      \item \textbf{Update} $\mathbf{g} \mid - \sim \Gauss\left\lbrace\frac{R}{\tau^2}\Phi^{-1}\Sigma_H^{-1}\overline{W}, \Phi^{-1}\right\rbrace$, where $\overline{W} = \frac{1}{R}\sum_{r=1}^R W_r$ and $\Phi = \frac{R}{\tau^2}\Sigma_H^{-1} + \frac{1}{\tau^2} C_\mathbf{g}^{-1}$.
      \item \textbf{Update} $\beta \mid -$ using a Metropolis random walk with proposal density \linebreak $\Gauss(\beta_{l-1}, s_1^2)$.
      \item \textbf{Update} $\tau \mid -  \sim \text{Inverse-Gamma}(\frac{nR}{2} + a_\tau, S + b_\tau)$, where $S = \frac{1}{2}\left[\mathrm{trace}\{(W - G)^\prime \Sigma_H^{-1} (W - G)\} + \mathbf{g}^\prime C_g^{-1} \mathbf{g}\right]$ and $G$ is a $n \times R$ matrix  \newline
      with all columns equal to $\mathbf{g}$.
      \item \textbf{Update} $\eta \mid -$ jointly via a Metropolis random walk 
      with proposal density $\Gauss(\eta^{(1)}_{l-1}, s_2^2)\times \Gauss( \eta^{(2)}_{l-1}, s_2^2)$, where $\eta^{(1)} = \log \sigma$ and $\eta^{(2)} = \log \phi$.
  \end{itemize}
 }
 \caption{MCMC algorithm to draw samples from the posterior under model \eqref{eq:frac_prob_with_priors}.}
 \label{algo:MCMC_one_species}
\end{algorithm}

%\vspace{0.1in}

We instead focus on Markov chain Monte Carlo (MCMC), developing a practical algorithm that exploits the structure of the model, as detailed in
Algorithm \ref{algo:MCMC_one_species}. We use $\theta \mid -$ to denote the full conditional distribution of a parameter $\theta$ given other parameters and the data in Algorithm \ref{algo:MCMC_one_species}. The Metropolis random walk steps to update the Hurst exponent and the Gaussian process kernel hyperparameters are implemented following the adaptive Metropolis algorithm \citep{roberts2001optimal}.  Adaptive Metropolis modifies the classical version of the algorithm by varying the covariance of the noise in the random walk targeting the optimal acceptance rate \citep{roberts2001optimal}. Suppose $s_1$ and $s_2$ are the noise variance of the random walk updates of $\beta$ and $\eta$, respectively.  We start with $s_1 = 0.1$ and $s_2 = 0.2$ and update them at MCMC iteration $l$ by increasing or decreasing by a factor of $\exp(l^{-0.5})$ whenever $l$ is divisible by 50. Adaptation targets an acceptance probability of $\sim 0.3$. Values of $f(\cdot)$ at a set of test points can also be evaluated by accommodating a further step in Algorithm \ref{algo:MCMC_one_species} following \cite[equations 2.22-2.24]{Rasmussen:2005:GPM:1162254}.
%We use $\theta \mid -$ to denote the full conditional distribution of a parameter $\theta$ given other parameters and the data. In the Metropolis step for updating $\beta$ in Algorithm \ref{algo:MCMC_one_species}, the acceptance probability is the minimum of one and the ratio: 
%\begin{align}\label{eq:MH_for_rho}
%    \dfrac{\left\lbrace\prod_{r=1}^R L(W_r\mid \mathbf{g}, \tau^2\Sigma_H)\right\rbrace\Pi_\beta(\beta_{l-1})}{\left\lbrace\prod_{r=1}^R L(W_r\mid \mathbf{g},\tau^2\Sigma_{H_{l-1}})\right\rbrace\Pi_\beta(\beta)},
%\end{align}
%where $l-1$ indexes parameter values at the previous iteration, 
%$H = 1/(1+\exp(-\beta))$, $H_{l-1} = 1/(1+\exp(-\beta_{l-1}))$, $\Sigma_H$, $\Sigma_{H_{l-1}}$ are the fGN covariance matrices computed at $H$ and $H_{l-1}$ respectively, and $L(\cdot\mid \mathbf{\mu}, \Omega)$ is the multivariate Gaussian likelihood with mean $\mu$ and covariance matrix $\Omega$. Updating of $(\log \sigma, \log \phi)$ proceeds similarly with proposal variance $s_2^2$. To achieve optimum mixing for random walk Metropolis, we adopt the adaptive algorithm of \cite{roberts2001optimal}. We start with $s_1 = s_2 = 0.5$ and update them at MCMC iteration $l$ by increasing or decreasing by a factor of $\exp(l^{-0.5})$ whenever $l$ is divisible by 50. 
%Adaptation targets an acceptance probability of 
%$\sim 0.3$.

The main computational bottleneck of Algorithm \ref{algo:MCMC_one_species} involves simulating the truncated Gaussian random variables for updating the latent variables $W_r$. This is done using \texttt{R} package \texttt{tmvtnorm}. Unfortunately, we found the popular circulant embedding algorithm  \citep[Chapter 2.11]{pipiras2017long} to simulate Gaussian long range dependent sequences to be quite slow when these constraints are imposed. 
%which uses an algorithm proposed in \cite{geweke1996bayesian}. 
%More recent algorithms for simulating truncated Gaussian random variables with linear constraints such as \cite{botev2017normal, pakman2014exact} resulted in unstable samplers in our numerical experiments.
To accelerate computation, the $R$ copies of the latent variables are generated in parallel.
%using \texttt{R} package \texttt{doParallel}.
The \texttt{R} code to implement the FRAP model given $R$ copies of discretized events is available \href{https://github.com/antik015/Fractional-Probit-Model}{here}.

\subsection{Asymptotics}
%Hoping to prove the semiparametric BvM so that testing can be done. Will complete as proof comes along.
Here we consider infill asymptotics so we assume we can make measurements at finer time points $\{t_0, \ldots, t_n\}$ as $n \to \infty$ within the interval $[0,T]$. We assume the noise variance $\tau = 1$. Also, we set the number of replications $R =1$ since the proof does not depend on a specific value of $R$. Let the true trend function be $f_0 \in \mathcal{F}$ and the true Hurst coefficient be $H_0$  satisfying $0<a<H_0<b<1$ for some $a, b \in (0,1)$. Define $\theta_0 = (f_0, H_0)$ and $P_0$ to be the true data generating probability measure and consider any weak neighborhood $U$ of $\theta_0$. By showing that the joint prior $\Pi \equiv \Pi_\beta \times \Pi_f$ has positive Kullback-Leibler support we have the following consistency result.
\begin{theorem}\label{theorem:weak_consistency}
Suppose $f_0 \in \mathcal{F}$ and $0<a<H_0<b<1$ for some $a, b \in (0,1)$. Write $\theta_0 = (f_0, H_0)$ and consider any weak neighborhood $U$ of $\theta_0$. Then the posterior probability of the set $U^c$ given the series of indicators $\Pi(U^c \mid Z_1, \ldots, Z_n) \to 0$ in $P_0-$probability as $n \to \infty$.
\end{theorem}

A proof of Theorem \ref{theorem:weak_consistency} is provided in the Appendix. 

\section{Simulation experiments} \label{sec:frac_prob_simulation}
We report the results of a detailed simulation study for different choices of the latent trend function $f(\cdot)$ in equation \eqref{eq:trend_and_noise}\textcolor{red}{,} while varying the number of replications $R$. We assume discretized observations are available for a period of $n = 90$ time units and the number of replications $R$ considered is $\{10, 25, 50\}$. The following choices of the trend function are considered:
\begin{enumerate}
    \item $f_1(t) = \sin \frac{4\pi t}{90}$
    \item $f_2(t) = 5[1 + \exp\{-2.5(t - 45)/15\}]^{-1}$
    \item $f_3(t) = -2\{(t-45)/45\}^2 + 2$
    \item $f_4(t) = -1.2 \{(t-45)/45\} + 0.5 \cos(\frac{3\pi t}{90}) - 1.7$
    \item $f_5(t) = 0.1 f_1(t)\log \{f_2(t)\}$.
\end{enumerate}
We note here that $f_2(\cdot)$ slightly violates the assumption that the non-stationary component in model \eqref{eq:trend_and_noise} at $t = 0$ is 0. We define the squared empirical $\ell_2$ norm of a function $g(\cdot)$ evaluated on the points $\{t_1, \ldots, t_n\}$ as $\norm{g}_{2,n} = n^{-1} \sum_{i=1}^n \{g(t_i)\}^2$. Given an estimator $\tilde{f}(\cdot)$ of $\underline{f}(\cdot) := f(\cdot)/\tau$ in model \eqref{eq:frac_prob_with_priors}, we evaluate the performance of Algorithm \ref{algo:MCMC_one_species} by computing the {\it{Re}}lative Mean Square Error (ReMSE) defined as $\mbox{ReMSE} = \norm{e}_{2,n}/\norm{\underline{f}}_{2,n}$, where $e(\cdot) = \underline{f}(\cdot) - \tilde{f}(\cdot)$.  The latent trends $f(\cdot)$ are chosen from the aforementioned list and $\tilde{f}(\cdot)$ is set to be the pointwise posterior mean of $f(\cdot)/\tau$ at $\{t_1, \ldots, t_{90}\}$ obtained under the hierarchy \eqref{eq:frac_prob_with_priors}. We considered three choices for the Hurst exponent, namely, $\{0.5, 0.75, 0.9\}$ ranging from independent increments for $H = 0.5$ to highly correlated increments for $H = 0.9$. We generated the binary data by first evaluating $y(t) = f(t) + B_H(t)$ at $\{t_0, t_1, \ldots, t_n\}$; to simulate the noise vector we sampled $\epsilon^H  \sim \Gauss(0, \tau^2 \Sigma_H)$ with $\tau^2 = 0.05^2, 0.1^2, 0.15^2$. Representing each positive increment of $y(\cdot)$ by 1 the discretized series $Z$ is obtained and the sampling is repeated $R$ times to complete the data generation process. For each combination of $f(\cdot)$, $H$, $\tau$ and $R$ we performed 30 independent evaluations of the proposed framework and in Table \ref{tab:mse} we report the average ReMSE and the average estimated Hurst exponent for $\tau = 0.1$ with the value of $\nu$ fixed at 0.001; results for $\tau = 0.05, 0.15$ are provided in the supplementary document.

\begin{center}
\begin{table}[!ht]
\caption{Relative mean square error (ReMSE) for different choices of the latent trend function $f(t)$ for the model \eqref{eq:fractional_probit} under hierarchy \eqref{eq:frac_prob_with_priors}. For each $f(t)$ three values of the Hurst exponent are considered: $\{0.5, 0.75, 0.9\}$ together with $\{10, 25, 50\}$ replications. The results reported are averages of 30 independent simulation experiments for each combination.}
%\vspace{0.2in}
\huge
\centering
%\begin{flushleft}
\scalebox{0.43}{
\begin{tabular}{cccccccccccc}\toprule
 && \multicolumn{2}{c}{$f_1(t)$} & \multicolumn{2}{c}{$f_2(t)$} & \multicolumn{2}{c}{$f_3(t)$} & \multicolumn{2}{c}{$f_4(t)$} & \multicolumn{2}{c}{$f_5(t)$} \\
\cmidrule{1-12}
Hurst exponent ($H$) & Replications ($R$)& MSE & $\hat{H}$ & MSE & $\hat{H}$ & MSE & $\hat{H}$ & MSE & $\hat{H}$ & MSE & $\hat{H}$\\
\cmidrule{1-12}
\multirow{3}{*}{0.5} & 10 & 1.26 & 0.55 & 1.02 & 0.49 & 0.12 & 0.52 & 0.09 & 0.52 & 1.38 & 0.50  \\
 & 25 & 0.58 & 0.48 & 0.40 & 0.48 & 0.01 & 0.50 & 0.01 & 0.51 & 0.96 & 0.51\\
 & 50 & 0.40 & 0.54 & 0.17 & 0.50 & 0.007 & 0.48 & 0.005 & 0.50 & 0.08 & 0.50\\
 \cmidrule{1-12}
 \multirow{3}{*}{0.75} & 10 & 2.13 & 0.76 & 1.88 & 0.76& 0.14 & 0.74 & 0.28 & 0.76 & 4.81 & 0.74\\
  & 25 & 1.37 & 0.75 & 1.20 & 0.77 & 0.06 & 0.76 & 0.03 & 0.75 & 1.46 & 0.75\\
 & 50 & 0.84 & 0.74 & 0.24 & 0.74  & 0.04 & 0.75 & 0.02 & 0.75 & 0.55 & 0.74\\
 \cmidrule{1-12}
 \multirow{3}{*}{0.9} & 10 & 4.18 & 0.88 & 6.52 & 0.87 & 0.70 & 0.88 & 0.20 & 0.90 & 14.96 & 0.87\\
  & 25 & 3.08 & 0.89 & 2.61 & 0.89 & 0.29 & 0.89 & 0.18 & 0.89 & 5.87 & 0.93\\
 & 50 & 1.11 & 0.89 & 0.99 & 0.88 & 0.08 & 0.87 & 0.07 & 0.89 & 3.34 & 0.88\\
 \hline
 \end{tabular}
}
%\end{flushleft}
\label{tab:mse}
\end{table}
\end{center}

\begin{figure}
%\centering
%\scalebox{0.75}{ 
\begin{subfigure}{0.43\textwidth}
\centering
\includegraphics[height = 4.5cm, width = \textwidth]{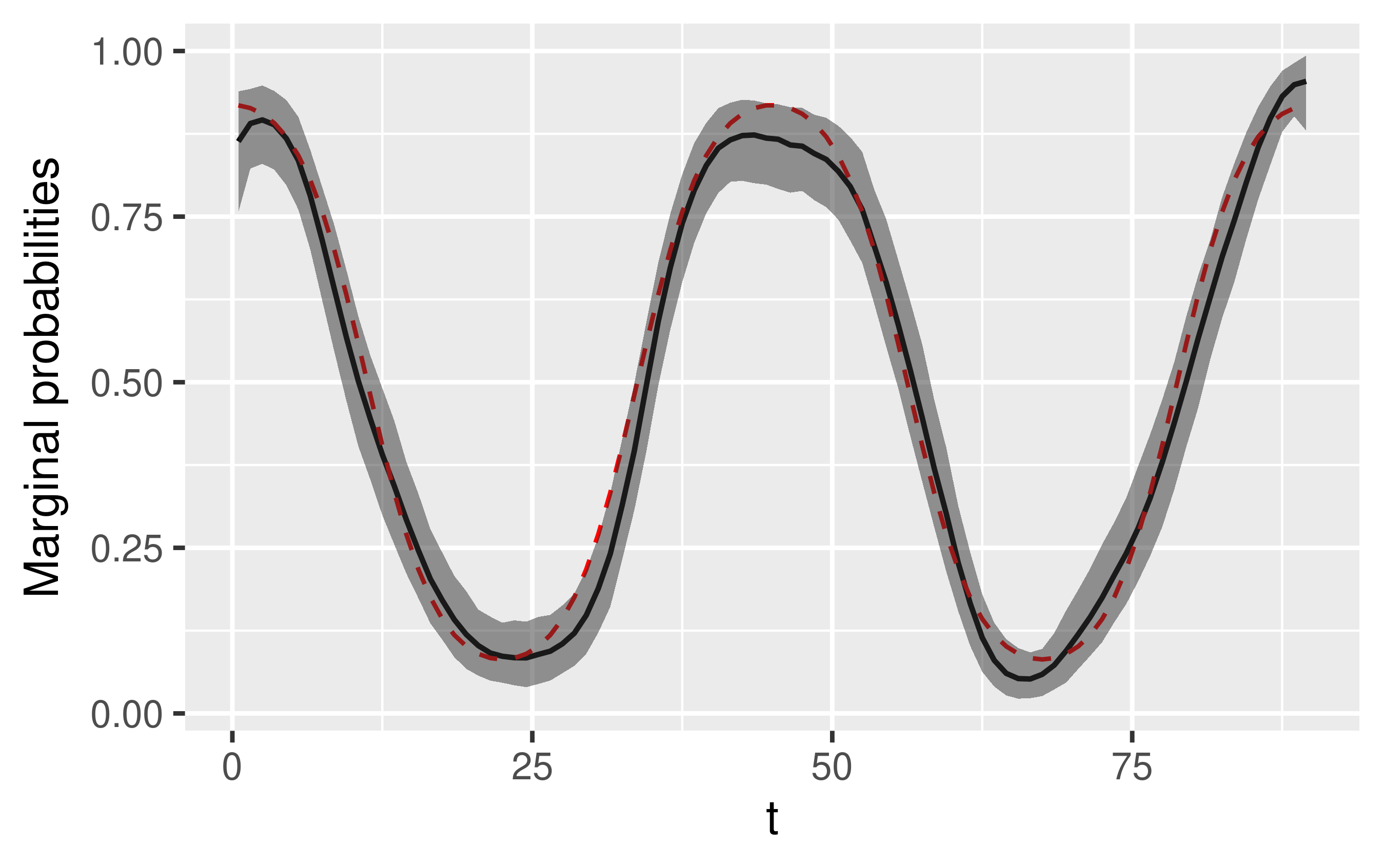} 
\caption{$y(t) = \sin \frac{4\pi t}{90} + B_H(t)$}
%\caption{$f_1(t)$}
\end{subfigure}
\begin{subfigure}{0.50\textwidth}
\centering
\includegraphics[height = 4.4cm, width = \textwidth]{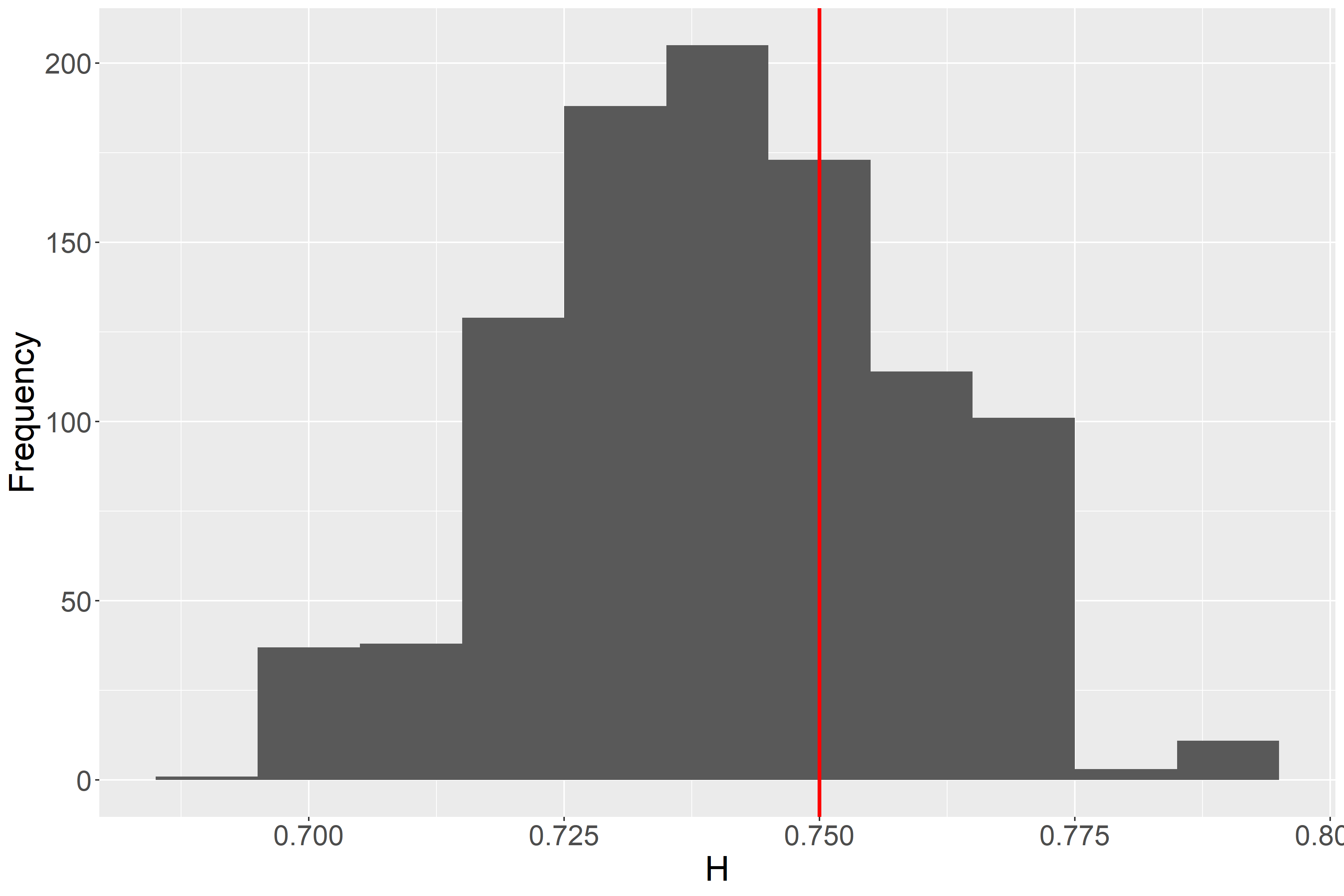}
%\caption{$f_5(t) = 0.1f_1(t)\log f_2(t)$}
\caption{Corresponding posterior samples of $H$}
\end{subfigure}
%}
%\end{center}
\caption{Figure (a) shows the posterior mean and 95\% credible bands for marginal probabilities in one minute intervals when $f(t) =\sin \frac{4\pi t}{90} $. The values of the Hurst coefficient and the number of replications were $H = 0.75$ and $R = 50$, respectively. Red dashed and black solid lines correspond to the true values and the posterior mean respectively. Gray shaded regions are credible bands. Corresponding posterior samples of $H$ are shown in (b). A red line is added at the true value $H = 0.75$.} 
\label{fig:simulation_figures}
\end{figure}
Naturally, the ReMSE in Table \ref{tab:mse} is inversely proportional to the number of replications $R$; decreasing by a factor of two when the number of replications is doubled. Interestingly, the degree of LRD also controls the ReMSE. For all the choices of $f(\cdot)$, the average ReMSE increases with $H$. Large $H$ implies strong dependence in the data which makes the problem of recovering $f(\cdot)$ harder. This was investigated formally in \cite{hall1990nonparametric} who observed that the rates of recovering $f(\cdot)$ decrease with $H$.
Estimates of the Hurst exponent are quite accurate across all the combinations of $R$, $H$ and $f(\cdot)$. This is important in the context of the Amazon bird vocalization data for which we have on average 10 days of data. In Figure \ref{fig:simulation_figures} we show the posterior mean and credible bands together with the true value for the marginal probabilities of events during intervals of size one time unit when the true Hurst exponent is set to $0.75$ and the number of replications available is $R = 50$. Set $p(t_1, t_2) = \Phi[\{f(t_2) - f(t_1)\}/\tau]$ as the true marginal probability under model \eqref{eq:trend_and_noise} with trend function $f(\cdot)$ and let $\hat{p}(t_1, t_2) = \Phi[\{\hat{f}(t_2) - \hat{f}(t_1)\}/\hat{\tau}]$ denote samples from the posterior distribution of $f$ and $\tau$ obtained fitting Algorithm \ref{algo:MCMC_one_species}. The black line in Figure \ref{fig:simulation_figures} is the posterior mean of the marginal probabilities $\hat{p}(t_1, t_2)$ and the red line plots $p(t_1, t_2)$. We also show the pointwise 95\% credible bands of $\hat{p}(t_1, t_2)$.  The best result is obtained for $f_1(t)$. The credible bands mostly provide accurate uncertainty quantification for all the cases. However, when the number of replications $R$ is smaller the problem of accurately estimating the marginal probabilities becomes much harder, especially for high values of $H$. Posterior samples of the Hurst exponent for one case are also included in the figure.  
%$f_1/\tau,(\cdot), \ldots, f_5(\cdot)/\tau$ when $R =50$ replications are available. For all the cases, pointwise credible bands capture the true functions with the best performance achieved for $f_3(\cdot)$ and $f_4(\cdot)$.  

To further investigate the behavior of the posterior distribution of the Hurst exponent, we carried out an independent simulation experiment focusing on the coverage probability of the credible intervals. We fix the number of replicates at $R = 5$ and vary the Hurst exponent together with the latent trends as above. For each such combination, we generated 100 data sets and applied model \eqref{eq:frac_prob_with_priors}. Our findings for 95\% credible intervals are summarized in Table \ref{tab:Hurst_coverage}. The coverage probabilities (CP) for all the cases considered are close to the nominal level. The average lengths (l) of the intervals vary substantially for different choices of $H$ along with the standard deviation. For example, the average length of the intervals are maximum for the case $H = 0.5$ with very little variation but when $H = 0.9$ the intervals become shorter on average although their variability increases by almost a factor of 3.

%\begin{center}
\begin{table}[!ht]
    %\centering
    \scalebox{0.82}{
    \begin{tabular}{ccccccccccc}\toprule
   \multirow{2}{*}{ $H$ }      & \multicolumn{2}{c}{$f_1(\cdot)$} & \multicolumn{2}{c}{$f_2(\cdot)$} & \multicolumn{2}{c}{$f_3(\cdot)$} & \multicolumn{2}{c}{$f_4(\cdot)$} & \multicolumn{2}{c}{$f_5(\cdot)$} \\
   \cmidrule{2-11}
      & CP & l &  CP & l  & CP & l  & CP & l  & CP & l   \\
      \cmidrule{1-11}
      0.5 & 0.97 & 0.16(0.05) & 0.94 & 0.19(0.04) & 0.92 & 0.20(0.02) & 0.92 & 0.21(0.03) & 0.98 & 0.20(0.02)\\
      0.75 & 0.91 & 0.15(0.03) & 0.92 & 0.14(0.14) & 0.92 & 0.14(0.02) & 0.90 & 0.14(0.02) & 0.93 & 0.13(0.01)\\
      0.9 & 0.90 & 0.13(0.10) & 0.89 & 0.14(0.11) & 0.91 & 0.12(0.10) & 0.90 & 0.12(0.09) & 0.94 & 0.14(0.10)\\
      \cmidrule{1-11}
    \end{tabular}
    }
    \caption{Coverage probability (CP) of 95\% credible intervals for the Hurst exponent under the hierarchy \eqref{eq:frac_prob_with_priors}. Also included are the average length (l) of the credible intervals with corresponding standard deviation inside parenthesis. The number of replicates in each case is $R = 5$.}
    \label{tab:Hurst_coverage}
\end{table}
%\end{center}

\section{Application to Amazon bird vocalization data} \label{sec:amazon_frap}
\subsection{Analysis and results}
We applied the FRAP model to the 15 bird species mentioned in Section \ref{sec:primary_analysis}. For each of these species we have 180 minutes of recordings available for multiple days. The estimated Hurst exponents for these 15 species are reported in Table $\ref{tab:hurst_data}$. All the species show high long range dependence in their temporal vocalization patterns indicating strong evidence of non-Poissonian dynamics. The posterior mean estimate of the Hurst exponent for the birds range from a minimum of 0.83 up to 0.94. The variation in the Hurst exponent across species is very small with an overall mean of 0.88 and standard deviation 0.04. The high value of the Hurst exponents is consistent with the data in the sense that birds either vocalize or remain silent over long periods of time. We note that this is a combination of two factors, which are occurrence and vocalization activity. First, due to their movement activity, a bird individual may be in the vicinity of the recorder for some time and then move to another location. Second, conditional on the bird being present, it may sustain its vocalization activity over some time and remain silent over another time.
%The presence or absence of an event in this case can be interpreted as alternating renewal process. An alternating renewal process is LRD if any of the activity or inactivity periods have a heavy tailed distribution \citep{heath1998heavy} which is in agreement with the patterns observed for these data. We note here that unlike for an alternating renewal process, we do not observe the exact lengths of times the system is active or inactive. 

Figure \ref{fig:individual_trends} shows posterior means and 95\% pointwise intervals for the species-specific marginal probabilities of vocalizations occurring in each of the 180 
time intervals between $5.15$ - $8.15$ AM for all 15 species listed in Table \ref{tab:hurst_data}.  Due to data sparsity and the high Hurst exponent, the raw posterior samples exhibited spiky patterns over time, and hence we (mildly) smoothed the samples prior to calculating the posterior summaries in Figure \ref{fig:individual_trends}.  While these trends should not be over-interpreted, we do see some general patterns appearing. For example, for Cercomarca cinerascens, Frederickena viridis, Grallaria varia, Micrastur mirandollei, Myrmeciza ferruginea, Percnostola ruffifrons, Pipra erythrocephala, Pithys albifrons and Ramphastos vitellinus we see an increase in vocalization activity after 7 AM, whereas Automolus ochrolaemus, Corythopis torquata, Hylexetastes perrotii, and Ibycter americanus more or less maintain a uniform activity level during this time.  Micrastur gilvicollis and Hylophilus muscicapinus show more activity during the early hours of the day. Since groups of birds show similar vocalization patterns, in the Supplementary Materials Section S.1, we extend the FRAP framework to a hierarchical setting that shares information across different species.

\begin{figure}[!ht]
\centering
\scalebox{0.88}{
\begin{subfigure}{0.3\textwidth}
\includegraphics[width=0.95\linewidth, height=3.5cm]{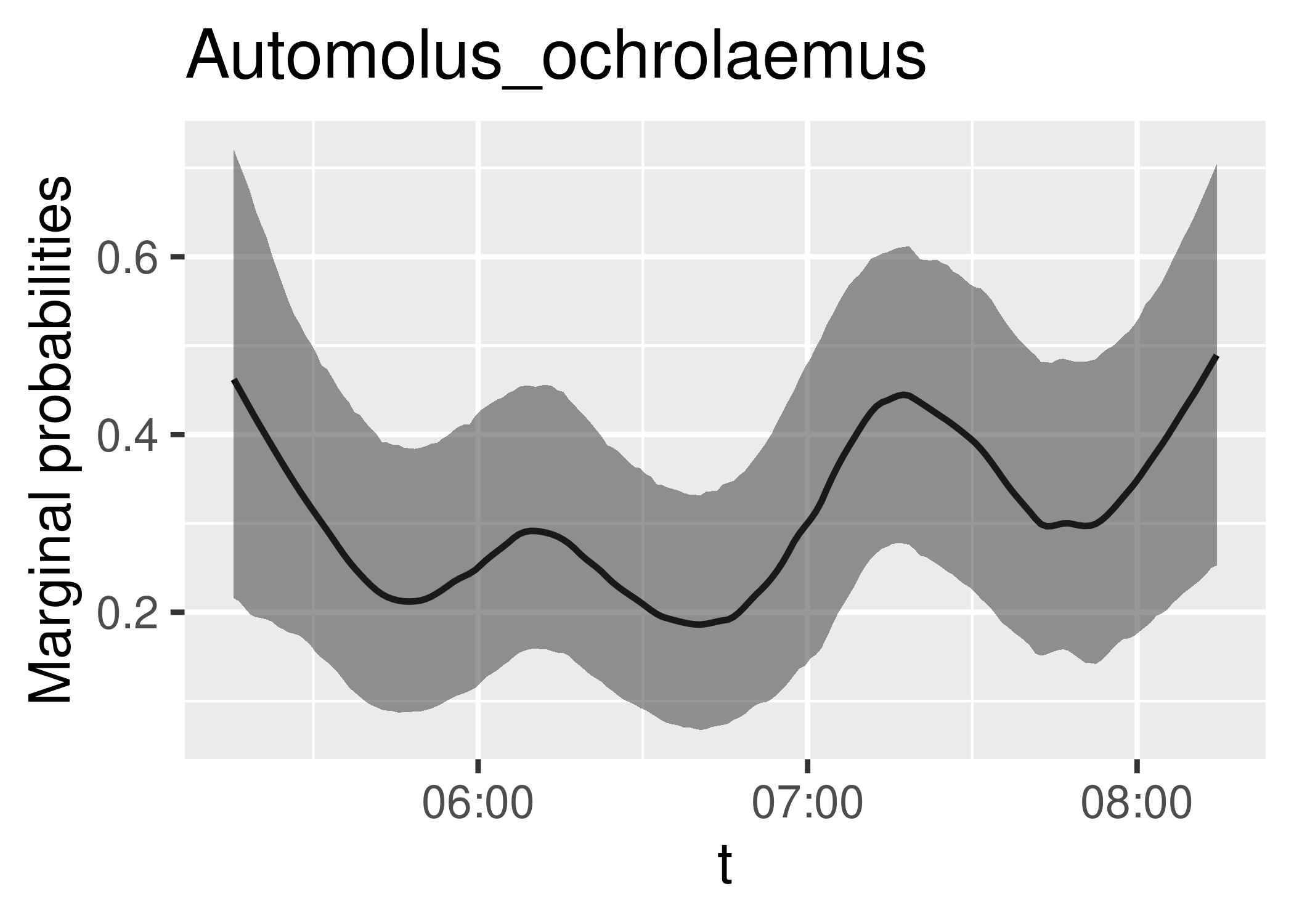} 
%\caption{Marginal probabilities of bird vocalization}
%\label{fig:subim1}
\end{subfigure}
\begin{subfigure}{0.3\textwidth}
\includegraphics[width=0.95\linewidth, height=3.5cm]{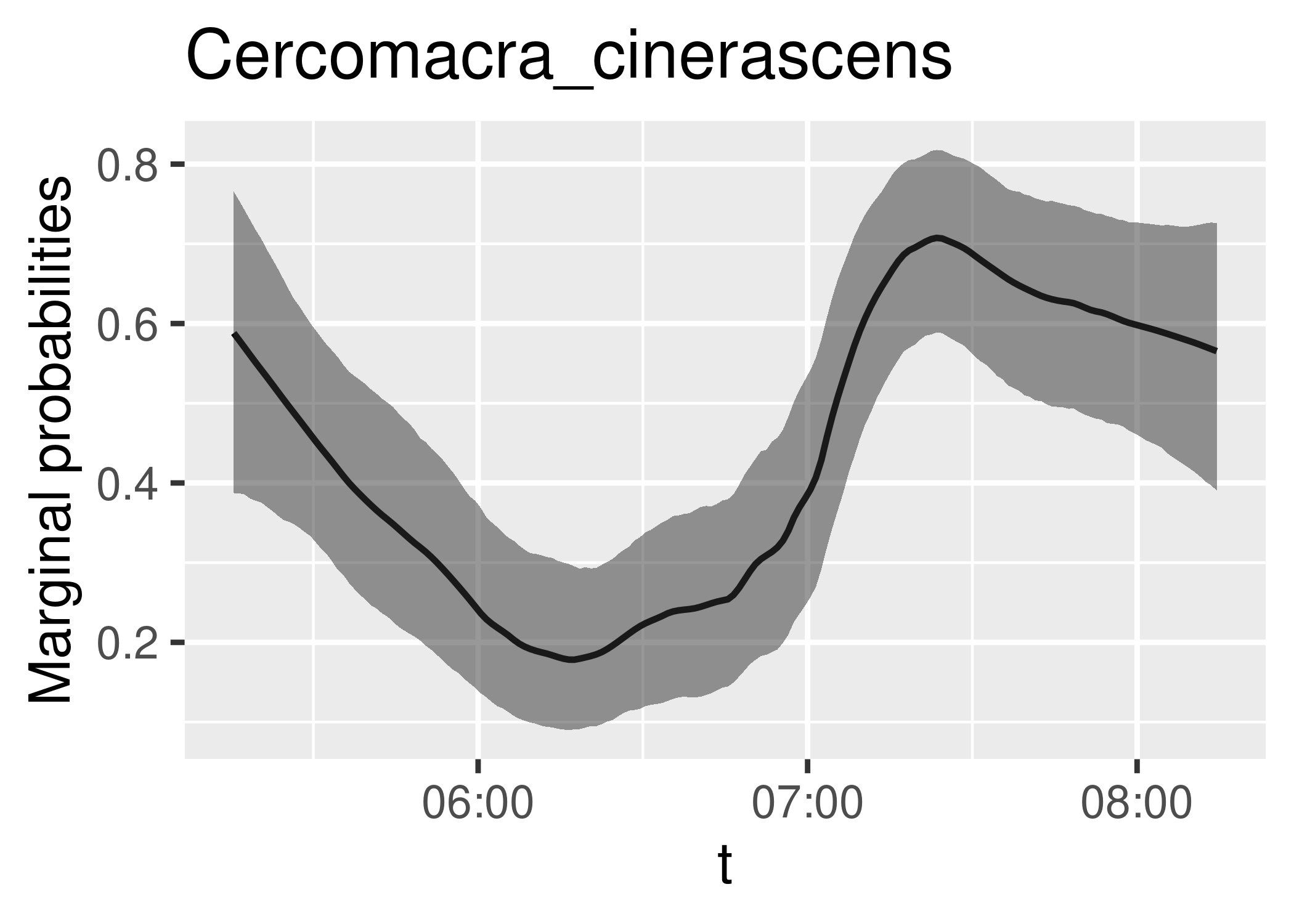}
%\caption{Conditional probabilities of bird vocalization}
%\label{fig:subim2}
\end{subfigure}
\begin{subfigure}{0.3\textwidth}
\includegraphics[width=0.95\linewidth, height=3.5cm]{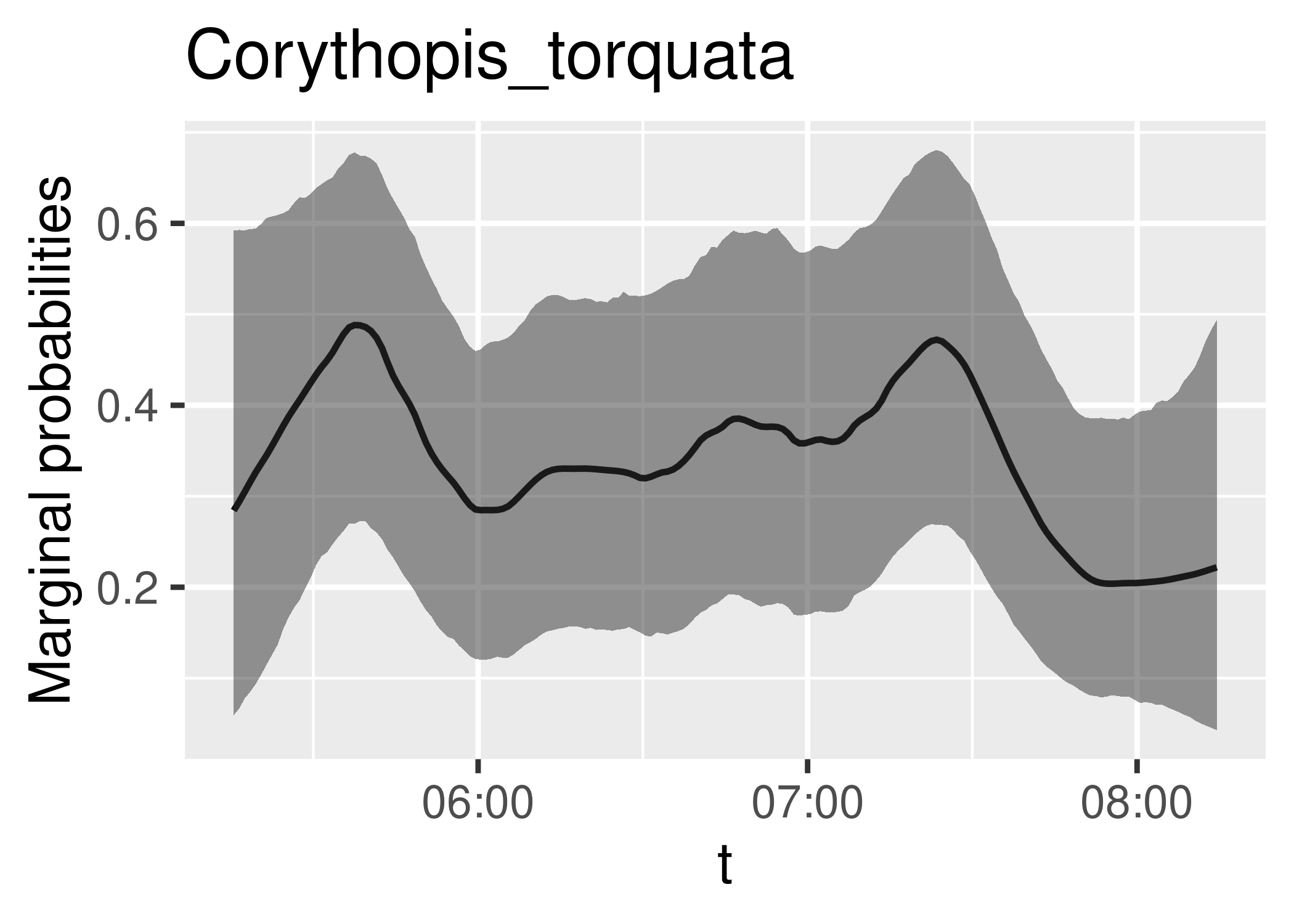}
%\caption{Conditional probabilities of bird vocalization}
%\label{fig:subim2}
\end{subfigure}
}
\scalebox{0.88}{ 
\begin{subfigure}{0.3\textwidth}
\includegraphics[width=0.95\linewidth, height=3.5cm]{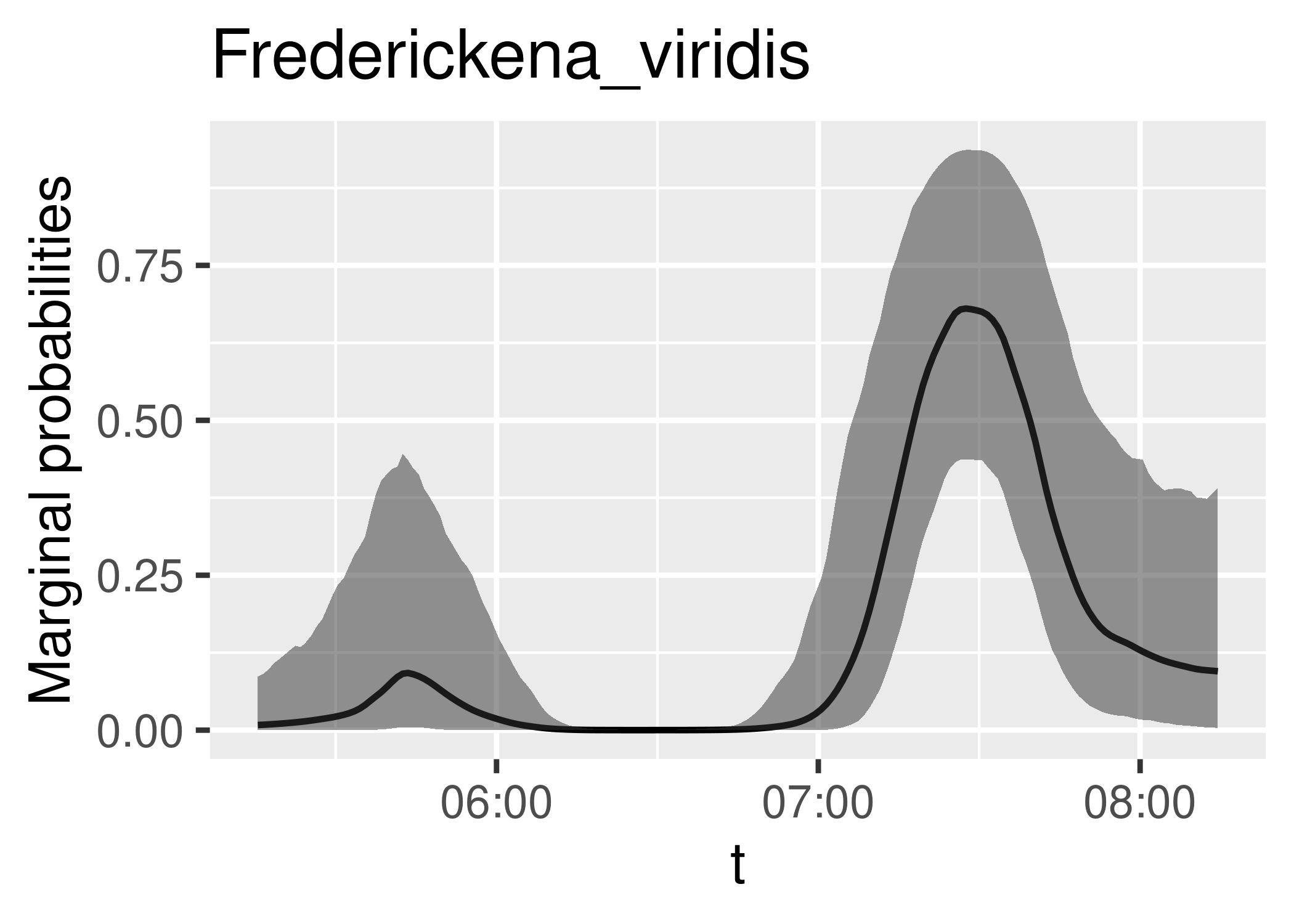}
%\caption{Conditional probabilities of bird vocalization}
%\label{fig:subim2}
\end{subfigure}
\begin{subfigure}{0.3\textwidth}
\includegraphics[width=0.95\linewidth, height=3.5cm]{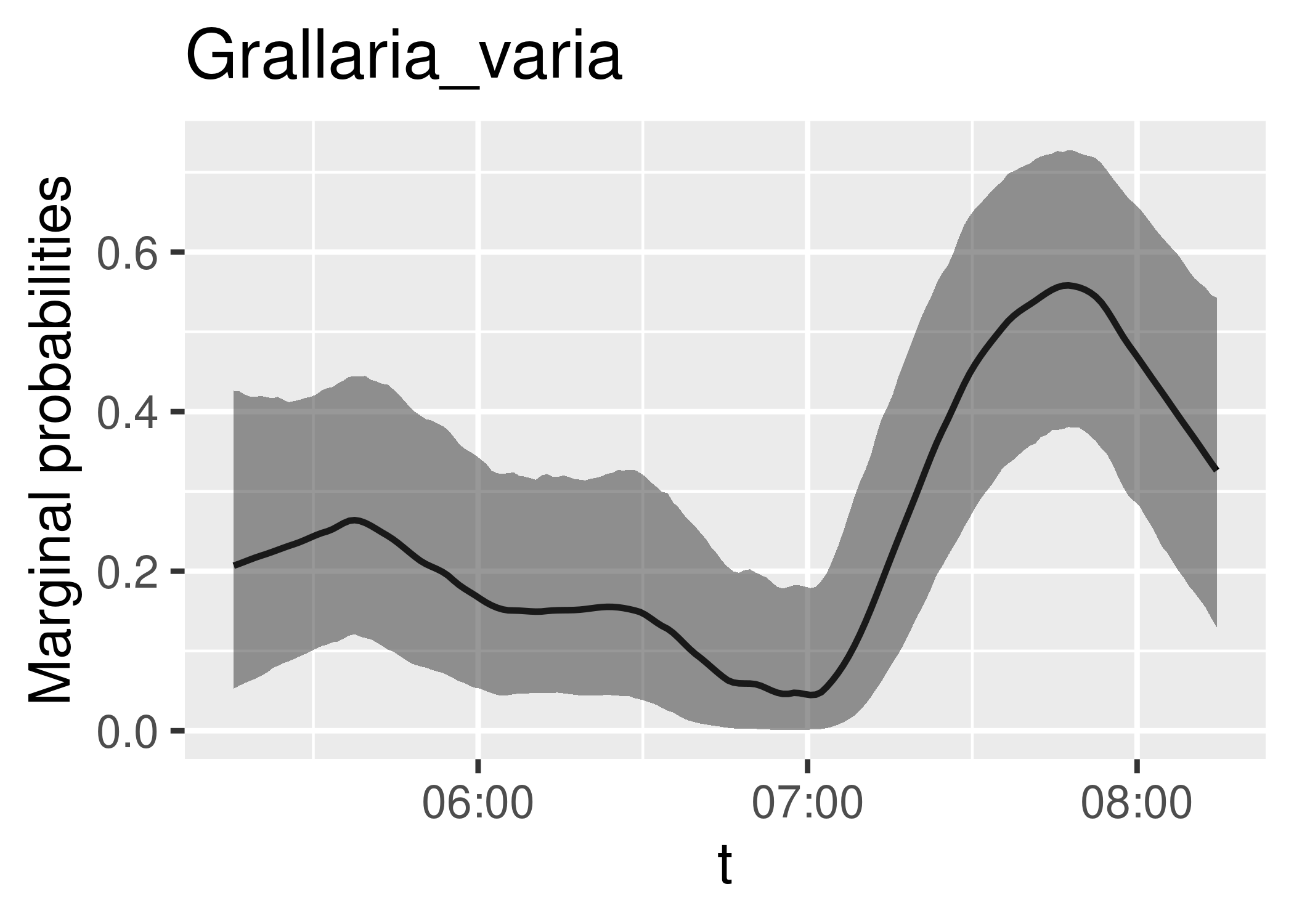}
%\caption{Conditional probabilities of bird vocalization}
%\label{fig:subim2}
\end{subfigure}
\begin{subfigure}{0.3\textwidth}
\includegraphics[width=0.95\linewidth, height=3.5cm]{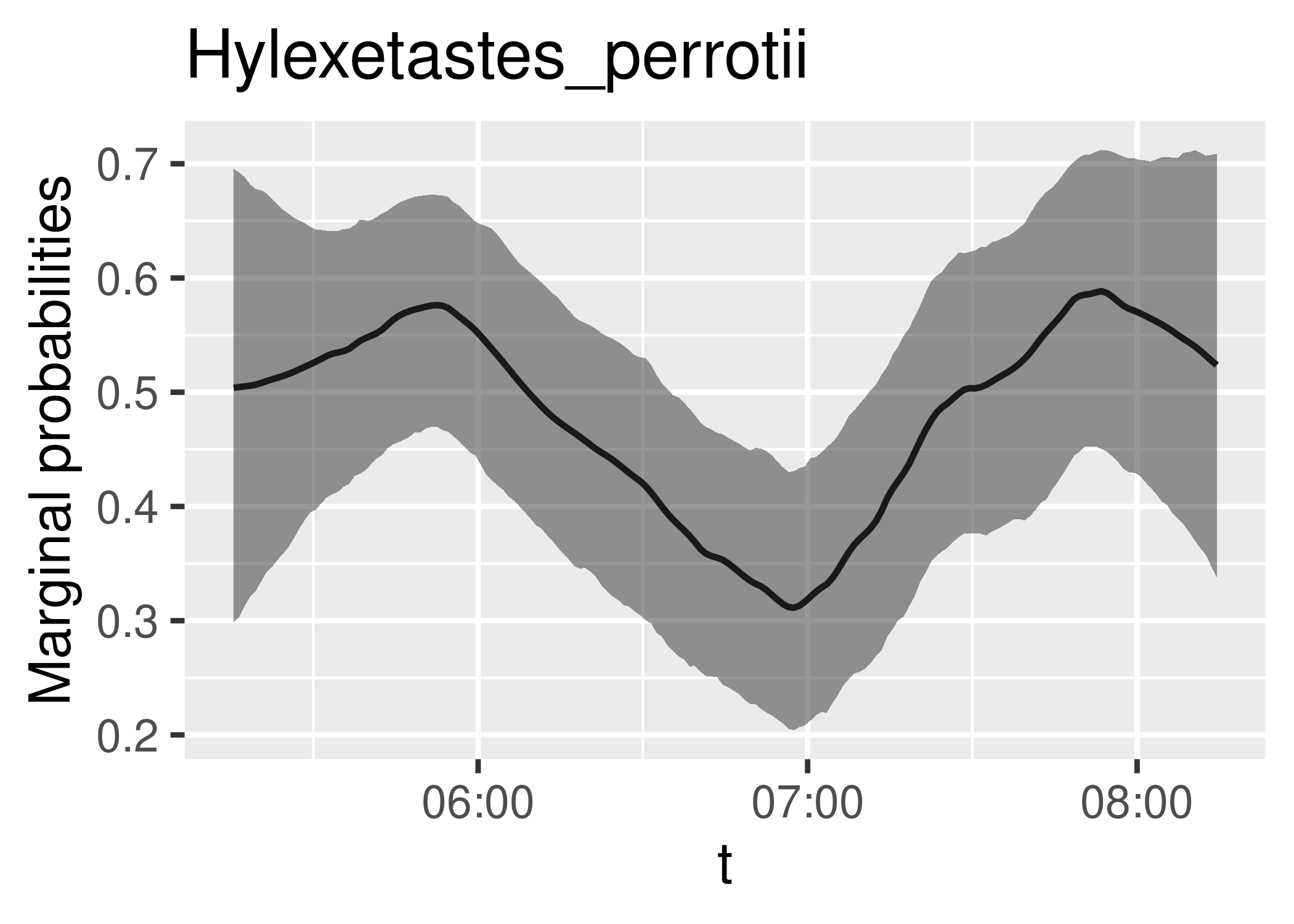}
%\caption{Conditional probabilities of bird vocalization}
%\label{fig:subim2}
\end{subfigure}
}
\scalebox{0.88}{ 
\begin{subfigure}{0.3\textwidth}
\includegraphics[width=0.95\linewidth, height=3.5cm]{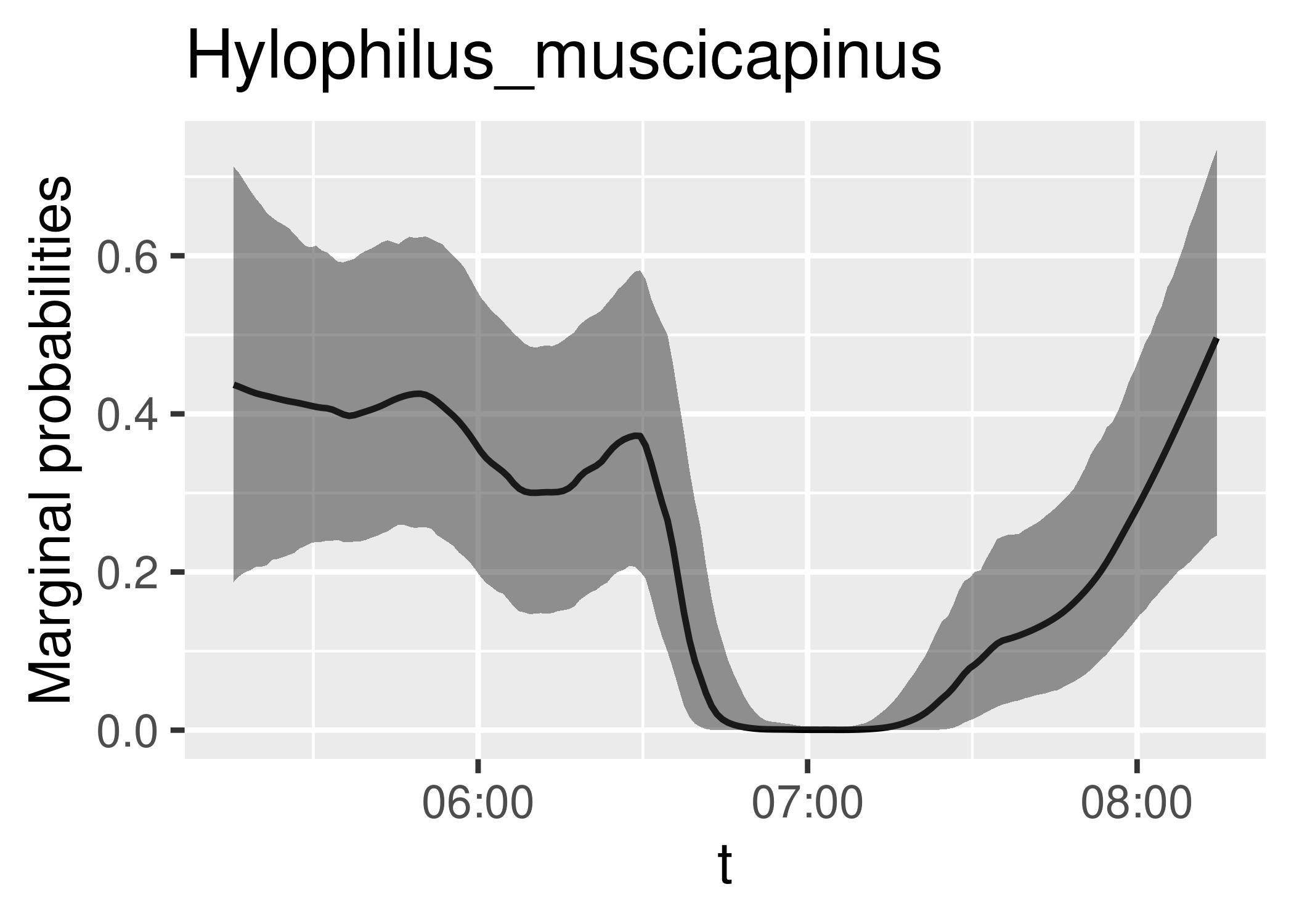}
%\caption{Conditional probabilities of bird vocalization}
%\label{fig:subim2}
\end{subfigure}
\begin{subfigure}{0.3\textwidth}
\includegraphics[width=0.95\linewidth, height=3.5cm]{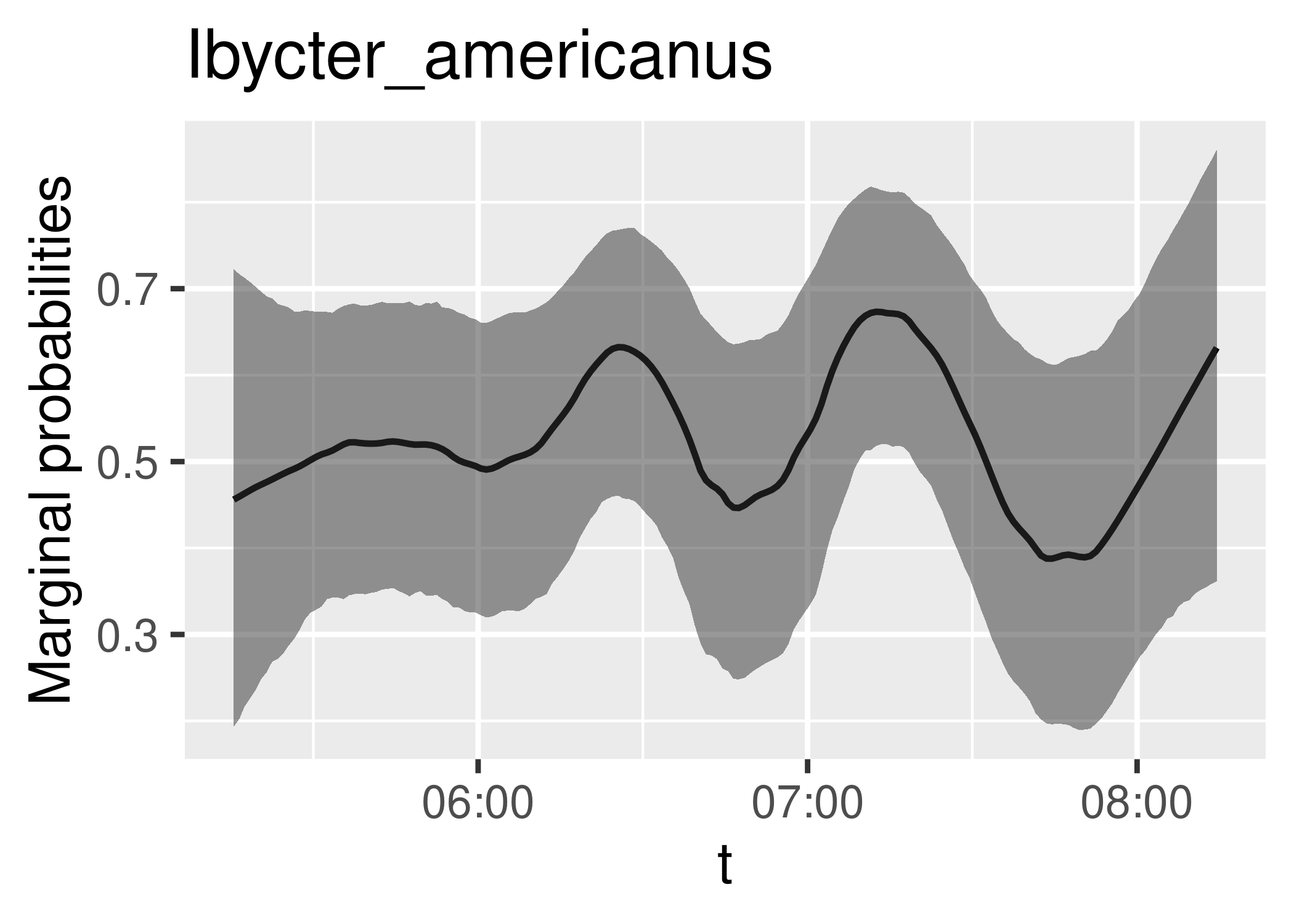}
%\caption{Conditional probabilities of bird vocalization}
%\label{fig:subim2}
\end{subfigure}
\begin{subfigure}{0.3\textwidth}
\includegraphics[width=0.95\linewidth, height=3.5cm]{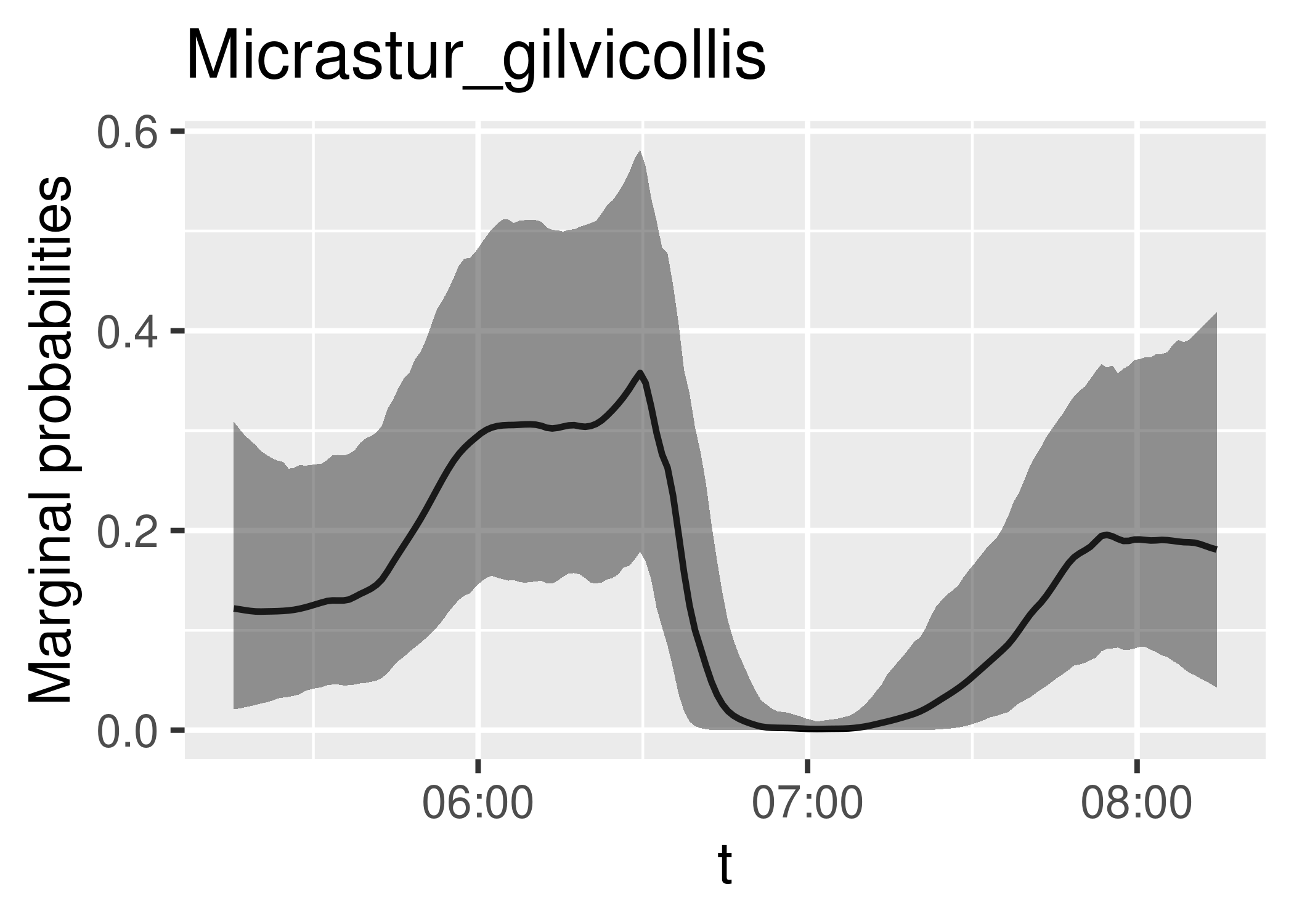}
%\caption{Conditional probabilities of bird vocalization}
%\label{fig:subim2}
\end{subfigure}
}
\scalebox{0.88}{ 
\begin{subfigure}{0.3\textwidth}
\includegraphics[width=0.95\linewidth, height=3.5cm]{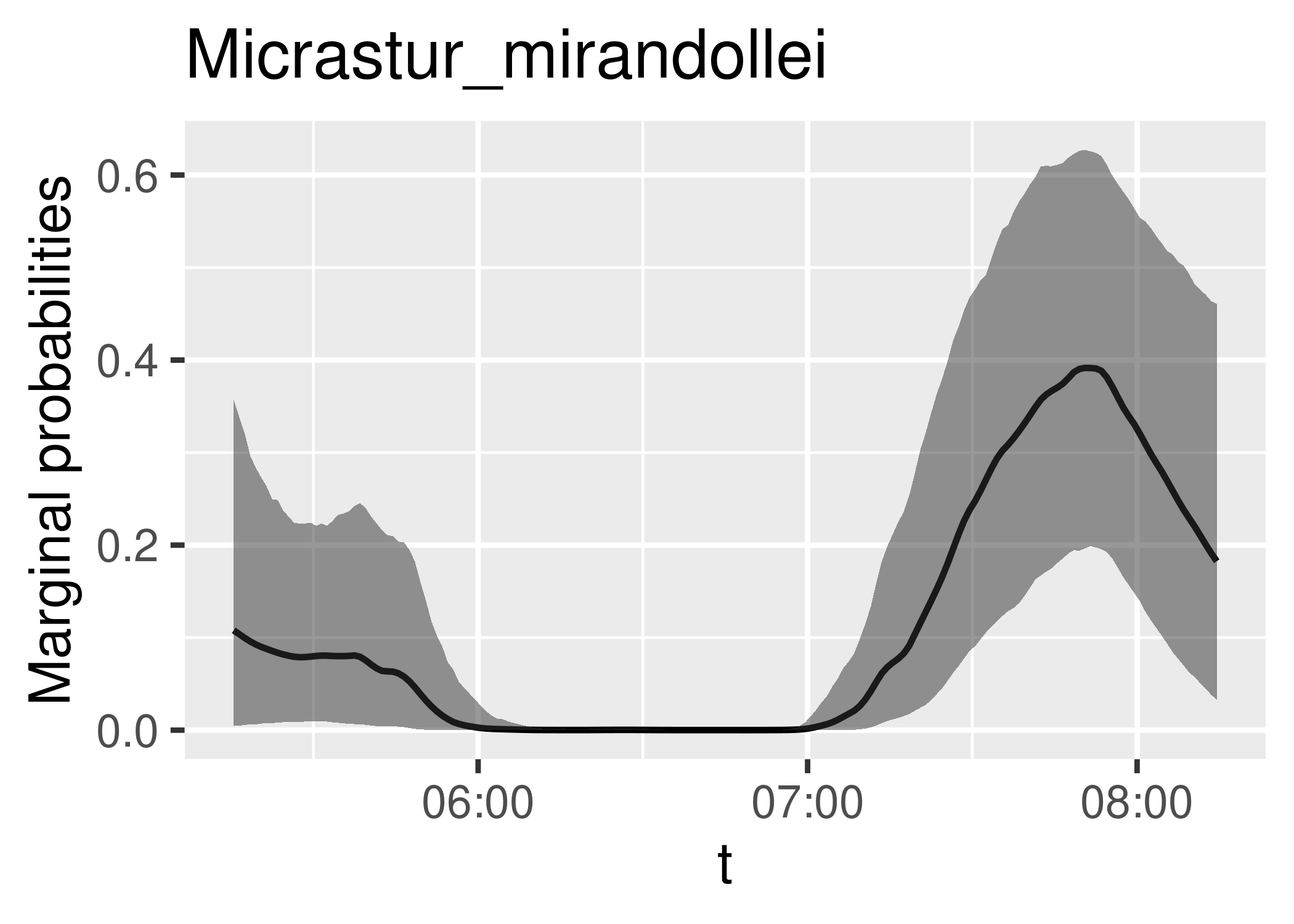}
%\caption{Conditional probabilities of bird vocalization}
%\label{fig:subim2}
\end{subfigure}
\begin{subfigure}{0.3\textwidth}
\includegraphics[width=0.95\linewidth, height=3.5cm]{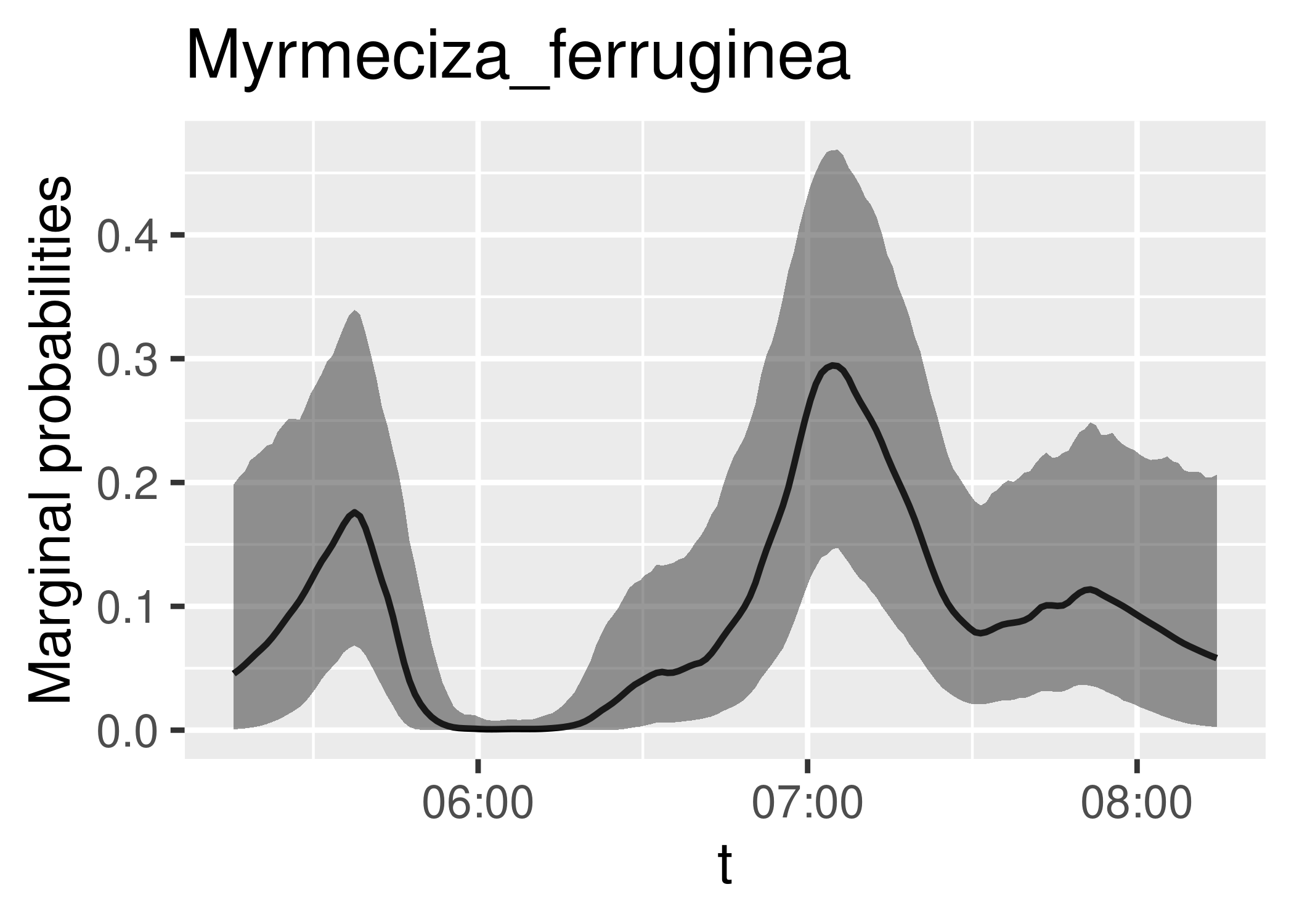}
%\caption{Conditional probabilities of bird vocalization}
%\label{fig:subim2}
\end{subfigure}
\begin{subfigure}{0.3\textwidth}
\includegraphics[width=0.95\linewidth, height=3.5cm]{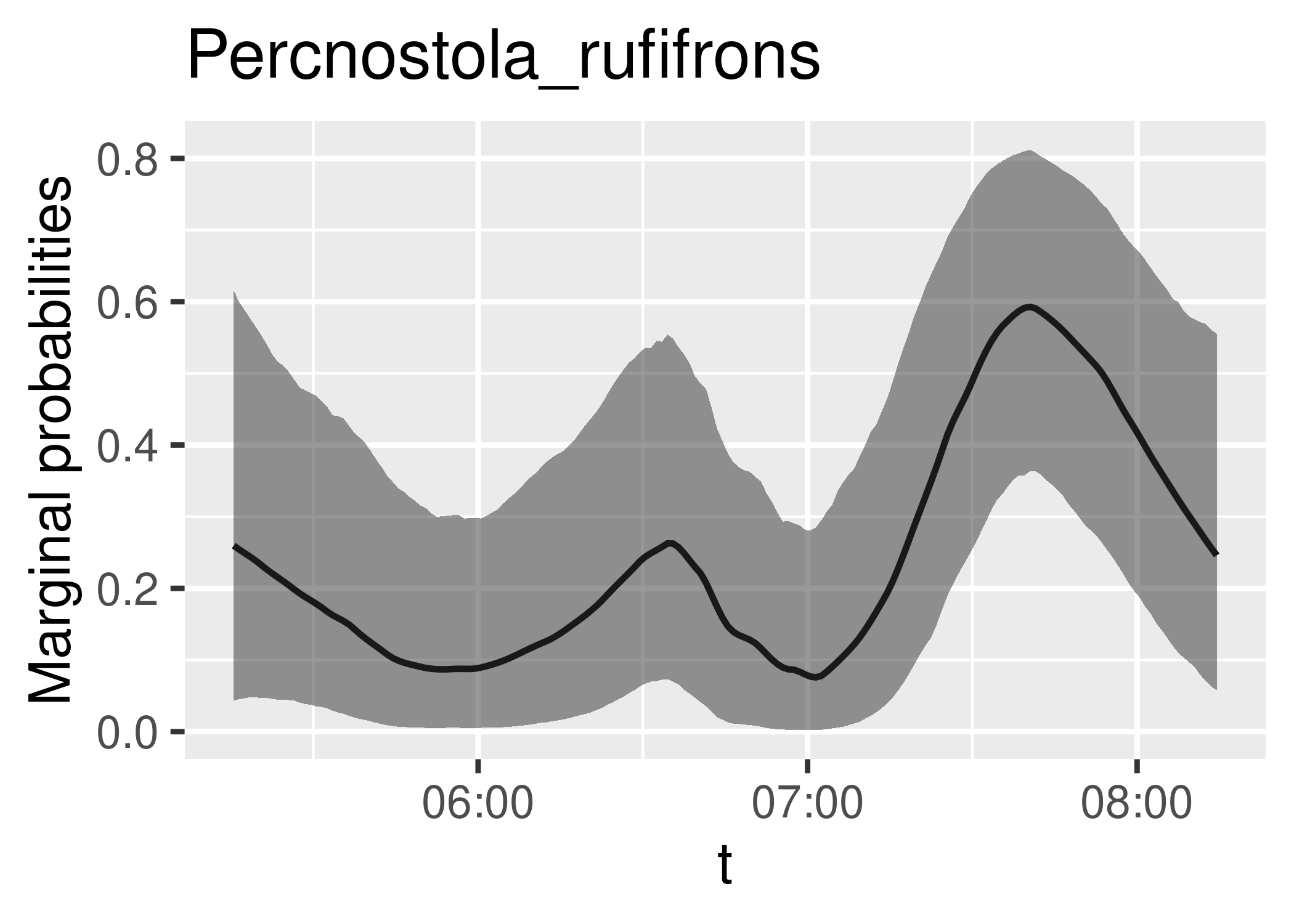}
%\caption{Conditional probabilities of bird vocalization}
%\label{fig:subim2}
\end{subfigure}
}
\scalebox{0.88}{ 
\begin{subfigure}{0.3\textwidth}
\includegraphics[width=0.95\linewidth, height=3.5cm]{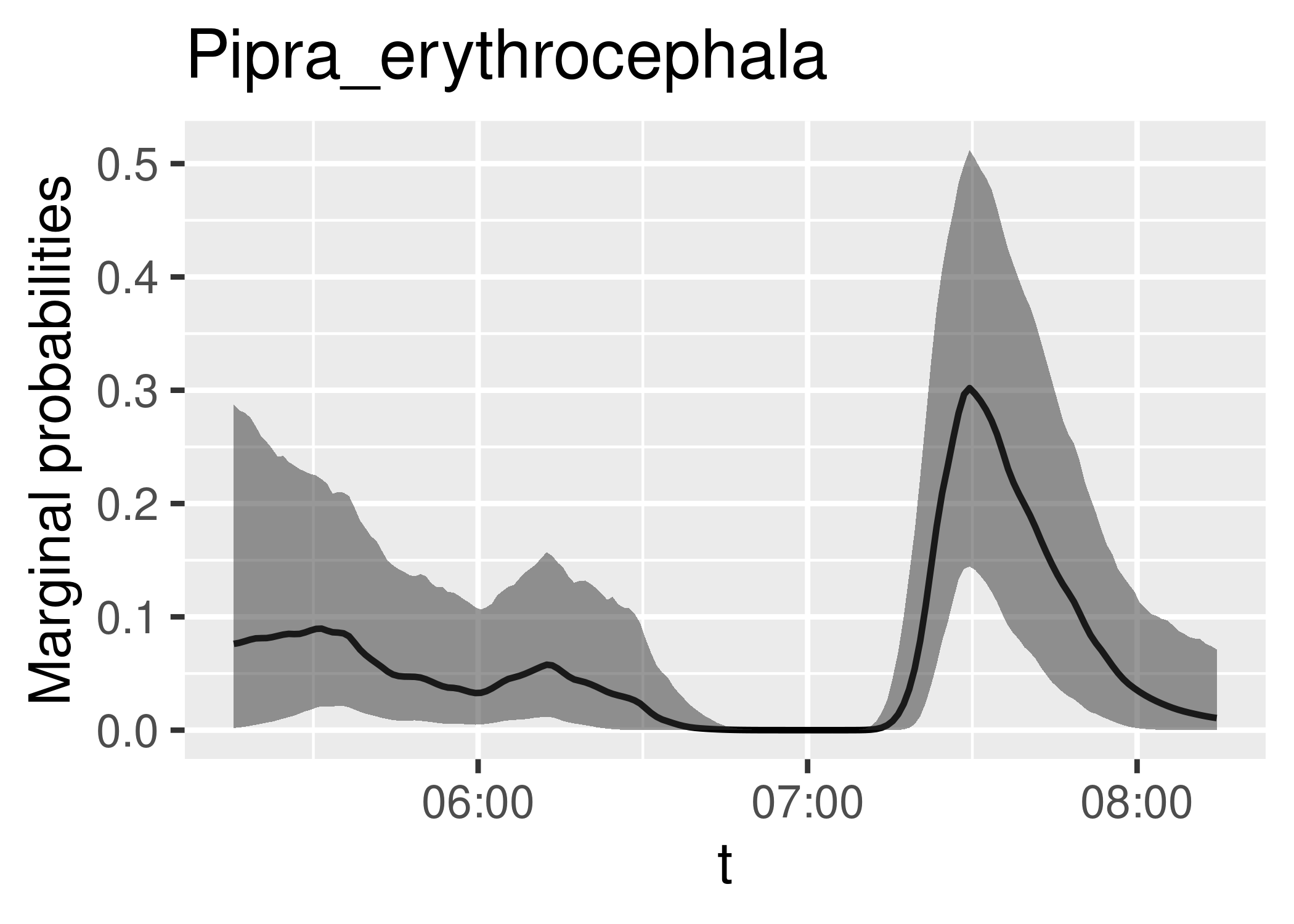}
%\caption{Conditional probabilities of bird vocalization}
%\label{fig:subim2}
\end{subfigure}
\begin{subfigure}{0.3\textwidth}
\includegraphics[width=0.95\linewidth, height=3.5cm]{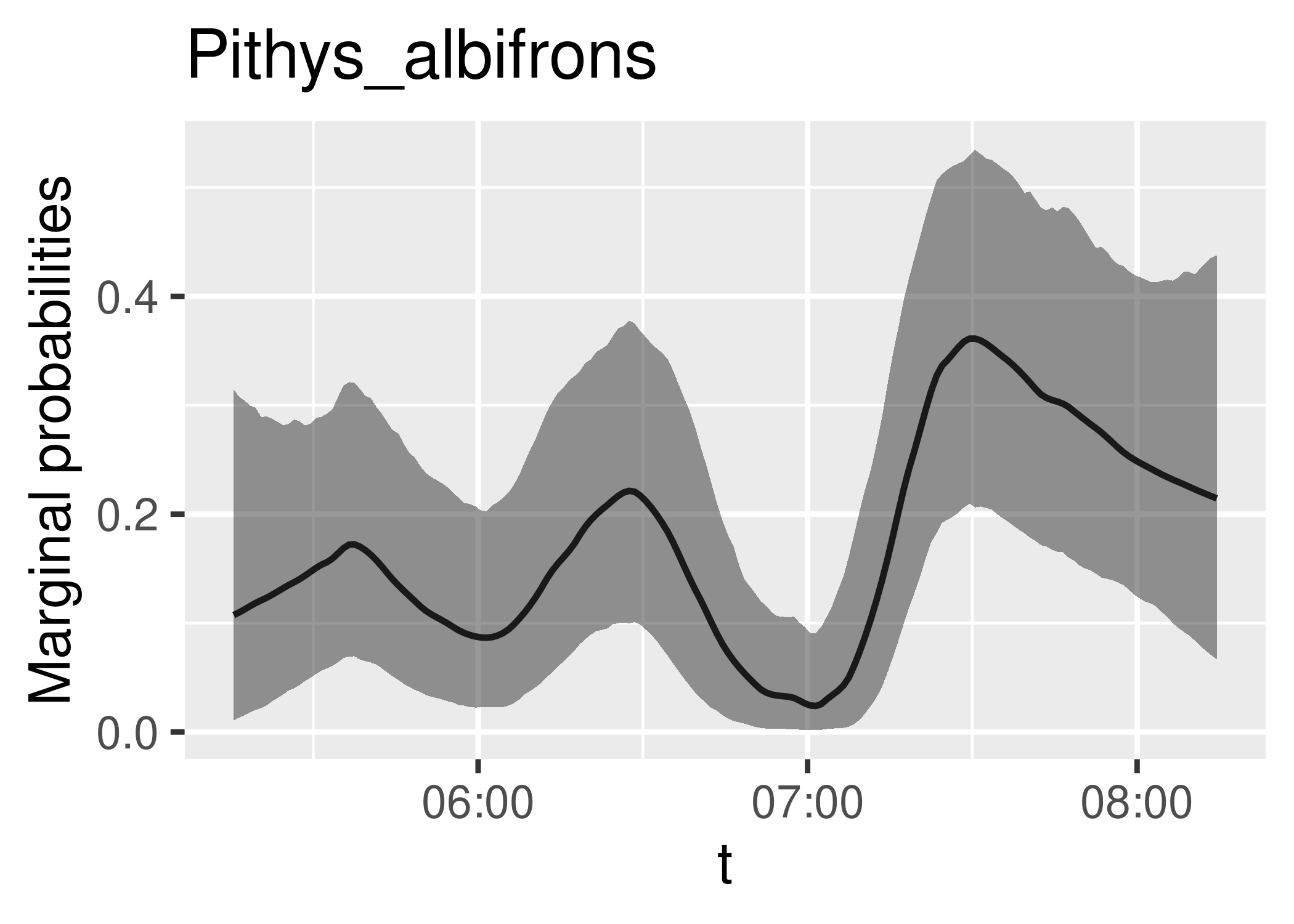}
%\caption{Conditional probabilities of bird vocalization}
%\label{fig:subim2}
\end{subfigure}
\begin{subfigure}{0.3\textwidth}
\includegraphics[width=0.95\linewidth, height=3.5cm]{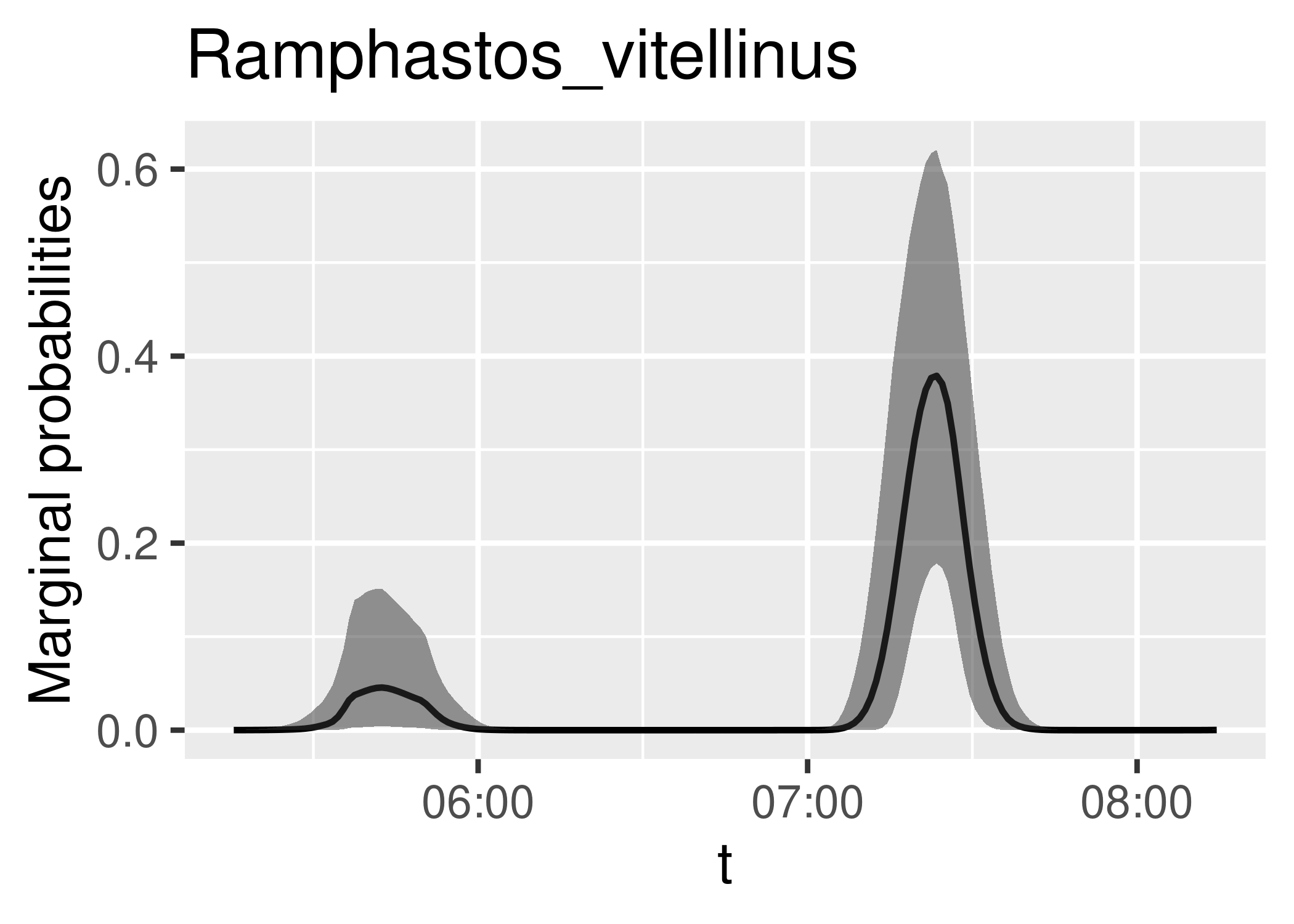}
%\caption{Conditional probabilities of bird vocalization}
%\label{fig:subim2}
\end{subfigure}
}
%}
\caption{Smoothed marginal probabilities of vocalization obtained by fitting model \eqref{eq:trend_and_noise} for the 15 species listed in Table \ref{tab:hurst_data} for 180 test intervals of duration one minute from $5.15$ - $8.15$ AM. Shaded regions are 95\% credible intervals and black lines are posterior means.
%Estimated marginal probability of vocalization in an one minute interval for the 15 species listed in Table \ref{tab:hurst_data}. Shaded regions correspond to 95\% credible intervals and black line is the posterior mean.
}
%\label{fig:marginal_and_conditional_probabilities}
\label{fig:individual_trends}
\end{figure}

\subsection{FRAP vs MMPP}\label{sec:goodness_of_fit}
We compare the fit of the proposed FRAP model with the MMPP model \citep{davison1996some} for discretized event data via summary statistics derived from the posterior distribution and maximum likelihood estimates, respectively. The particular summary statistics that we are interested are the conditional probabilities in Figure \ref{fig:marginal_and_conditional_probabilites}. In the context of the FRAP model,
the distribution of the binary indicators $Z$ is completely characterized by the latent variables $W$. The posterior predictive distribution of $W_{R+1}$ given the observed binary indicators $Z_1, \ldots, Z_R$ is $p(W_{R+1} \mid Z_1, \ldots, Z_R) = \int p(W_{R+1} \mid \theta_*) p(\theta_* \mid Z_1, \ldots, Z_R)$, where $\theta_* = (\mathbf{f}, \beta, \tau, \sigma, \phi)^\T$ and $p(W_{R+1} \mid \theta_*) \sim \mathrm{N}(A\mathbf{f}, \tau^2\Sigma_H)$, $H = \log\{\beta/(1-\beta)\}$. To sample the latent variable $W_{R+1}$, we use the MCMC samples of $\theta_*$ obtained from Algorithm \ref{algo:MCMC_one_species}, i.e. given $\theta_*^{(l)}$, the $l$-th MCMC sample from $p(\theta_* \mid Z_1, \ldots, Z_R)$, we draw $W_{R+1}^{(l)} \sim \mathrm{N}(A\mathbf{f^{(l)}}, \tau^{2{(l)}}\Sigma_{H}^{(l)})$. Then equation \eqref{eq:trend_and_noise} is used to obtain the corresponding binary series $Z_{R+1}^{(l)}$.  

The MMPP assumes event occurrence is governed by specific states of an unobserved continuous time Markov chain, hereafter referred to as CTMC, $X(t)$ with finite state space $\{1, 2, \ldots, K\}$ and instantaneous transition probability matrix $G \in \Re^{K \times K}$. Given the chain is in state $k \in \{ 1,\ldots, K\}$ at time $t$, events occur following a Poisson process with rate $\lambda_k$. The event generating process is then parameterized by the $G$ and $\mathbf{\lambda} = \{\lambda_1, \ldots, \lambda_k\}$. The likelihood of a discretized series of events under the MMPP model has been derived in \cite{davison1996some}. Let $\hat{G}$ and $\hat{L}$ denote the maximum likelihood estimates of $G$ and $L$, respectively using $R$ replicates of binary event indicators $Z_1, \ldots, Z_R$. For the Amazon bird vocalization data, we generate a series of binary event indicators $Z_{R+1}$ using the plug-in estimates $\hat{G}$ and $\hat{L}$ with $k = 2$.

Having generated event indicators $Z_{R+1}$ from the two models for each of the 15 species in Table \ref{tab:hurst_data}, we compute the conditional probability of occurrence of a vocalization given a vocalization in the previous interval for time scales $\Delta t = \{1,2,4,9,15,30,60,90\}$; for the FRAP model we compute the conditional probabilities for each MCMC sample $Z_{R+1}^{(l)}$ and consider the average. In the left panel of Figure \ref{fig:mmpp_vs_frap} we plot these probabilities using the estimates obtained from the MMPP model and in the right panel we plot the average conditional probability for different time scales across MCMC samples. The proposed FRAP model captures the scaling of the conditional probabilities seen in the observed data (Figure \ref{fig:marginal_and_conditional_probabilites}) while the MMPP does not. We also fitted the MMPP with $K = 3$ states but the results were very similar.

\begin{figure}
\centering
 \scalebox{0.95}{
 \begin{subfigure}{0.65\textwidth}
\includegraphics[width=0.95\linewidth, height=6cm]{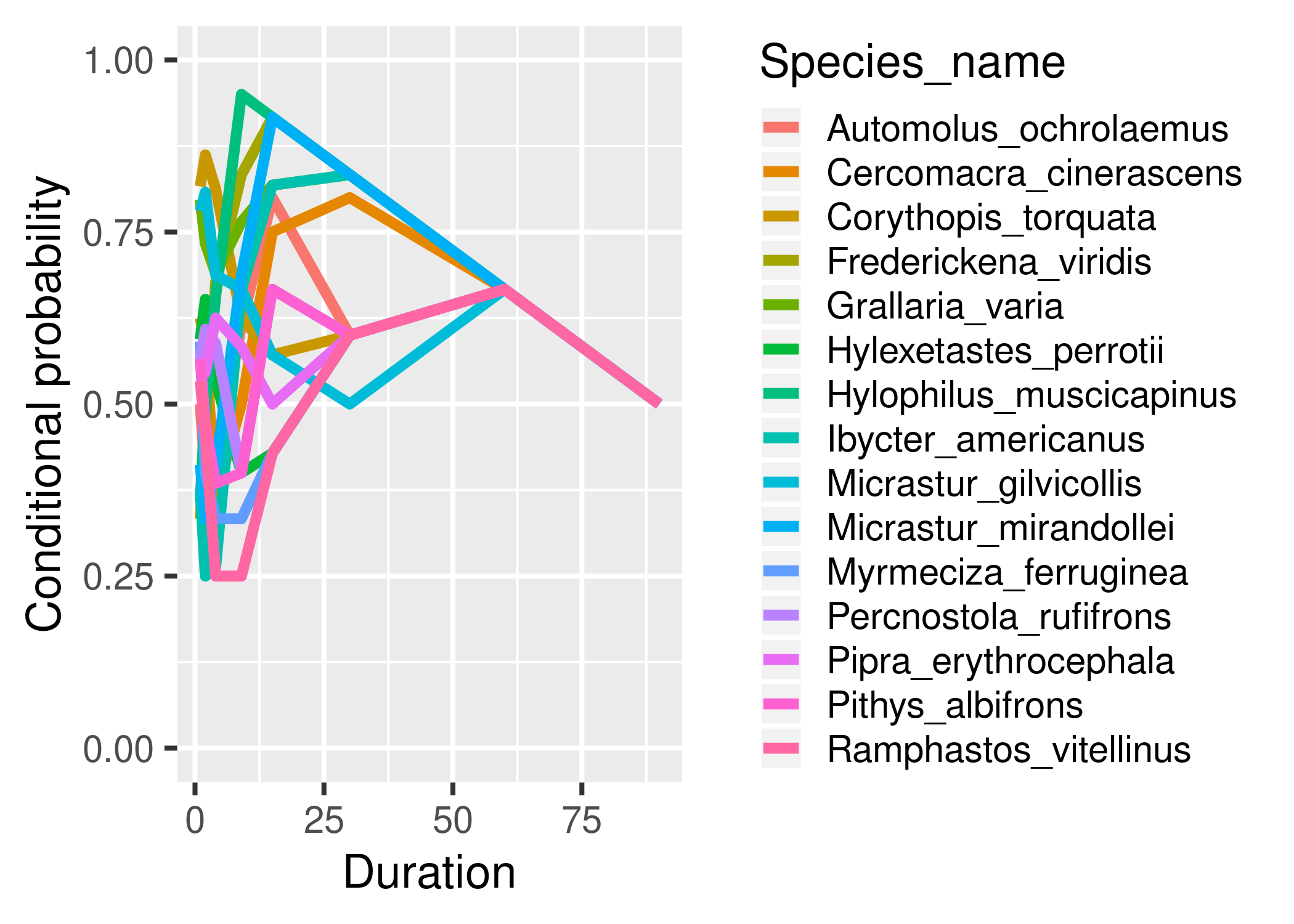} 
\caption{MMPP}
\label{fig:subim1}
\end{subfigure}
\begin{subfigure}{0.35\textwidth}
\includegraphics[width=0.85\linewidth, height=6cm]{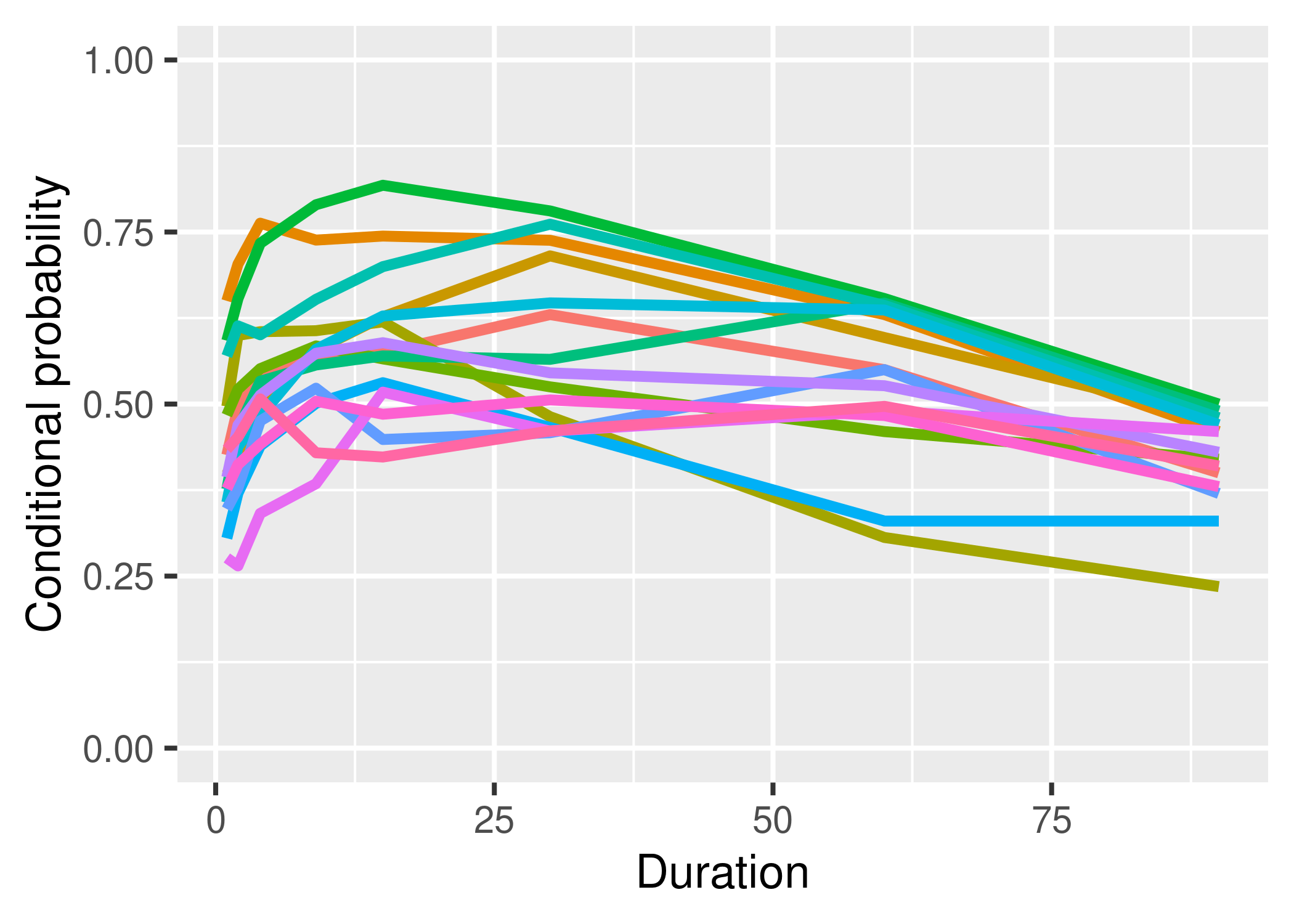}
\caption{FRAP}
%\label{fig:subim2}
\end{subfigure}
 }
 \caption{Conditional probabilities of vocalizations for the 15 different species at different time scales $\Delta t = \{1,2,4, 9, 15, 30, 60, 90\}$ obtained from fitted model for the MMPP (left) and samples from posterior predictive for the FRAP model (right).}
\label{fig:mmpp_vs_frap}
\end{figure}
\subsection{Model diagnostics}
We also carried out typical model diagnostics for count time series data discussed in \cite{czado2009predictive, kolassa2016evaluating}. Specifically, we use marginal calibration plots to assess model fit. We first draw samples from the predictive distribution of $Z_{180}|Z_{1}, \ldots, Z_{179}$ for a particular species of bird. We then compute $P(Z_{180} = 1\mid Z_1, \ldots, Z_{179})$ using the Monte Carlo average. This predictive probability is then matched with the observed probability $P(Z_{180} = 1)$ which is computed as $R^{-1}\sum_{r=1}^R Z^{(r)}_{180}$. In Figure \ref{fig:marginal_calibration_plot} we plot the differences in the predicted and observed probabilities for the 15 different species. For some birds the difference is very small whereas for other birds this difference goes up to 0.25, especially when the number of replicates available is small. Overall, the model performs adequately; prediction of vocalizations can potentially be improved by including covariates, such as weather and habitat conditions at the sampling site.  
\begin{figure}
    \centering
    \includegraphics[height=6cm, width = 8.5cm]{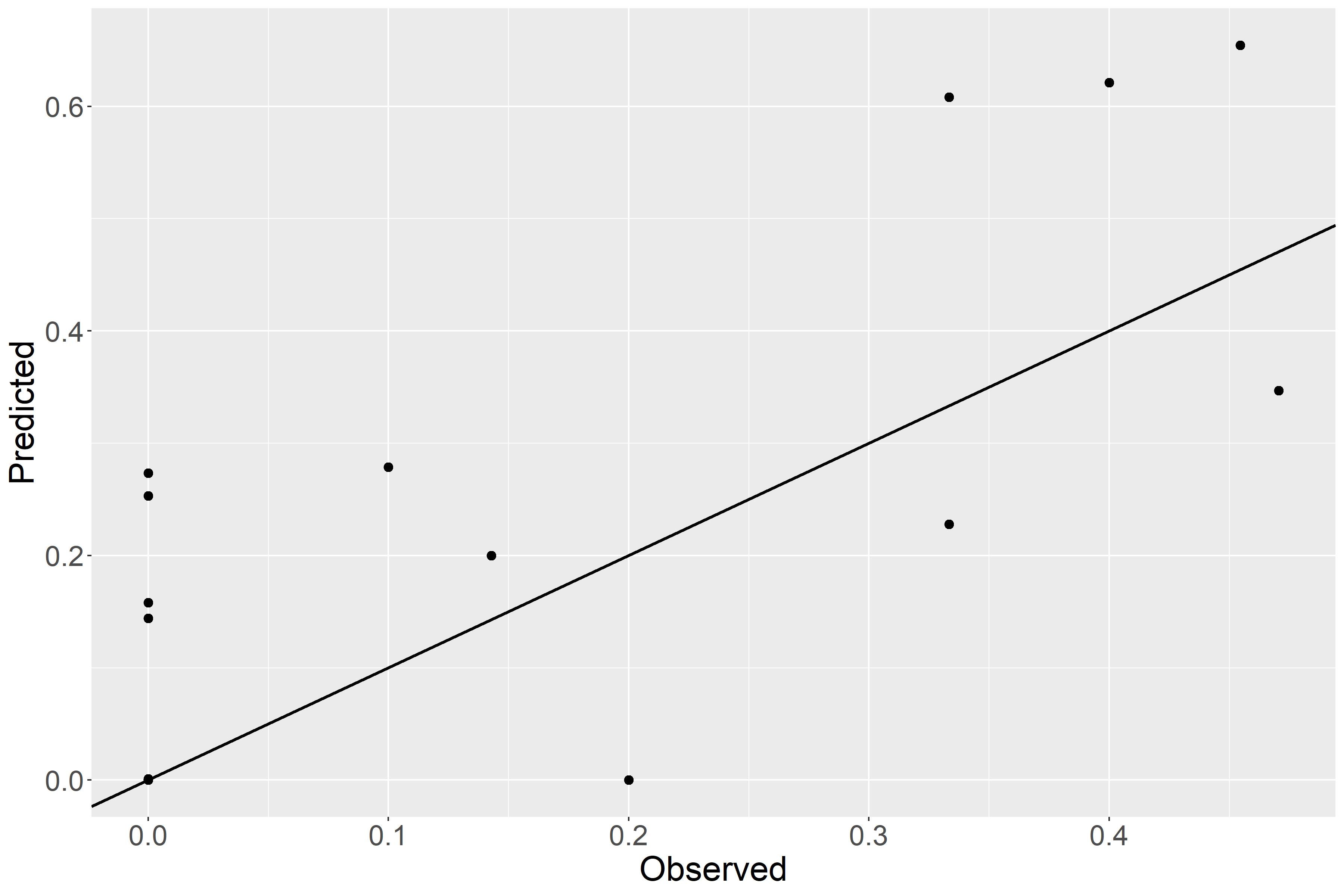}
    \caption{Difference between one-step ahead prediction probabilities for $Z_{180}=1$ and observed probabilities for the 15 species of bird analyzed here.}
    \label{fig:marginal_calibration_plot}
\end{figure}
\section{Discussion}
In this article, we proposed a novel class of models for characterizing long range dependence in discretized event data, along with a Bayesian approach to inference under these models.  We are particularly motivated by bird vocalization studies, and indeed are involved in ongoing collaborations collecting many such datasets across the globe in order to obtain new insights into biodiversity, interactions among species, and the role of biotic and abiotic factors.  The proposed class of FRAP models provide an important starting point for building realistic models for these emerging datasets as well as related datasets from precipitation and storm event modeling.  Immediate next directions are to add complexity to the models in order to more realistically characterize structure in the data, ranging from spatial dependence to covariate effects.  Such extensions are conceptually quite straightforward. 

There are several other important directions that are potentially less trivial.  The first is to broaden the class of models from a latent fractional Brownian motion to a broader class of stochastic processes with long range dependence.  This may include long range modifications to usual Gaussian process covariance kernels (eg, Matern), as well as non-Gaussian cases; e.g, Levy processes, alpha-stable processes, etc.  The second critical direction is developing much faster computational algorithms.  There is an immense literature on algorithms for accelerating computation in Gaussian process models, but to our knowledge very little consideration of the case in which there is long range dependence.  In our motivating applications, we are faced with immense datasets containing automated recordings over time at many different locations around the world.  To scale up to such datasets, we plan to consider divide-and-conquer algorithms and variational approximations, among other directions.

%\pagebreak
\appendix

\section{Appendix section}\label{app}
\subsection{Proof of Lemma \ref{lemma:conditional_scaling}}\label{sec:proof_of_lemma}
Since the fBM is a Gaussian process, from Corollary 2.6.3 of \cite{pipiras2017long} we get $\mathbb{E}\{B_H(i)\} = 0$ and $\mathbb{E}\{B_H(i)\}^2 = i^{2H}$ for any $i \in \mathbb{N}$. Hence, $B_H(i) \sim \Gauss(0, i^{2H})$. By stationarity of the incremental process of fBM it is enough to show \eqref{eq:scale_free_cond_probs} holds for $i = 1$. Define $Y_1^{(m)} = X_{{2^m}}$ and $Y_2^{(m)} = X_{{2^{m+1}}} - X_{{2^m}}$. Then $Y_1^{(m)} \sim \Gauss(0, 2^{2Hm})$. Also, $\mathbb{E}(Y_2^{(m)})^2 = \mathbb{E}(X_{2^{m+1}})^2 + \mathbb{E}(X_{2^{m}})^2 - 2 \mbox{Cov}(X_{2^{m+1}}, X_{2^{m}})$. From \eqref{eq:fBM_covariance} we get $\mbox{Cov}(X_{2^{m+1}}, X_{2^{m}}) = 2^{2H(m+1)}$. Thus, we have $Y_2^{(m)} \sim \Gauss(0, 2^{2Hm})$. Finally, $\mbox{Cov}(Y_1, Y_2) = \mbox{Cov}(X_{2^m}, X_{2^{m+1}}) - \mathbb{E}(X_{2^m})^2$ which after applying \eqref{eq:fBM_covariance} again we obtain $\mathrm{Cov}(Y_1, Y_2) = 2^{2Hm}(2^{2H - 1} - 1)$. Setting $\lambda^2 = 2^{2Hm}$, 
\begin{align*}
    P(Z_2^{(m)} & = 1 | Z_1^{(m)} = 1) = \dfrac{P(Y_1>0, Y_2> 0)}{P(Y_1>0)}\\
    & = 2P(Y_2/\lambda>0, Y_2/\lambda> 0)\\
    & = 2\left[\frac{1}{4} + \frac{1}{2\pi} \arcsin\left\lbrace\frac{1}{\lambda^2}\mathrm{Cov}(Y_1, Y_2)\right\rbrace\right]\\
    & = \frac{1}{2} + \frac{1}{\pi} \arcsin(2^{2H - 1} - 1).
\end{align*}
\subsection{Mixing of MCMC chain in Algorithm \ref{algo:MCMC_one_species}}
We briefly comment on the mixing of the MCMC chain obtained via Algorithm \ref{algo:MCMC_one_species}. With $L$ MCMC samples we calculate the effective sample sizes (ESS) for the parameters $f(\cdot)/\tau$ and $H$ as,
\begin{equation}
    \mbox{ESS} = \dfrac{L}{1+ 2\sum_{j=1}^J \rho(k)},
\end{equation}
where $\rho(j)$ is the autocorrelation at lag $j$. We set $J = 30$ as the maximum lag and $L = 10000$. For the 180 parameters $f(t)/\tau$, where $t = 1, \ldots, 180$, the average effective sample size for the 15 species were 2012.21 and that for the Hurst coefficient $H$ averaged over all the species is 1941.44.

\subsection{Proof of Theorem \ref{theorem:weak_consistency}}
For any set $V \in \Theta$ the posterior probability $\Pi(V \mid Z_1, \ldots, Z_n) = \int \Pi(V \mid W, Z_1, \ldots, Z_n) \Pi(W \mid Z_1, \ldots, Z_n) dW$. Now fix any weak neighborhood $U$ of $\theta_0$. Weak consistency conditional on the latent variables is proved in Section S6 of the supplementary document.  Thus the random variable $\Pi(U^c \mid W, Z_1, \ldots, Z_n)$ converges to 0 in $P_0-$probability. We now extend the proof for the marginal probability $\Pi(U^c \mid Z_1, \ldots, Z_n)$. Fix any $\delta > 0$. Then we have, 
\begin{align*}
   &E_{P_0}\Pi(U^C \mid Z_1, \ldots, Z_n)  = E_{P_0}\int \Pi(U^c \mid W, Z_1, \ldots, Z_n) \Pi(W \mid Z_1, \ldots, Z_n) dW \\
   & = E_{P_0}\underset{\Pi(U^c \mid W, Z_1, \ldots, Z_n) \leq \delta}{\int} \Pi(U^c \mid W, Z_1, \ldots, Z_n) \Pi(W \mid Z_1, \ldots, Z_n) dW \\
   & + E_{P_0}\underset{\Pi(U^c \mid W, Z_1, \ldots, Z_n) > \delta}{\int} \Pi(U^c \mid W, Z_1, \ldots, Z_n) \Pi(W \mid Z_1, \ldots, Z_n) dW\\
   & \leq \delta + P_0 \{\Pi(U^C \mid W, Z_1, \ldots, Z_n) > \delta \},
\end{align*}
where we use the fact that $\Pi(U^c \mid W, Z_1, \ldots, Z_n) \leq 1$.
\section*{Acknowledgements}
The authors acknowledge support from the United States Office of Naval Research (ONR) and the European Research Council (ERC).
%And this is an acknowledgements section with a heading that was produced by the
%$\backslash$section* command. Thank you all for helping me writing this
%\LaTeX\ sample file. See \ref{suppA} for the supplementary material example.

%\begin{supplement}
%\sname{Supplement A}\label{suppA}
%\stitle{Title of the Supplement A}
%\slink[url]{http://www.e-publications.org/ims/support/dowload/imsart-ims.zip%}
\section*{Supplementary material}
The supplementary document contains an extension of the FRAP framework to a grade-of-membership model, related priors and computational details for joint inference on multiple species, technical results for proving Theorem \ref{theorem:weak_consistency}, additional simulation results from Section \ref{sec:frac_prob_simulation}.
%\end{supplement}

\bibliography{sparse_reduced_rank_refs}
\bibliographystyle{apalike}

\clearpage\pagebreak\newpage
%\pagestyle{fancy}
%\fancyhf{}
%\rhead{\bfseries\thepage}
%\lhead{\bfseries SUPPLEMENTARY MATERIALS}

\baselineskip 20pt
\vspace{-0.5cm}
\begin{center}
{\LARGE{Supplementary materials \\ \bf{Bayesian semiparametric long memory models for discretized event data}}}
\end{center}

\setcounter{equation}{0}
\setcounter{page}{1}
\setcounter{table}{1}
\setcounter{figure}{0}
\setcounter{section}{0}
\numberwithin{table}{section}
\renewcommand{\theequation}{S.\arabic{equation}}
\renewcommand{\thesubsection}{S.\arabic{section}.\arabic{subsection}}
\renewcommand{\thesection}{S.\arabic{section}}
\renewcommand{\thepage}{S.\arabic{page}}
\renewcommand{\thetable}{S.\arabic{table}}
\renewcommand{\thefigure}{S.\arabic{figure}}
\baselineskip=18pt

\section{Hierarchical FRAP model}\label{sec:shared_trait}

The FRAP model \eqref{eq:trend_and_noise} in Section \ref{sec:FRAP} is designed to handle one bird species at a time. In this section we will develop an integrated model for dealing with multiple bird species having different series of event indicators $\mathbf{Z}^{(j)} = \{Z^{(j)}_1, Z^{(j)}_2, \ldots, Z^{(j)}_{R_j}\}, \, j  = 1,\ldots, m$, with $j$ indexing the species. For ease of exposition, we let $R_j = R, $ for $j = 1, 2, \ldots, m$. However, the general case of unequal number of replicates can be handled similarly. The main motivation for extending the model proposed in the main body of the paper is to share information across similar species of birds or locations; such sharing is particularly important given the sparsity of the data, with certain bird species vocalizing only a few times on average, see Section \ref{sec:primary_analysis}.
%For the bird vocalization data $j$ is the index corresponding to a particular bird species, whereas for rainfall data $j$ may denote the different locations at which the event data is recorded.  

Let $f_j(t)$ denote the latent non-stationarity  of the $j$-th bird species.   We assume that there are $K$ types of extremal behavioral profiles representing different patterns of vocalization behavior with time of the day.  Potentially, we can attempt to assign each species to one of these profiles, leading to a type of functional clustering; for related methods, refer to  \cite{chiou2007functional, jacques2014model, rodriguez2009bayesian} among others.  However, we view clustering as overly restrictive, and instead propose a functional mixed membership model \citep{manrique2010longitudinal}.  Let the $k$-th extremal profile by represented as 
 $h_k(t)$.   The behavioral trajectory $f_j(t)$ for species $j$ is represented by a combination of these extremal profiles having weights $\omega^{(j)}$,
%We assume for each of these groups there is an overall trend $h_k(t), j = 1, 2,\ldots,  K$ and each individual bird species shares traits of these groups with a varying degree of strength. Thus we have the following representation of $f_j(t)$ of the $j$-th bird in terms of the different groups,
\begin{equation}\label{eq:shared_trait_model}
    f_j(t) = \sum_{k=1}^K \omega^{(j)}_k h_k(t), \quad \omega^{(j)}\in \mathbf{\Delta}^{K-1}, 
\end{equation}
where $\mathbf{\Delta}^{K-1}$ is the $K-1$ dimensional probability simplex.

The vector $\omega^{(j)}$  describes the proportional membership of the $j$-th bird in each of the $K$ different groups.
For two different species $j$ and $j'$ having similar vocalization profiles, we expect the respective weight vectors $\omega^{(j)}$  and $\omega^{(j')}$ to be close. The representation of $f_j(t)$ in \eqref{eq:shared_trait_model} is also similar to a semiparametric latent factor model \citep{seeger04semiparametriclatent} where multiple functional data are represented by a linear combination of basis functions to model dependencies across different subjects; however, our model differs in constraining the factor loadings to be constrained to the probability simplex.
%Despite these similarities we would like to highlight two key distinctive feature of model \eqref{eq:shared_trait_model}. Firstly, factor loadings in a latent factor model are only identifiable upto orthogonal transformations but in model \eqref{eq:shared_trait_model} the weights and the corresponding functional values are identifiable due to the restriction $\omega^{(j)} \in \mathbf{\Delta}^{K-1}$ upto a permutations of $\{1,2,\ldots, K\}$. Moreover, in latent factor models the underlying factors $h_k(t)$ suffer from lack of interpretability. For the shared trait model \eqref{eq:shared_trait_model} each $h_k(t)$ represents the trend of each class. 
%\scalebox{0.5}{
\begin{figure}
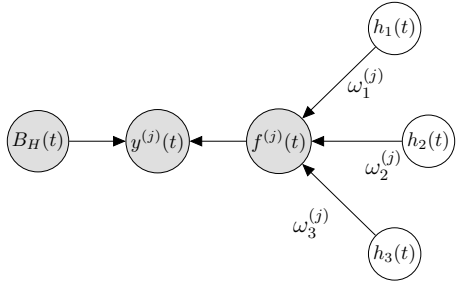

 \begin{center}
 \scalebox{0.75}{ 
\tikz{
  \node[obs](y) {$y^{(j)}(t)$};
  \node[obs, left = of y](fBM) {$B_{H}(t)$};
  \node[obs, right = of y](f_j) {$f^{(j)}(t)$};
  %\node[latent, right = of y, yshift = 1cm](theta_{1i}) {$\theta_{1i}$};
  %\node[latent, right = of y, yshift = -1cm](theta_{3i}) {$\theta_{3i}$};
  % Connect the nodes
  \node[latent, right = of f_j, yshift = 2cm](h_1) {$h_1(t)$};
  \node[latent, right = of f_j, xshift = 0.6cm](h_2) {$h_2(t)$};
  \node[latent, right = of f_j, yshift = -2cm](h_3) {$h_3(t)$};
  %\draw[-] (-2,0) -- (2,0)  (0,0.1) -- (0,-0.1)node[midway,bullet]{}
  %\node[obs, below = of y](Z_1) {$Z_{}$}
  \path[->] (h_1) edge  node[right=0.01cm ]{$\omega^{(j)}_1$} (f_j);
  \path[->] (h_2) edge  node[below right=0.01cm ]{$\omega^{(j)}_2$} (f_j);
  \path[->] (h_3) edge  node[below left=0.01cm ]{$\omega^{(j)}_3$} (f_j);
  \edge {f_j,fBM} {y} ; %
%\edge{h_1}{f_j}  ;
%\edge{h_2}{f_j};
%\edge{h_3}{f_j}
%\edge{Omega_{2}}{theta_{2i}} ;
%\edge{Omega_{3}}{theta_{3i}} ;
    
%\plate[inner sep = 0.5cm, xshift = 0.65cm,yshift = 0.35cm] {plate1} {
%(y)} {}%
}
}
 \end{center}
  \caption{Dependence structure of the latent process $y^{(j)}(t)$ under the decomposition \eqref{eq:shared_trait_model} for $K = 3$.}
  \label{fig:hier_frap_3}
\end{figure}
%}

%\FloatBarrier
We assume the same additive structure of the latent process $y^{(j)}_r(t)$ driving the bird vocalizations as in equation \eqref{eq:trend_and_noise} for the $j$-th species on the $r$th day, that is $y^{(j)}_r(t) = f^{(j)}(t) + B_{H,r}(t)$. The incurred dependency structure between the latent process and extremal class profiles is shown in Figure \ref{fig:hier_frap_3} when there are 3 extremal classes. The corresponding binary series $Z_r^{(j)}$ for the time interval $(t_{i-1}, t_i]$ has the following representation,
%Under this framework, the latent process for the $j$-th species on the $l$-th day we have the following representation for each set of binary time series sets $\mathbf{Z}^{(j)}$.
\begin{equation}\label{eq: shared_trait_and_noise}
    Z^{(j)}_r(t_{i-1}, t_i) = \begin{cases} 1 \mbox{ if } y^{(j)}_r(t_{i}) - y^{(j)}_r(t_{i-1}) = f^{(j)}(t_{i}) - f^{(j)}(t_{i-1})+ \epsilon^H_r>0\\
   0 \mbox{\quad otherwise,} \quad r = 1,\ldots, R,\, j = 1, \ldots, m,
   \end{cases}
\end{equation}
where $\epsilon^H_r \sim \Gauss(0, \tau^2\Sigma_H)$ is the realization of the fGN with Hurst exponent $H$ on the $r$-th day,  and the function $f^{(j)}(\cdot)$ satisfies the decomposition \eqref{eq:shared_trait_model}. The likelihood of the observed series $Z_r^{(j)} \in E_r^{(j)} \subset \{0,1\}^n$ can then be written as
\begin{align}\label{eq:shared_fractional_probit}
\nonumber
    P(Z_r^{(j)} \in E_r^{(j)}\mid W_r^{(j)}, \mathbf{f}^{(j)}, H)& = P(W_r^{(j)} \in E_{W_r^{(j)}} \mid \mathbf{f}^{(j)}, H),\\
     W_r^{(j)} & \sim \mathrm{N}(A\mathbf{f}^{(j)}, \tau^2\Sigma_H),
\end{align}
where $E_{W_r^{(j)}}$ is defined similarly as in equation \eqref{eq:frac_prob_with_priors} and $\mathbf{f}^{(j)} = \{f^{(j)}(t_0), \ldots, \\ f^{(j)}(t_n)\}$ for $r = 1, \ldots, R$ and $j = 1, \ldots, m$. The precision parameter $\tau^2$ is assumed to be equal for all realizations of the latent process across the species and days. We impose the restriction that the species vocalization profiles satisfy $f^{(j)}(0) = 0$, or equivalently $h_k(0) = 0$ for all $k = 1, \ldots, K$ to ensure identifiability. The matrix $A$ is defined as in Section \ref{sec:frac_prob}.  We assume that the Hurst coefficient is the same across the extremal classes based on the analysis in Section \ref{sec:amazon_frap}; see also Table \ref{tab:hurst_data}.

\section{Priors and posterior computation}
We assume $K$ is fixed and use model assessment diagnostics to choose a good value of $K$. Recall $\beta = \log\{H/(1-H)\}$. The unknown parameters in 
\eqref{eq:shared_fractional_probit} include $\Theta = \{(h_1, \ldots, h_k,\Omega, \beta, \tau): h_k \in \mathcal{F}, \, k =1, \ldots, K, \Omega \in \mathbf{\Delta}^{K-1 \times m}, \,\beta \in \Re, \tau > 0\}$, where $\mathbf{\Delta}^{K-1 \times m}$ is the space of all $K$ by $m$ matrices such that each column $\omega^{(j)}$ of the matrix $\Omega$ is in $\mathbf{\Delta}^{K-1}$, $\mathcal{F}$ is the space of continuously differentiable functions on $[0,T]$, for $k = 1, \ldots, K$. We place independent Gaussian process priors on the extremal profiles $h_k(t)$ with a squared exponential covariance kernel $\tau^2C_k(s,t)$ as in equation \eqref{eq:squared_exponential_kernel} with potentially different amplitude and length-scale parameters $(\sigma_k, \phi_k)$ scaled by the noise variance $\tau^2$. We add a small positive number $\nu$ to the diagonals of $C_k$ for numerical stability. Similar to Section \ref{sec:frac_prob} we augment the parameter space $\Theta$ with $(\sigma_k, \phi_k), \, k = 1, \ldots, K$, to obtain $\Theta_* = \Theta \times \eta_1 \times, \ldots, \eta_K$, where $\eta_k = \{(\log \sigma_k, \log \phi_k): \sigma_k, \phi_k >0\}$. We write $\Pi_{h_k}$ as the prior on $h_k(\cdot)$. The prior on $\beta$ is $\Pi_\beta$ as in Section \ref{sec:frac_prob}. Similarly, we use the prior $\Pi_{\eta_k}$ for the individual covariance kernel parameters $\eta_k$, where each $\Pi_{\eta_k} \equiv \Pi_\eta$ is as defined in Section \ref{sec:frac_prob}.

We also need a prior distribution for $ \Omega \in \mathbf{\Delta}^{K-1 \times m}$. Since each column of $\Omega \in \mathbf{\Delta}^{K-1}$, a natural choice is a Dirichlet prior. However, this  leads to non-conjugate posterior updates. To circumvent this problem, we assign independent truncated Gaussian priors $\Pi_{\omega^{(j)}}$ for $\omega^{(j)}$, i.e. $\omega^{(j)} \sim \Gauss(0, \lambda^2 \mathrm{I}_{K})\ind_{\Delta^{K-1}}$. In Section \ref{sec:sample_gaussian_on_simplex} we describe how to sample from a $K$- dimensional Gaussian distribution supported on a $K-1$-dimensional simplex. We let $\Pi_{\omega^{(j)}}$ denote the truncated Gaussian prior on the matrix $\omega^{(j)}, j= 1, \ldots, m$. The prior for $\tau$ is $\Pi_\tau$ from Section \ref{sec:frac_prob}.

The vector $\mathbf{f}^{(j)} = \{f^{(j)}(t_1), \ldots, f^{(j)}(t_n)\}$ can be written as $\mathbf{f}^{(j)} =  \mathbf{h}\omega^{(j)}$, where $\mathbf{h} =(\mathbf{h}_1, \ldots, \mathbf{h}_{K})$ with $\mathbf{h}_k = \{h_k(t_1), \ldots, h_k(t_n)\}^\T$ and $h_k(0) = 0, \mbox{for } k = 1, \ldots, K$. 
Given $m$ binary time series $\mathbf{Z}^{(j)} = \{Z^{(j)}_1,\ldots, Z^{(j)}_R\}, \, j = 1, \ldots, m$, we have the following representation of model \eqref{eq:shared_fractional_probit},
\begin{align} \nonumber 
    P(Z^{(j)}_r \in E^{(j)}_r\mid & \, W^{(j)}_r, \mathbf{f}^{(j)}, \Omega, \beta, \tau)  = P(W_r^{(j)} \in E_{W_r^{(j)}} \mid \mathbf{f}^{(j)}, \beta, \tau),\\ \nonumber
    & W_r^{(j)}\mid \beta, \lambda, \Omega, \mathbf{f}^{(j)}, \tau \sim \mathrm{N}(A\mathbf{f}^{(j)}, \tau^2\Sigma_H), \, \, r = 1, \ldots, R\\\nonumber
    & \mathbf{f}^{(j)} = \mathbf{h}\omega^{(j)}, \,\, \mathbf{h}_k\mid \eta_k, \tau^2 \sim \Pi_{h_k},\, \eta_k \sim \Pi_{\eta_k},\, k = 1, \ldots, K \\\nonumber
    & \omega^{(j)}\mid \lambda \sim  \Pi_{\omega^{(j)}},\, j = 1, \ldots, m\\ \nonumber
    & \beta \sim \Pi_\beta,\,\, \tau \sim \Pi_\tau\\ \label{eq:shared_frac_prob_with_priors}
\end{align}
Algorithm \ref{algo:MCMC_one_species} can be extended to carry out posterior analysis for the hierarchy \eqref{eq:shared_frac_prob_with_priors}. Details of our algorithm are provided below.

Let $C_{h_k}$ denote the prior covariance matrix for $\mathbf{h}_k$ under prior $\Pi_{h_k}$. Let $\mathbf{g}_k = A\mathbf{h}_k$. Then the induced covariance matrix for $\mathbf{g}_k$ is $C_{g_k} = \tau^2 A C_{h_k} A^\T + \nu \mathrm{I}_n$. The details of the MCMC implementation are provided below in Algorithm \ref{algo:MCMC_multiple_species}.
%\KwData{$\mathbbm{Z}=\{\mathbf{Z}^{(1)}, \ldots, \mathbf{Z}^{(m)}\}$, $M$ = number of MCMC samples}

% \KwResult{$L$ posterior samples from $\Pi(\Theta_* \mid \mathbbm{Z})$ : $ \left\lbrace\Hat{\Theta}_*^{(l)}\right\rbrace_{l=1}^L$}
% Initialize $\rho = 0$, $\Omega = \Omega_0$, $\mathbf{g}_1, \ldots, \mathbf{g}_K = 0$, \,  and $\eta_k = (0, 0),\mbox{for all} k$. Set $\Psi = \mathbf{g}\Omega$, where $\mathbf{g} = \{\mathbf{g}_1, \ldots, \mathbf{g}_K\}$ and let $ \Psi^{(j)}$ denote the j-th column of the matrix $\Psi$.
%\scalebox{0.85}{
\begin{algorithm}
Initialize $\rho = 0$, $\Omega = \Omega_0$, $\mathbf{g}_1, \ldots, \mathbf{g}_K = 0$, \,  and $\eta_k = (0, 0),\forall k$. Set $\Psi = \mathbf{g}\Omega$, where $\mathbf{g} = \{\mathbf{g}_1, \ldots, \mathbf{g}_K\}$ and let $ \Psi^{(j)}$ denote the $j$-th column of the matrix $\Psi$. For the $l$-th MCMC sample,
 \begin{itemize}[leftmargin=*]
      \item Update $W_r^{(j)}\mid - \overset{indp.}{\sim}\Gauss(\Psi^{(j)}, \tau^2\Sigma_H)\ind_{E_{W^{(j)}_r}}$, for $j =1, \ldots, m$ and $r = 1, \ldots, R$.
      
      \item Define $\overline{W} = (W^{(1)}, \ldots, W^{(m)})$, where $W^{(j)} = \frac{1}{R}\sum_{r=1}^R W_r^{(j)}$ and $\Psi_{-k} = \mathbf{g}_{-k} \Omega_{-k}$ \linebreak
      where $\mathbf{g}_{-k}$ is the matrix $\mathbf{g}$ without the $k$-th column and $\Omega_{-k}$ is the matrix $\Omega$ \linebreak
      without the $k$-th row. Set $W_{-k} = \overline{W} - \Psi_{-k}$ and $\delta_k = \Omega_k^\T \Omega_k$. Then $\mathbf{g}_k$ is updated \linebreak
      using $\mathbf{g}_k \mid - \sim \Gauss(\mu_k, \Phi_k)$, where $\Phi_k = \left(\frac{\delta_k R}{\tau^2} \Sigma_H^{-1} + \frac{1}{\tau^2}C_{\mathbf{g}_k}\right)^{-1}$ and 
      $\mu_k = \dfrac{R}{\tau^2}\Phi_k\Sigma_H^{-1}W_{-k}\Omega_k$  .
      %$g_k \sim \Gauss\left\lbrace\frac{R}{\tau^2}\Omega^{-1}\Sigma_H^{-1}\overline{W}, \Omega^{-1}\right\rbrace$, where $\overline{W} = \frac{1}{R}\sum_{r=1}^R W_r$ and $\Omega = \frac{R}{\tau^2}\Sigma_H^{-1} + K_\mathbf{g}^{-1}$.
     
      \item Define $P = \frac{R}{\tau^2}\mathbf{g}\Sigma_H^{-1}$ and $Q = P\mathbf{g}^\T$. Update $\Omega^{(j)}\mid -  \sim \Gauss(u_j, V_j)\ind_{\mathbf{\Delta}^{K-1}}$ where $V_j = (Q + \lambda^2 \mathrm{I}_k)^{-1}$ and $u_j = V_j P\overline{W}^{(j)}$, $\overline{W}^{(j)}$ being the $j$-th column of the \linebreak
      matrix $\overline{W}$ from the first step. 
      
      \item Update $\beta \mid -$ using Metropolis-Hastings with proposal density $\Gauss(\beta_{l-1}, s_1^2)$.
      
      \item For any $k = 1, \ldots, K$,  $\eta_k$ is updated jointly via random walk Metropolis-Hastings.
      
      \item Define $S = \frac{1}{2}\sum_{j=1}^{m} \left[\mbox{trace}\left\lbrace (W_j - G_j)^\T \Sigma_H^{-1} (W_j - G_j)\right\rbrace \right] + \frac{1}{2} \sum_{k = 1}^{K} \mathbf{g}_k^\T C_{g_k}^{-1} \mathbf{g}_k$ \linebreak
      where $W_j = (W^{(j)}_1, \ldots, W^{(j)}_R)$ and $G_j$ is a matrix with all columns equal to \linebreak
      $\mathbf{g}\omega^{(j)}$. Update $\tau \mid - \sim \mbox{Inverse-Gamma}\left\lbrace\frac{n(Rm + K)}{2} + a_\tau, S + b_\tau \right\rbrace$. 
  \end{itemize}
 
 \caption{Gibbs sampling algorithm to fit model \eqref{eq:shared_frac_prob_with_priors}.}
 \label{algo:MCMC_multiple_species}
\end{algorithm}

The computational bottleneck of Algorithm \ref{algo:MCMC_multiple_species} is in updating latent variables $W_r^{(j)}$. In addition, the algorithm is potentially subject to label switching; this is only a problem if there is interest in the membership matrix $\Omega$ and the corresponding extremal profiles.  As in other contexts in which label switching occurs, post-processing methods can be applied to relabel the MCMC output before inferences; see, for example \cite{stephens2000dealing}. In our experiments we did not encounter label switching and hence did not implement such approaches. 

The Metropolis-Hastings update for $\eta_k$ in Algorithm \ref{algo:MCMC_multiple_species} is similar to Algorithm \ref{algo:MCMC_one_species} and we consider individual proposal variances for each class $k$, which are adapted on the fly to maintain an overall acceptance probability of $\sim$ 0.3. Algorithm \ref{algo:MCMC_multiple_species} performed similarly in terms of mixing assessed by the effective sample size; see Appendix of the main document. When applied to the Amazon bird vocalization data, with 10000 MCMC samples and a maximum lag of 30, the effective sample size for $H$ is 1582.61. For the elements of the membership matrix $\Omega$, 1240.91 is the average effective sample size and that of $h_k(\cdot)/\tau$ is 962.83. 

To select the number of extremal classes, we use the Deviance Information Criterion (DIC) of \cite{spiegelhalter2002bayesian} adapted to our latent variable setting. Let $D(\theta) = -2 \sum_{j=1}^m \sum_{r=1}^R \log P(W_r^{(j)} \in E_{W_r^{(j)}} \mid \theta)$ where $\theta = (\mathbf{f}^{(j)}, \beta, \tau)^\T$. Then we define the DIC as,
\begin{equation}\label{eq:DIC_def}
    \mbox{DIC} = 2\bar{D} - D(\bar{\theta}),
\end{equation}
where $\bar{\theta}$ is the posterior mean of $\theta$. We fit model \eqref{eq:shared_frac_prob_with_priors} with different choices of $K$ and choose the value for which the DIC is minimized.

%\subsection{MCMC algorithm to fit model \eqref{eq:shared_trait_model} with hierarchy \eqref{eq:shared_frac_prob_with_priors}}\label{sec:MCMC_hierarchical_frap}

%}
%\vspace{0.2in}

%\subsection{Connection with other discretized models}
%The latent process $y(t)$ can be modified to suit the specifications of the MMPP model for discretized event data proposed in \cite{davison1996some}. Indeed, by restricting $y(t)$ to be piecewise positive constant within random sub-intervals of [0,T] and letting the changepoints be governed by a continuous time Markov chain, one recovers the MMPP. Within this formulation, the constants correspond to the logarithm of the Poisson rates of the MMPP. The more general class of doubly stochastic Poisson processes with a stochastic rate $\Lambda(t)$ is recovered by setting $y(t) = \log\{\Lambda(t)\}$, so that the increments on $y(t)$ within an interval $(t_1, t_2]$ are Poisson with rate $y(t_2) - y(t_1)$ and the resulting discretization forms a Bernoulli process. This requires $y(t)$ to be monotonically increasing. Gaussian process priors on constrained function spaces can be employed to carry out a full Bayesian analysis, see \cite{lin2014bayesian,maatouk2017gaussian, niu2018intrinsic}.

%We now describe a Markov chain Monte Carlo algorithm to efficiently sample from the posterior distribution $\Pi(\Theta_*\mid \mathbf{Z})$. We define $\mathbf{g}_k = A\mathbf{h}_k, k = 1, \ldots, K$ and $C_{\mathbf{g}_k} = AC_k A^\T$.

\section{Simulation experiments for the hierarchical FRAP model}\label{sec:hier_frap_simulations}
We considered two cases of simulation experiments implementing model \eqref{eq:shared_fractional_probit}. For both cases, we consider $m = 20$ species, $n = 90$ time units and $R = 20$ replications. We set $K = 3$ for these two scenarios. The $m$ species are first assigned random labels $k \in \{1,\ldots, 3\}$. Given species $j$ is assigned label $1$, we set the membership vector $\omega^{(j)} = (\omega^{(j)}_1, \omega^{(j)}_2, \omega^{(j)}_3)^\T$, where $\omega_3^{(j)} = \{1 + \exp(X_1) + \exp(X_2)\}^{-1}$, $\omega_1^{(j)} = \exp(X_1) \omega_3^{(j)}$, $\omega_2^{(j)} = \exp(X_2)\omega_3^{(j)}$ and $(X_1, X_2)$ follows a bivariate Gaussian distribution with mean vector $(3, 0)^\T$ and covariance matrix $\mathrm{I}_2$. If the assigned label is 2, then we repeat the same steps as above but generate $(X_1, X_2)^\T$ from a bivariate Gaussian with mean vector $(0, 3)^\T$ with the same covariance matrix. Finally, if the assigned label is 3, then we generate $(X_1, X_2)^\T$ with mean vector $(0,0)^\T$. The   
membership vectors are then combined with extremal profiles $h_1(\cdot)$, $h_2(\cdot)$ and $h_3(\cdot)$ to compute the individual species profile $f^{(j)}(\cdot)$. The two choices of the extremal profiles and the resulting individual profiles are listed below:
\begin{itemize}
    \item Case 1: $f^{(j)}(t) = \omega^{(j)}_1 f_1(t) + \omega^{(j)}_2 f_2(t) + \omega^{(j)}_3 f_3(t)$.
    \item Case 2: $f^{(j)}(t) = \omega^{(j)}_1 f_3(t) + \omega^{(j)}_2 f_4(t) + \omega^{(j)}_3 f_5(t)$.
\end{itemize}
Here we use the same definitions of the functions $f_1(t), \ldots, f_5(t)$ as in Section \ref{sec:frac_prob_simulation}. We fixed $\tau = 0.1$ and $H = 0.75$. Having generated the individual profiles, $R$ binary series of length $n$ are then generated following the steps from Section \ref{sec:frac_prob_simulation} for each species $j = 1, \ldots, m$.

We implemented hierarchy \eqref{eq:shared_frac_prob_with_priors} with $K= 2, 3, 4$ and selected the value leading to the minimal DIC; for both cases this yielded the true value of $K=3$.  In the first case, the ReMSEs  of the posterior mean of the three estimated extremal profiles are $0.24$, $0.05$ and $0.16$, respectively, for $f_1(\cdot)$, $f_2(\cdot)$ and $f_3(\cdot)$. The 95\% credible interval for the Hurst coefficient for this case is $[0.74, 0.85]$. For Case 2 the ReMSE values are $0.29$, $0.09$ and $0.06$ for $f_3(\cdot)$, $f_4(\cdot)$ and $f_5(\cdot)$, respectively. The 95\% credible interval for the Hurst coefficient in this case is $[0.74, 0.83]$. In Figure \ref{fig:hier_frap_results} we show the accuracy of the estimate of the weight matrix via a heat map.

\begin{figure}
\centering
\scalebox{0.95}{
\begin{subfigure}{0.5\textwidth}
\includegraphics[width=0.95\linewidth, height=3.5cm]{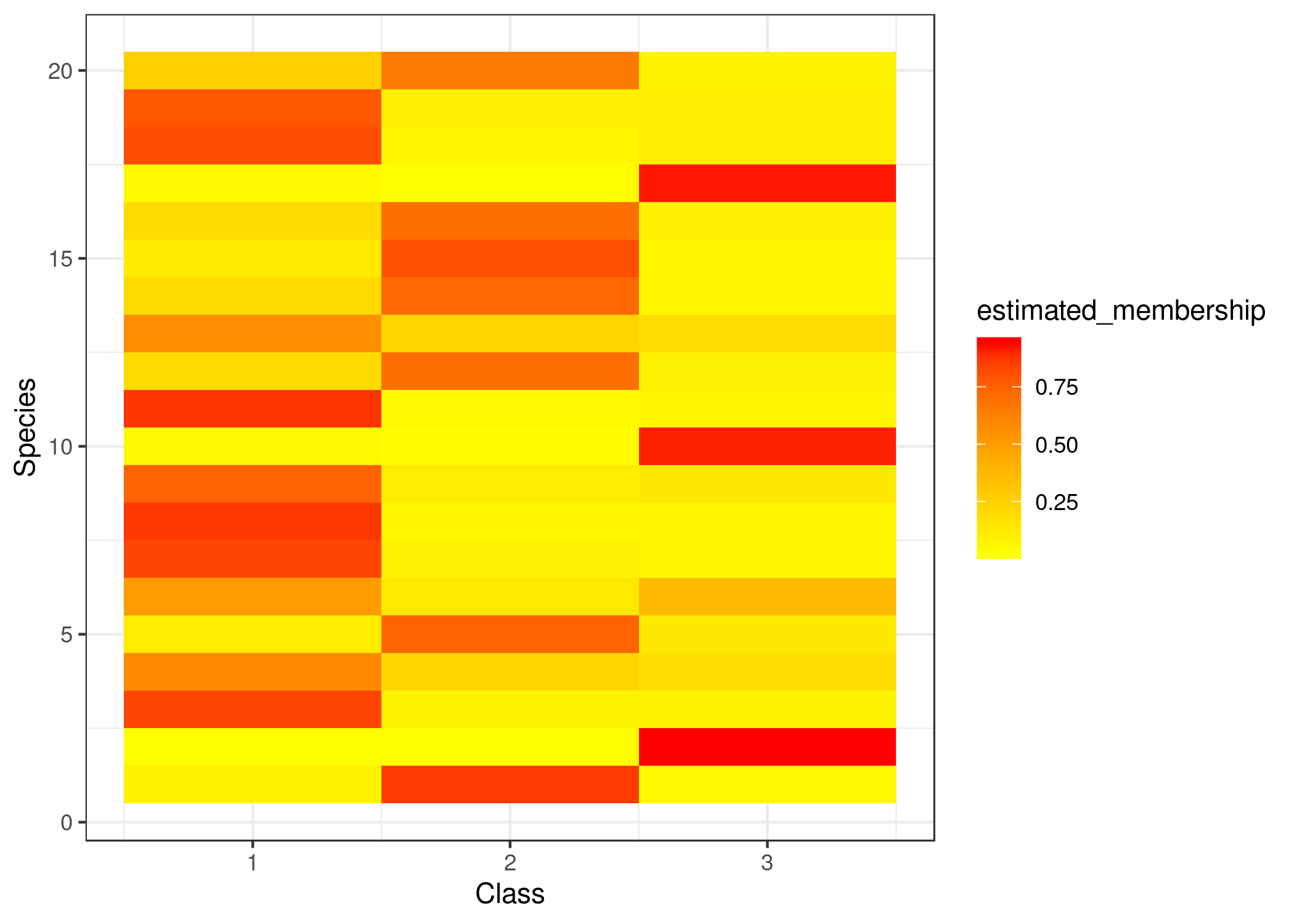} 
\caption{Estimated}
\label{fig:est_1}
\end{subfigure}
\begin{subfigure}{0.5\textwidth}
\includegraphics[width=0.85\linewidth, height=3.5cm]{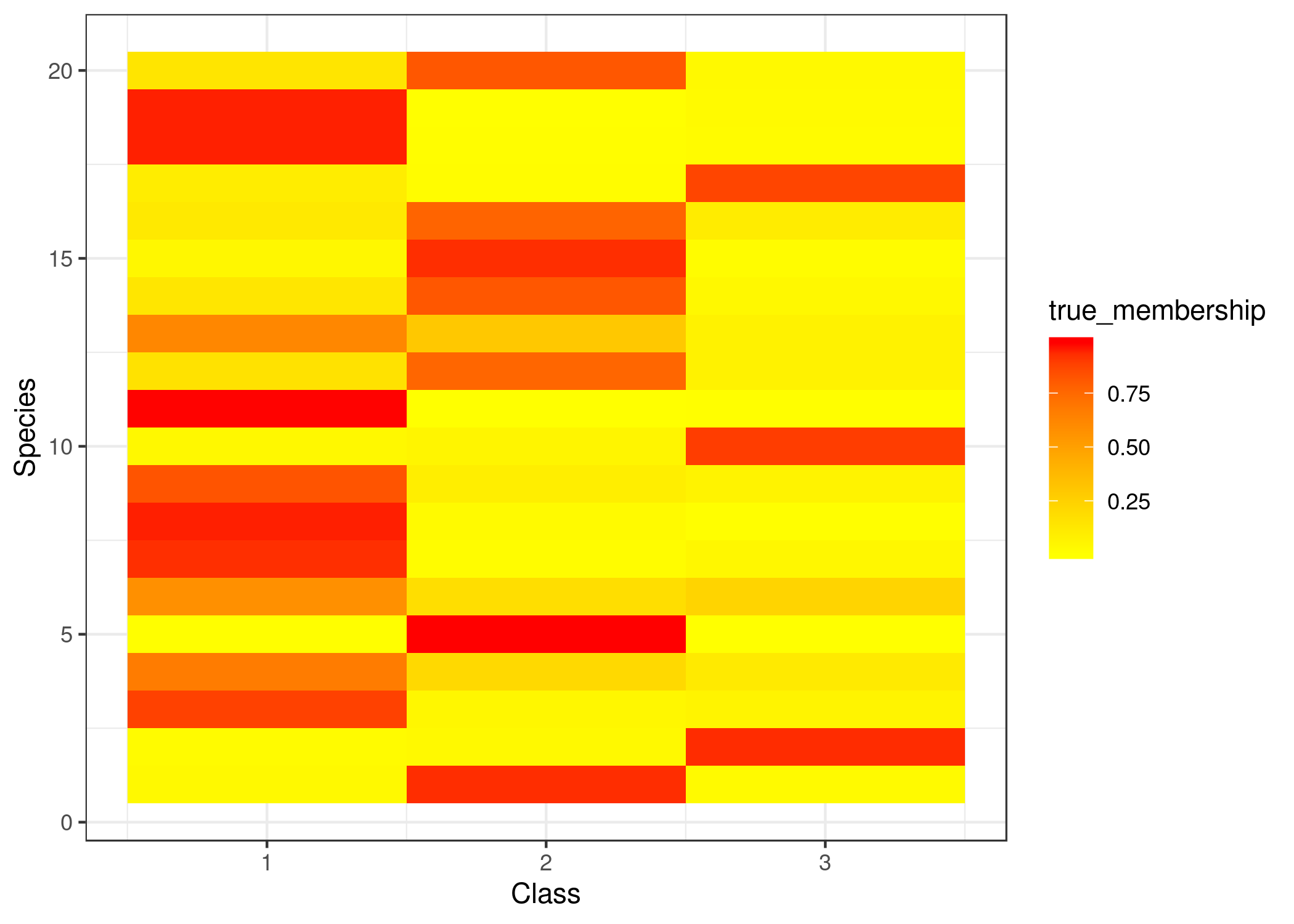}
\caption{True}
\label{fig:true_1}
\end{subfigure}
}
\scalebox{0.95}{ 
 \begin{subfigure}{0.5\textwidth}
\includegraphics[width=0.95\linewidth, height=3.5cm]{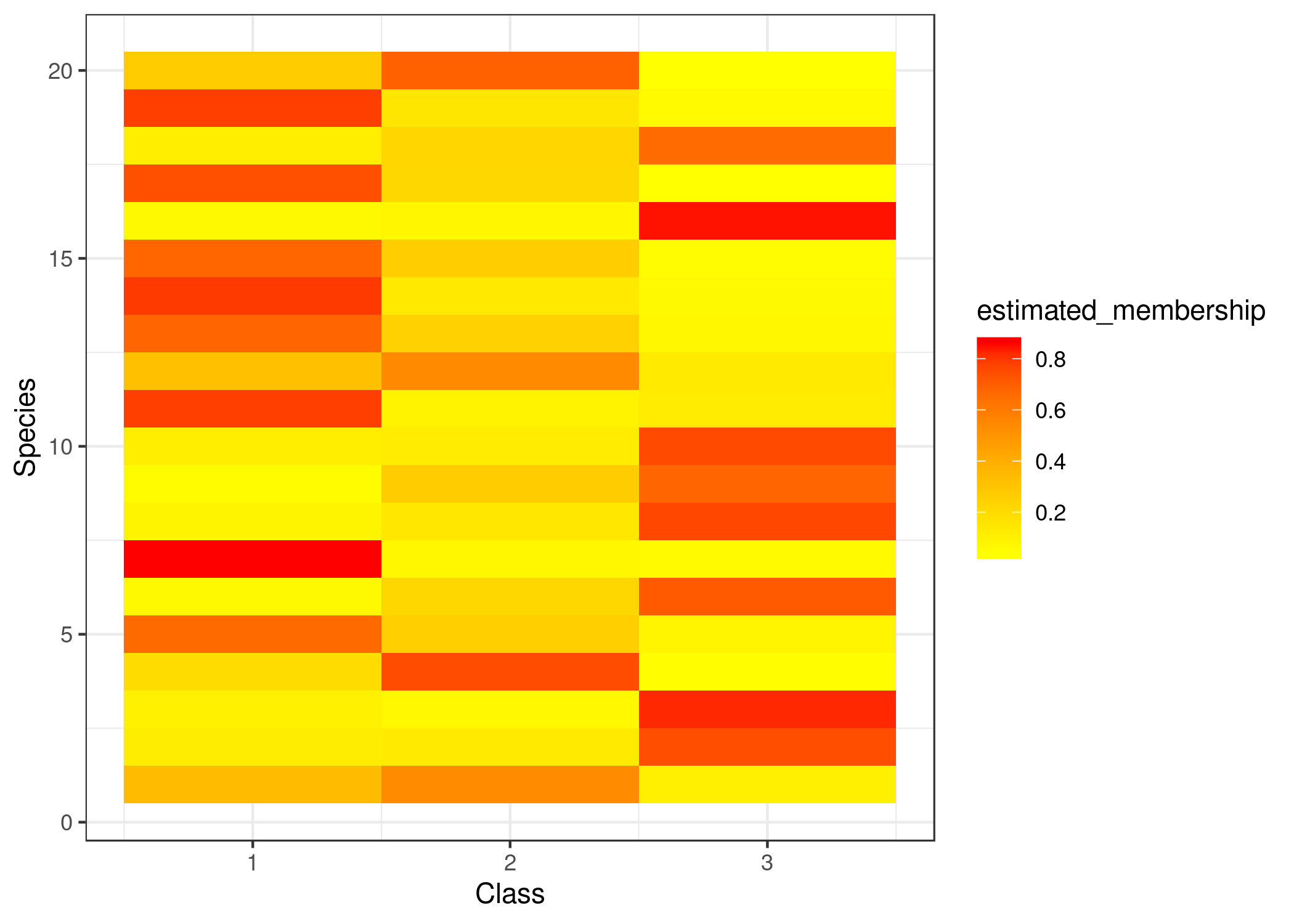} 
\caption{Estimated}
\label{fig:est_2}
\end{subfigure}
\begin{subfigure}{0.5\textwidth}
\includegraphics[width=0.9\linewidth, height=3.5cm]{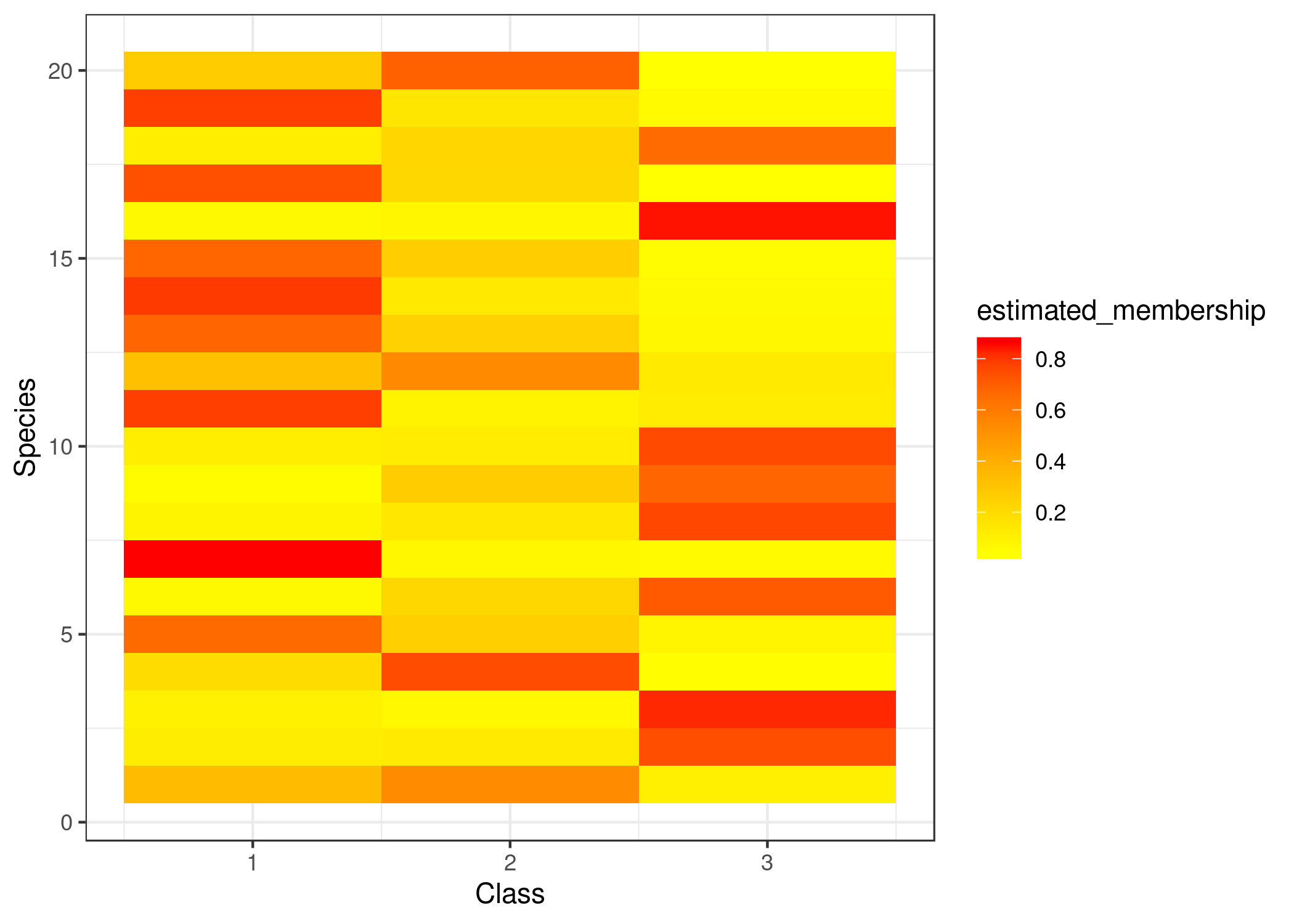}
\caption{True}
\label{fig:true_2}
\end{subfigure}
}
 \caption{Heat map of posterior mean of the membership matrix for the hierarchical FRAP model. Top row corresponds to Case 1 and bottom row corresponds to Case 2. Left panel shows the estimated membership matrix and the right panel shows the true memberships.}
\label{fig:hier_frap_results}
\end{figure}

%\subsection{Simulation experiments}

\section{Application to Amazon bird vocalization data}
We applied the hierarchical FRAP model to the Amazon bird vocalization data with $K = 2, 3$. The DIC values for $K = 2, 3$ are respectively 11365.4 and 10925.95. The 95\% credible interval for the Hurst coefficient obtained for $K = 3$ is $[0.78, 0.84]$. We show the posterior mean of the membership matrix in Figure \ref{fig:hier_frap_bird_data}. 

\begin{figure}
    \centering
    \scalebox{0.8}{
    \includegraphics[width=0.8\textwidth]{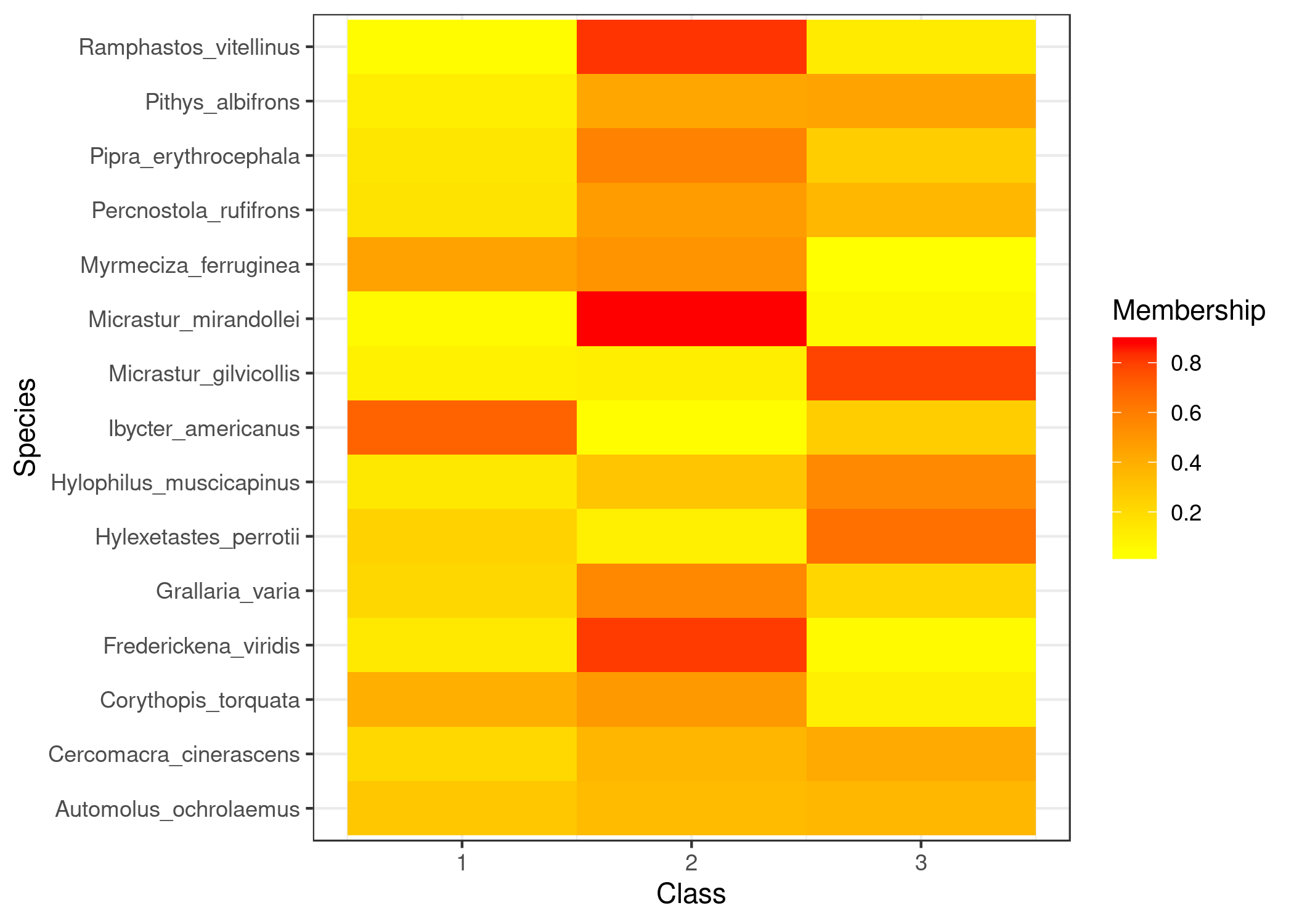}
    }
    \caption{Heat map of membership matrix obtained from the hierarchical FRAP model \eqref{eq:shared_trait_model} for the 15 bird species. }
    \label{fig:hier_frap_bird_data}
\end{figure}

The predominant class among the birds appears to be Class 2 in Figure \ref{fig:hier_frap_bird_data} where Frederickena viridis, Grallaria varia, Micrastur mirandollei, Percnostola rufifrons, Pipra erythrocephala,  Ramphastos vitellinus have maximum membership. This is evident from their spike in vocalization probabilities after the one and half hour mark from sunrise, see Figure \ref{fig:individual_trends}; such a trend is apparent for Cercomarca cinerascens also. Hylexetastes perrotii, Hylophilus muscicapinus, Micrastur gilvicollis show a high membership weight on Class 3 due to their higher vocalization probabilities at the onset of the recording. Ibycter americanus, however, identifies with Class 1. 

\section{Sampling a $K$ dimensional multivariate Gaussian truncated on $\mathbf{\Delta}^{K-1}$}\label{sec:sample_gaussian_on_simplex}
Let $X \sim \Gauss(\mu, \Sigma)$ where $\mu \in \Re^K$ and $\Sigma \in \Re^{K \times K}$. Suppose we want to sample $X$ restricted to the set $\Delta^{K-1}$, where $\Delta^{K-1}$ is the $K$ dimensional simplex. We first reparameterize $X$ as $X = (Y, 1- Y^\T \mathbf{1}_{K-1})$, where $\mathbf{1}_{K-1}$ is a $K -1$ dimensional vector of all 1's. Then $X$ can be rewritten as $X = \mathbf{J}Y + \alpha$, where $\mathbf{J} = (\mathrm{I}_{K-1}, -\mathbf{1}_{K-1})^\T$ and $\alpha = (\mathbf{0}_{K-1}^{\T}, 1)^{\T}$. 

Under this reparameterization, it is straightforward to observe that $Y \sim \Gauss_{K-1}(\mu_Y, \Sigma_Y)$, where $\Sigma_Y^{-1} = \mathbf{J}^\T \Sigma^{-1}\mathbf{J}$ and $\mu_Y = \Sigma_Y \mathbf{J}^\T \Sigma^{-1}(\mu - \alpha)$. Thus the problem of sampling $X$ on $\Delta^{K-1}$ can be reformulated as sampling $Y \sim \Gauss(\mu_Y, \Sigma_Y)$ subject to the constraint $DY + e \geq 0$, where $D = (\mathrm{I}_{K-1}, -\mathrm{I}_{K-1}, \\ -\mathbf{1}_{K-1})^\T$ and $e = (\mathbf{0}_{K-1}^\T, \mathbf{1}_{K-1}^\T, -1)^\T$. This is done via the exact Hamiltonian truncated Gaussian sampler proposed in \cite{pakman2014exact}.

\section{Proof of weak consistency}\label{sec:latent_variable_weak_consistency}
We first prove weak consistency assuming the latent variables are observed. Extension of the result to the setting when we only observe the binary indicators, that is model \eqref{eq:trend_and_noise}, is provided in the Appendix of the main document.

Consider the model $y_i = f(t_i) + u_i, \, i =0, 1,\ldots, n$, where the design points $t_0, \ldots t_n$ are fixed and belong to a compact interval $[0,T]$ and we assume that $f(t_0) = 0$. Without loss of generality we let $T = 1$. Suppose we observe the random variables $w_1, \ldots, w_n$ where $w_i = y_{i} - y_{i-1}, \, i= 1, \ldots, n$. Then the model for the observed random variables is $w_i =  A \mathbf{f} + \epsilon_i$ where the matrix $A \in \Re^{n \times n}$ with all elements zero except $A_{11} = 1$ and $A_{ii} = A_{i,i-1} = 1, \, i = 1, \ldots, n$. Here $\mathbf{f} = \{f(t_0), \ldots, f(t_n)\}^\T$. Suppose the error $\epsilon = (\epsilon_1, \ldots, \epsilon_n)^\T$ satisfies $\epsilon \sim \Gauss(0, \Sigma_H)$ where $0<H<1$. Given the data we want to recover the true parameters, say $(f_0(\cdot), H_0)$. Define $\mathbf{f}_0 = \{f_0(t_0), \ldots, f_0(t_n)\}^\T$. We additionally assume that $0<a <H < b< 1$ for some constants $0<a<b<1$ and $f_0$ is in the space of continuously differentiable functions on $[0,1]$. The maximum and minimum singular value of a matrix $P$ is written as $s_{\max}(P)$ and $s_{\min}(P)$. 
Let $p$, $p_0$ denote density of $w = (w_1, \ldots, w_n)^\T$ with respect to the Lebesgue measure under any generic $(f, H)$ and $(f_0, H_0)$ respectively. Since $H_0$ is bounded by assumption, so are $s_{\max}(\Sigma_{H_0})$ and $s_{\min}(\Sigma_{H_0})$.
For any fixed $\delta > 0$, we are interested in characterizing the set $\{p: \mathrm{KL}(p_0, p) < \delta\}$ where for two densities $q_1$ and $q_2$ $\mathrm{KL}(q_1, q_2) = \int \log (q_1/q_2) q_1$. For the densities $p_0$, $p$ we have,
\begin{equation}
\mathrm{KL}(p_0, p) = \frac{1}{2}\log \frac{|\Sigma_H|}{ |\Sigma_{H_0}|}+\frac{1}{2}\mathrm{tr}(\Sigma_H^{-1}\Sigma_{H_0} - \mathrm{I}_q) + \frac{1}{2} (A\mathbf{f} - A\mathbf{f_0})^\T\Sigma_H^{-1}(A\mathbf{f}-A \mathbf{f_0}).
\end{equation}
Unless otherwise specified, we shall write $\|x\|$ for the Euclidean norm of the vector $x$. For a matrix $P$, $\|P\|_2$ denotes the operator norm, that is, $\|P\|_2 = s_{\max}(P)$ . 
\begin{lemma}\label{lm:kl_setup}
Let $\Sigma_H, \Sigma_{H_0}$ be $n \times n$ covariance matrices with Hurst coefficient $H$ and $H_0$ respectively and $\delta \in (0,1)$. If $\| \Sigma_H - \Sigma_{H_0} \|_F \leq \delta$ and $\delta/s_{\mathrm{min}}(\Sigma_{H_0})<1/2$, then 
$$\mathrm{tr}(\Sigma_{H_0}\Sigma_H^{-1} - \mathrm{I}_q) - \log \mid \Sigma_{H_0} \Sigma_H^{-1} \mid \leq \frac{(K \log \rho) \delta^2}{s_{\mathrm{min}}^2(\Sigma_{H_0})},$$
where $K$ is some absolute positive constant and $\rho = 2 s_{\mathrm{max}}(\Sigma_{H_0})/s_{\mathrm{min}}(\Sigma_{H_0})$. Furthermore,
$$(A\mathbf{f} - A\mathbf{f}_0)^\T\Sigma_H^{-1}(A\mathbf{f}-A\mathbf{f}_0) \leq \{64/s_\mathrm{min}^2(\Sigma_{H_0})\}\| \mathbf{f} - \mathbf{f_0} \|^2$$
\end{lemma}
\begin{proof}
For the first claim see Lemma 1.3 in the supplementary document of \cite{pati2014posterior}. To prove the second claim, we use the inequality $\| P x \| \leq \| P \|_2 \| x \|$ to get $\|(Af-Af_0)^\T\Sigma_H^{-1}(Af-Af_0) \| \leq \|A\mathbf{f}- A\mathbf{f_0} \|^2 \| \Sigma_H^{-1} \|_2^2$. When $\| \Sigma_H - \Sigma_{H_0} \|_F \leq \delta$, the lemma from \cite{pati2014posterior} also provides a lower bound of $s_\mathrm{min}(\Sigma_H)$ as $s_\mathrm{min}(\Sigma_{H_0})/2$. Since $\| \Sigma_{H_0}^{-1} \|_2 = 1/s_\mathrm{min}(\Sigma_{H_0})$ and $\|A\mathbf{f} - A\mathbf{f}_0\| \leq \|A\|_2 \|\mathbf{f} - \mathbf{f}_0\|$, the result follows immediately as $\|A\|_2 = 4$. 
\end{proof} 

The covariance matrix $\Sigma_H$ is parameterized by the Hurst coefficient $H \in (0,1)$ where the $(i,j)$-th element of $\Sigma_H$ is $\Sigma_{ij, H} = \frac{1}{2}\{|k+1|^{2H} -2|k|^{2H} - |k-1|^{2H}\}$ where $k = i-j$. We now prove that when $H$ and $H_0$ are close, $\Sigma_H$ and $\Sigma_{H_0}$ are also close in the Frobenius sense.

\begin{proposition}\label{lm:H_setup}
Consider the $n \times n$ covariance matrices $\Sigma_H$ and $\Sigma_{H_0}$ for $H \in (0,1)$. Fix $\delta > 0$. If $|H - H_0|<\delta/n$ then $\|\Sigma_h - \Sigma_{H_0}\|_F < L\delta$ for some $L > 0$.  
\end{proposition}
\begin{proof}
The function $g(x) = a_1^x + a_2^x + a_3^x$ for fixed constants $a_1, a_2, a_3$ has bounded derivatives on $(0,1)$. Hence, by the mean value theorem $|g(x) - g(y)| \leq L |x-y|$ for some positive constant $L$. When $|H - H_0|<\delta$, we have, $\|\Sigma_H - \Sigma_{H_0}\|_F^2 = \sum_{i=1}^n \sum_{j=1}^n (\Sigma_{ij, H} - \Sigma_{ij,H_0})^2 \leq n^2L^2 (H - H_0)^2 \leq L^2\delta^2$.  
\end{proof}
In view of Lemma \ref{lm:kl_setup} and Proposition \ref{lm:H_setup} the set $\{\mathrm{KL}(p_0, p) < \epsilon\} \supset \mathcal{B} = \{(\mathbf{f}, H): \|\mathbf{f} - \mathbf{f}_0\|< \delta_1,\, |H - H_0|< \delta_2 \} $ for suitably chosen $\delta_1$ and $\delta_2$. Recall the joint prior $\Pi \equiv \Pi_f \times \Pi_\beta$ from the main document where $\beta = \log(H/(1-H))$. We then have that the prior probability $\Pi\{p: \mathrm{KL}(p,p_0) < \epsilon\} \geq \Pi_f\{f:\|\mathbf{f} - \mathbf{f}_0\|< \delta_1 \} \Pi_\beta \{\beta: |H - H_0|<\delta_2\}$. We trivially have that $$\Pi_\beta \{\beta: |H - H_0|<\delta_2\} > 0$$ because of the full support of the univariate Gaussian distribution on the real line. Combining the above fact with large support property of Gaussian process priors with squared exponential kernels on $\mathcal{F}$, the space of continuously differentiable functions, \citep[Theorem 4.2 - 4.4]{tokdar2007posterior} we have that,
$$\Pi\{p: \mathrm{KL}(p,p_0) < \epsilon\} \geq \Pi_f\{f:\|\mathbf{f} - \mathbf{f}_0\|< \delta_1 \} \Pi_\eta \{\eta: |H - H_0|<\delta_2\} > 0.$$
This proves that for any weak neighborhood $U$ of $p_0$, $\Pi(U^c\mid w_1, \ldots, w_n) \to 0$ in $P_0-$probability, where $P_0$ is the induced measure by $p_0$ \citep{ghosal2017fundamentals}.

\section{Simulation results and plots from main document}
Here we summarize the simulation results from Section 4 of the main document for $\tau = 0.05, 0.15$. Unlike the main document, here we report the error in estimating the marginal probability function. Recall that when the latent process $y(t)$ is formulated as $y(t) = f(t) + B_H(t)$, then the marginal probability of observing an event according to the proposed FRAP model during any arbitrary interval $(t_1, t_2]$ is $\Phi[\{f(t_2) - f(t_1)\}/\tau]$; for this interval we write this quantity as $m(t_1, t_2)$. We estimate this using $\hat{m}(t_1, t_2) = \Phi\{\tilde{f}(t_2) - \tilde{f}(t_1)\}$ where $\tilde{f}(\cdot)$ is the posterior mean of $\underline{f}(\cdot) = f(\cdot)/\tau$. Finally, the cumulative error is computed by summing squares of errors over all intervals and then taking the average, i.e. $e(m, \hat{m}) = n^{-1}\sum_{i=1}^{n} \{m(t_{i-1}, t_{i}) - \hat{m}(t_{i-1}, t_{i})\}$ where $t_0 = 0$. These errors (MSE) averaged over 30 independent replicates together with estimates of the Hurst coefficient $H$ is given in the following table.

\begin{center}
\begin{table}[!ht]
\caption{Mean square error (MSE) in estimating the marginal probability function for different choices of the latent trend function $f(t)$ for the model \eqref{eq:fractional_probit} under hierarchy \eqref{eq:frac_prob_with_priors}. For each $f(t)$ three values of the Hurst exponent are considered: $\{0.5, 0.75, 0.9\}$ together with $\{10, 25, 50\}$ replications. The results reported are averages of 30 independent simulation experiments for each combination.}
%\vspace{0.2in}
\huge
\centering
%\begin{flushleft}
\scalebox{0.39}{
\begin{tabular}{ccccccccccccc}\toprule
&&& \multicolumn{2}{c}{$f_1(t)$} & \multicolumn{2}{c}{$f_2(t)$} & \multicolumn{2}{c}{$f_3(t)$} & \multicolumn{2}{c}{$f_4(t)$} & \multicolumn{2}{c}{$f_5(t)$} \\
\cmidrule{1-13}
$\tau $& Hurst exponent ($H$) & Replications ($R$)& MSE & $\hat{H}$ & MSE & $\hat{H}$ & MSE & $\hat{H}$ & MSE & $\hat{H}$ & MSE & $\hat{H}$\\
\cmidrule{1-13}
\multirow{9}{*}{0.05}&\multirow{3}{*}{0.5} & 10 & 0.006 & 0.38 & 0.007 & 0.43 & 0.003 & 0.44 & 0.001 & 0.47 & 0.002 & 0.49  \\
 && 25 & 0.0009 & 0.41 & 0.001 & 0.49 & 0.0006 & 0.52 & 0.001 & 0.49 & 0.0008 & 0.52\\
 && 50 & 0.0006 & 0.48 & 0.001 & 0.49 & 0.0003 & 0.51 & 0.0009 & 0.50 & 0.0007 & 0.50\\
 \cmidrule{1-13}
 &\multirow{3}{*}{0.75} & 10 & 0.004 & 0.79 & 0.003 & 0.72& 0.009 & 0.78 & 0.005 & 0.76 & 0.004 & 0.74\\
  && 25 & 0.001 & 0.79 & 0.003 & 0.74 & 0.005 & 0.76 & 0.004 & 0.76 & 0.003 & 0.76\\
 && 50 & 0.0007 & 0.75 & 0.001 & 0.76  & 0.0008 & 0.75 & 0.002 & 0.75 & 0.001 & 0.74\\
 \cmidrule{1-13}
& \multirow{3}{*}{0.9} & 10 & 0.005 & 0.91 & 0.007 & 0.86 & 0.009 & 0.91 & 0.008 & 0.90 & 0.01 & 0.87\\
 & & 25 & 0.0009 & 0.88 & 0.005 & 0.92 & 0.006 & 0.89 & 0.006 & 0.88 & 0.008 & 0.92\\
 && 50 & 0.0007 & 0.89 & 0.002 & 0.89 & 0.002 & 0.88 & 0.003 & 0.89 & 0.002 & 0.88\\
 \cmidrule{1-13}
 \multirow{9}{*}{0.15}&\multirow{3}{*}{0.5} & 10 & 0.008 & 0.41 & 0.008 & 0.44 & 0.008 & 0.55 & 0.005 & 0.48 & 0.006 & 0.48  \\
 && 25 & 0.006 & 0.45 & 0.008 & 0.48 & 0.006 & 0.51 & 0.003 & 0.48 & 0.004 & 0.51\\
 && 50 & 0.001 & 0.49 & 0.003 & 0.50 & 0.004 & 0.51 & 0.0009 & 0.51 & 0.0009 & 0.48\\
 \cmidrule{1-13}
 &\multirow{3}{*}{0.75} & 10 & 0.008 & 0.77 & 0.007 & 0.73& 0.011 & 0.78 & 0.007 & 0.78 & 0.009 & 0.78\\
  && 25 & 0.001 & 0.79 & 0.006 & 0.75 & 0.008 & 0.77 & 0.004 & 0.73 & 0.008 & 0.75\\
 && 50 & 0.0007 & 0.75 & 0.003 & 0.75  & 0.006 & 0.74 & 0.003 & 0.74 & 0.004 & 0.74\\
 \cmidrule{1-13}
& \multirow{3}{*}{0.9} & 10 & 0.009 & 0.92 & 0.001 & 0.86 & 0.013 & 0.92 & 0.010 & 0.91 & 0.012 & 0.89\\
 & & 25 & 0.0009 & 0.88 & 0.008 & 0.90 & 0.009 & 0.88 & 0.008 & 0.87 & 0.009 & 0.91\\
 && 50 & 0.0007 & 0.89 & 0.006 & 0.88 & 0.007 & 0.90 & 0.005 & 0.91 & 0.007 & 0.92\\
 \hline
 \end{tabular}
}
%\end{flushleft}
\label{tab:mse}
\end{table}
\end{center}
\begin{figure}
\centering
\begin{subfigure}{0.45\textwidth}
\centering
\includegraphics[width=0.95\linewidth, height=4.5cm]{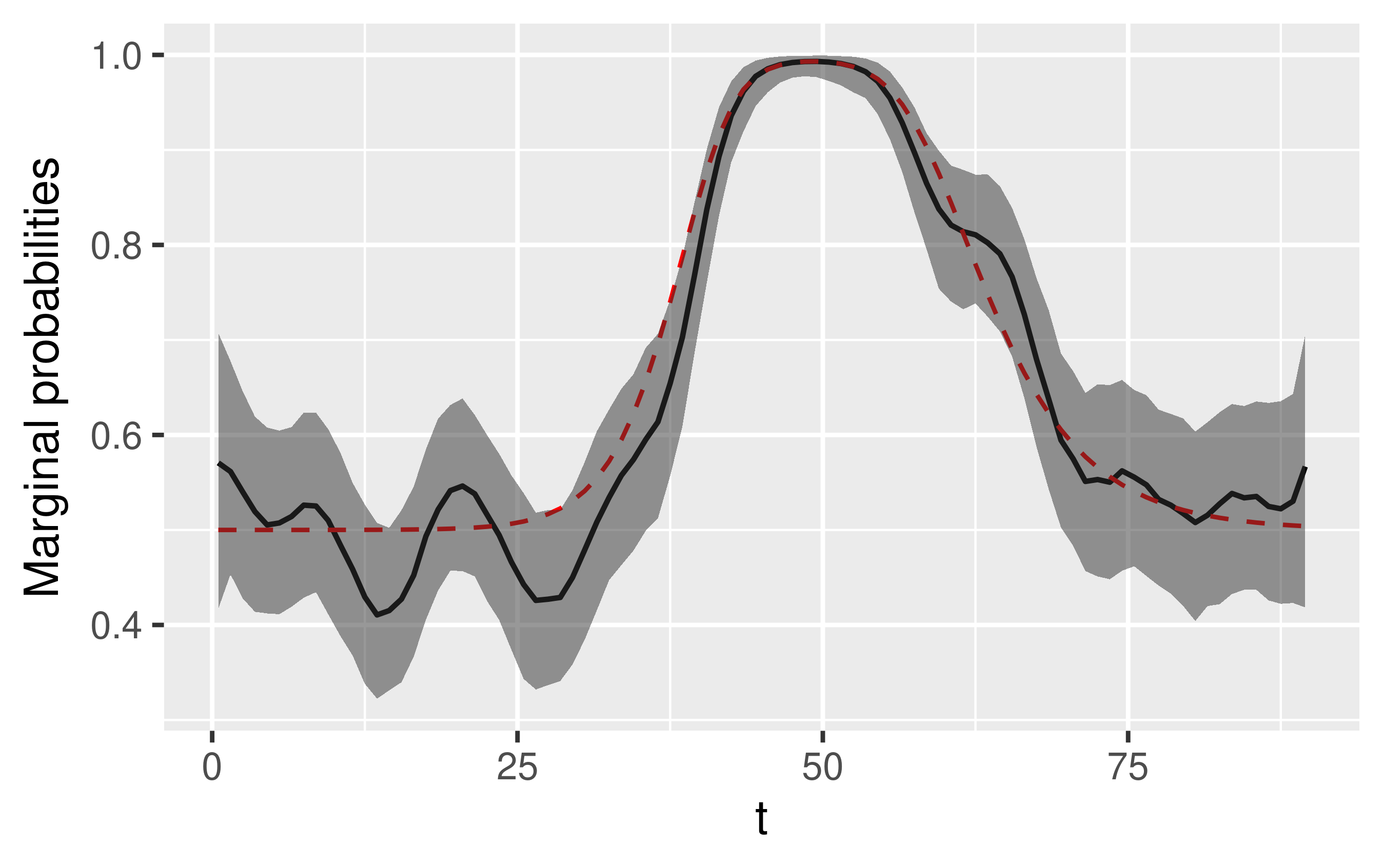}
%\caption{$f_2(t)=5[1 + \exp\{-2.5(t - 45)/15\}]^{-1}$}
\caption{$f_2(t)$}
\end{subfigure}
%\scalebox{0.75}{ 
\begin{subfigure}{0.45\textwidth}
\centering
\includegraphics[width=0.95\linewidth, height=4.5cm]{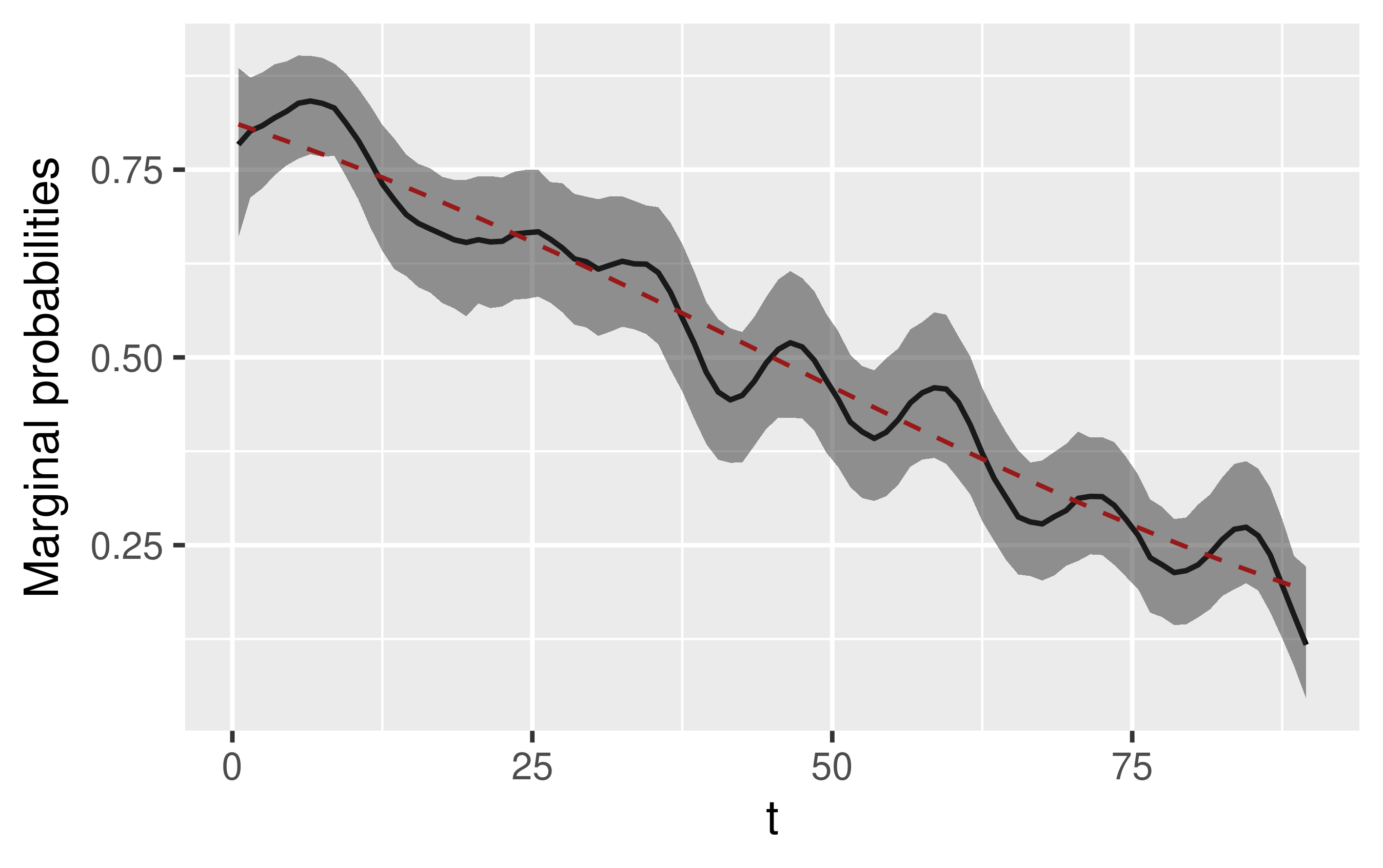}
%\caption{$f_3(t) = -2\{(t-45)/45\}^2 + 2$}
\caption{$f_3(t)$}
\end{subfigure}
\begin{subfigure}{0.45\textwidth}
\centering
\includegraphics[width=0.95\linewidth, height=4.5cm]{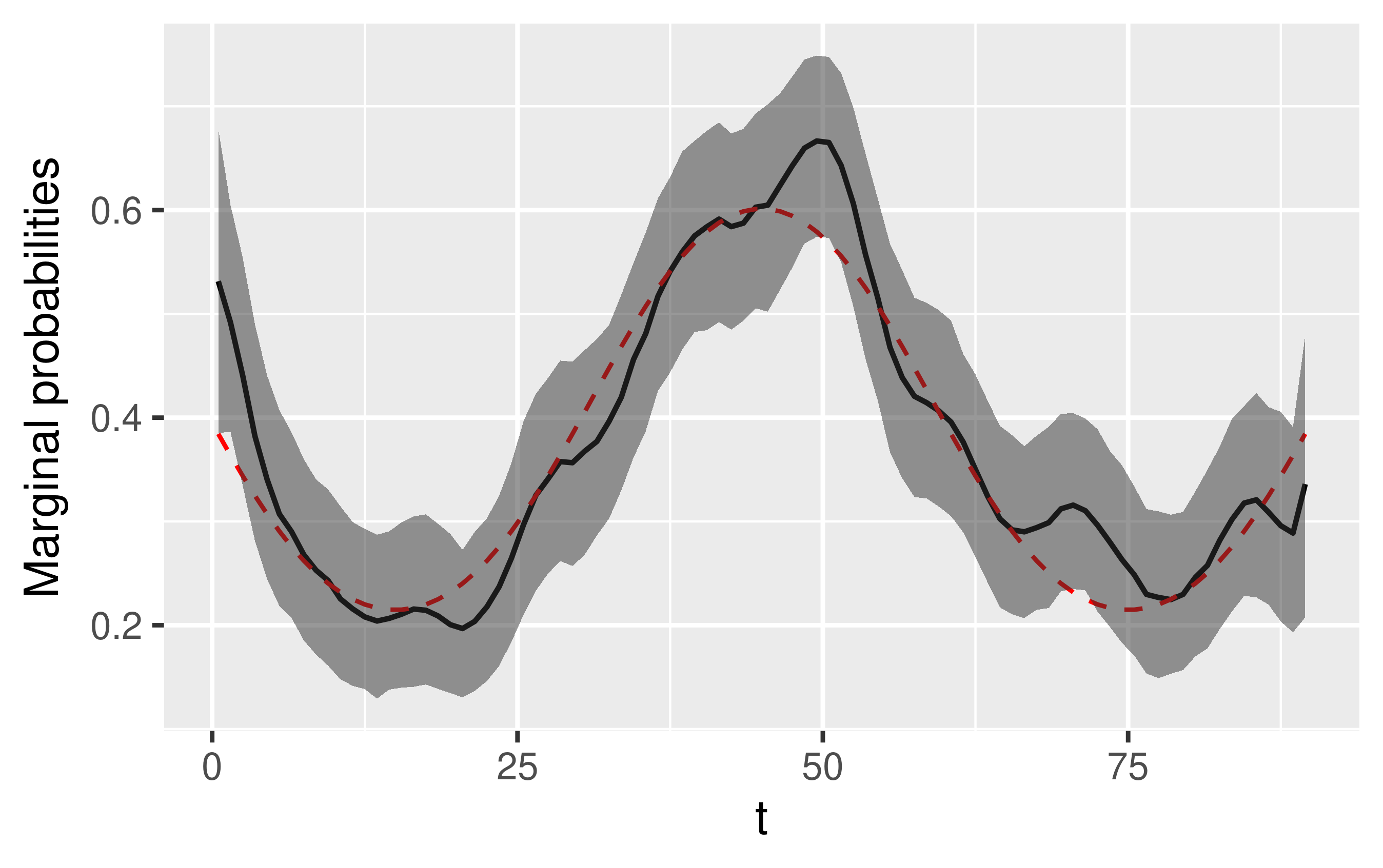}
%\caption{$f_4(t) = -1.2\frac{(t-45)}{45} - 0.5\cos \frac{3\pi t}{90} - 1.7$}
\caption{$f_4(t)$}
\end{subfigure}
%}
%\begin{center}
%\scalebox{0.75}{ 
\begin{subfigure}{0.45\textwidth}
\centering
\includegraphics[width=0.95\linewidth, height=4.5cm]{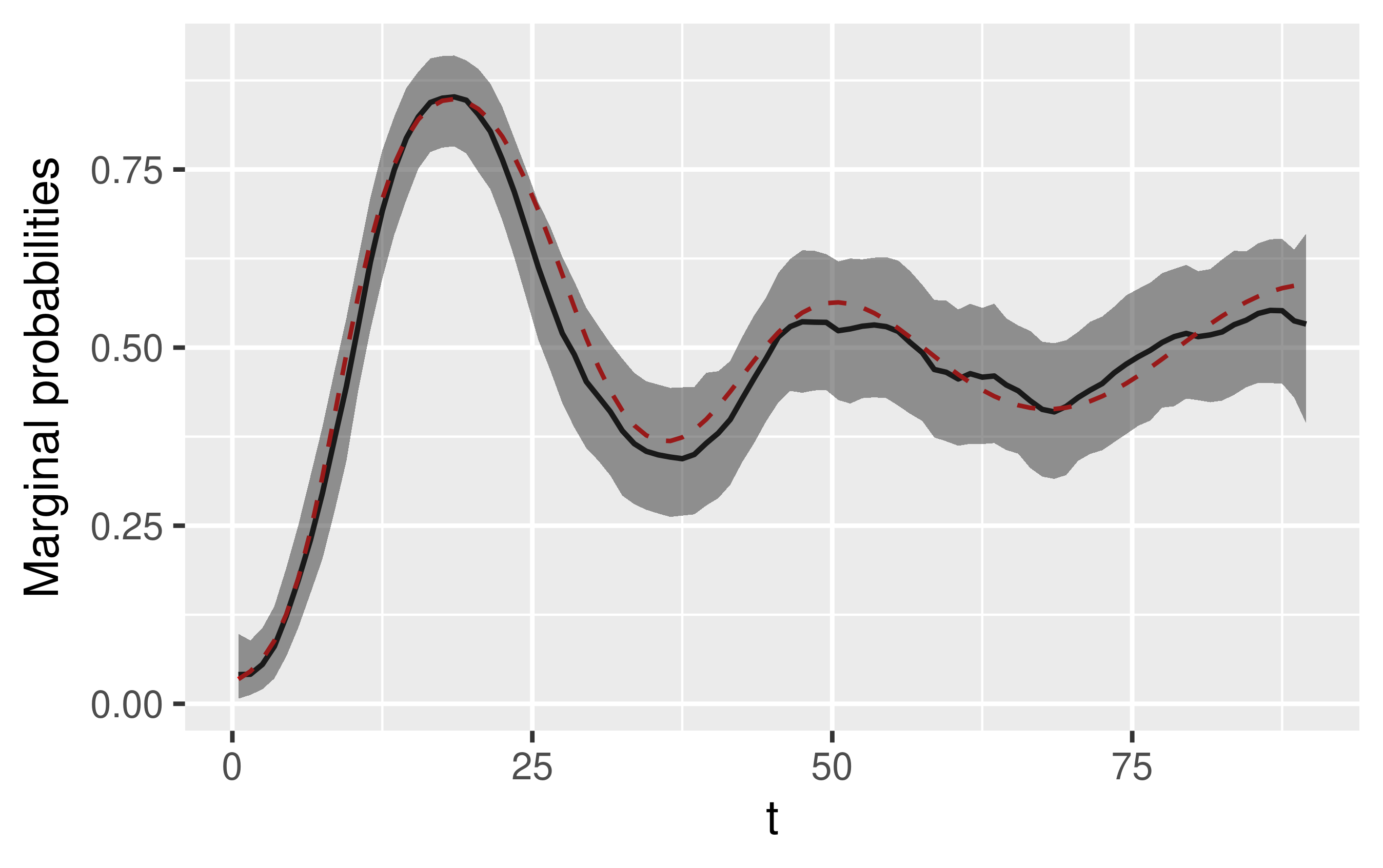}
%\caption{$f_5(t) = 0.1f_1(t)\log f_2(t)$}
\caption{$f_5(t)$}
\end{subfigure}
%\hspace{0.5in}

\caption{Figures (a)-(d) show the posterior mean and 95\% credible bands for marginal probabilities in one minute intervals for $f(t) = f_2(t), f_3(t), f_4(t), f_5(t)$ in Section 4 of main document. The values of the Hurst coefficient and the number of replications were $H = 0.75$ and $R = 50$, respectively. Red dashed and black solid lines correspond to the true values and the posterior mean respectively. Gray shaded regions are 95\% credible bands.}
\end{figure}

\end{document}